\documentclass[a4paper,10pt,openright]{article} % Per avere margini destro e sinistro diversi
\usepackage[english]{babel} % per lingua. TeXnic dà warning per misteriosi motivi...
\usepackage[utf8]{inputenc} 
\usepackage{amsmath,amssymb, stmaryrd,amsthm,amsfonts,amscd,color,enumerate,eucal,latexsym,mathrsfs,mathtools,cancel,tikz}
\usepackage{epsfig, hieroglf, protosem}  % ciao
\usepackage{simplewick, bm}
\usepackage{verbatim}
\usepackage{hyperref}
\usepackage[all,cmtip]{xy}
\usetikzlibrary{matrix,calc}

\numberwithin{equation}{section}

\definecolor{LorColor}{RGB}{230, 74, 129}
\definecolor{PaColor}{RGB}{56,174,199}
\definecolor{NiColor}{RGB}{77,77,255}
\definecolor{ClaColor}{RGB}{0,0,255}
\definecolor{NiColoRed}{RGB}{255,77,77}
\definecolor{NiCitation}{RGB}{0,181,26}

\newtheoremstyle{TheoremStyle}% <name>
{3pt}% <Space above>
{3pt}% <Space below>
{\slshape}% <Body font>
{}% <Indent amount>
{\bf}%{\itshape}% <Theorem head font>
{:}% <Punctuation after theorem head>
{.5em}% <Space after theorem head>
{}% <Theorem head spec (can be left empty, meaning 'normal')>

\theoremstyle{TheoremStyle}
\newtheorem{theorem}{Theorem}[section]
\newtheorem{corollary}[theorem]{Corollary}
\newtheorem{proposition}[theorem]{Proposition}
\newtheorem{lemma}[theorem]{Lemma}
\newtheorem{definition}[theorem]{Definition}
\newtheorem{remark}[theorem]{Remark}
\newtheorem{example}[theorem]{Example} %[theorem]{Example}

\newcommand{\rombo}{\begin{tikzpicture}[thick,scale=1.2]
\draw (0,0) -- (0,0.3);
\draw (0,0.3) -- (.15,0.15);
\draw (0,0.3) -- (-.15,0.15);
\draw (0,0) -- (.15,0.15);
\draw (0,0) -- (-.15,0.15);
\end{tikzpicture}}
\newcommand{\fish}{\begin{tikzpicture}[thick,scale=1.2]
\draw (0,0) edge [out=30,in=-30] node[above] {} (0,.25);
\draw (0,0) edge [out=150,in=210] node[above] {} (0,.25);
\end{tikzpicture}}
\newcommand{\fiammifero}{\begin{tikzpicture}[thick,scale=1.2]
\draw (0,0) -- (0,0.3);
\filldraw (0,0.3)circle (1pt);
\end{tikzpicture}}
\newcommand{\propagatore}{\begin{tikzpicture}[thick,scale=1.2]
\draw (0,0) -- (0,0.3);
\end{tikzpicture}}

\title{A Microlocal Approach to Renormalization in Stochastic PDEs}

\author{
	Claudio Dappiaggi\thanks{CD:
		Dipartimento di Fisica,
		Universit\`a degli Studi di Pavia \& INFN, Sezione di Pavia, 
		Via Bassi 6,
		I-27100 Pavia,
		Italia;
		\mbox{claudio.dappiaggi@unipv.it}
	}
	\and
	Nicol\`o Drago \thanks{ND: Department of Mathematics, 
		Julius Maximilian University of W\"urzburg,
		Emil-Fischer-Stra\ss e 31,
		D-97074 W\"urzburg,
		Germany;
		\mbox{nicolo.drago@mathematik.uni-wuerzburg.de}
}
\and
Paolo Rinaldi \thanks{PR: Dipartimento di Fisica,
	Universit\`a degli Studi di Pavia \& INFN, Sezione di Pavia, 
	Via Bassi 6,
	I-27100 Pavia,
	Italia;
	\mbox{paolo.rinaldi01@universitadipavia.it}
}
\and
Lorenzo Zambotti \thanks{LZ: Laboratoire de Probabilit\'es, Statistique et Mod\'elisation, Sorbonne Universit\'e, Universit\'e de Paris, CNRS, 4 Place Jussieu, 75005 Paris, France; 
\mbox{zambotti@lpsm.paris}}
}

\date{\today}

\begin{document}
\maketitle
\begin{abstract}
	We present a novel framework for the study of a large class of non-linear stochastic PDEs, which is inspired by the algebraic approach to quantum field theory. The main merit is that, by realizing random fields within a suitable algebra of functional-valued distributions, we are able to use techniques proper of microlocal analysis which allow us to discuss renormalization and its associated freedom without resorting to any regularization scheme and to the subtraction of infinities. As an example of the effectiveness of the approach we apply it to the perturbative analysis of the stochastic $\Phi^3_d$ model.
\end{abstract}


\tableofcontents

\section{Introduction}\label{Sec: introduction}

Stochastic partial differential equations (SPDEs) play a distinguished role in modeling different phenomena, especially in physics. Notable examples range from turbulence, interface dynamics, {\it e.g.} the KPZ equation \cite{Kardar:1986xt}, and quantum field theory, {\it e.g.} stochastic quantization \cite{PW81}. From a mathematical viewpoint SPDEs have been thoroughly studied and a well-established framework is nowadays available -- see for example the recent reviews \cite{Chandra, Hairer09}. While linear SPDEs are well-understood, the non-linear ones are still under close scrutiny. Most often a non-linear behaviour is analyzed within a perturbative scheme, although Hairer's theory of regularity structures has provided a non-perturbative framework to study existence of solutions for a large class of equations \cite{Hairer13, Hairer14,Hairer15}, see also \cite{Da Prato} and the paracontrolled distribution approach of Gubinelli et al. \cite{GIP}. Alas, from these techniques one cannot extract a constructive method to build the expectation value and the correlations of the non perturbative solution.

In addition all approaches to non-linear SPDEs have to cope with a common hurdle, namely the necessity of applying a renormalization scheme. Heuristically speaking this is due to the presence of pathological singular behaviors. These can be ascribed to the necessity of considering suitable powers of the fundamental solution of the linear part of the underlying SPDE and such powers do not identify well-defined distributions. Renormalization is thus nothing but the generic name under which one recollects all different ways to give a coherent mathematical meaning in the framework of the theory of distributions to the above mentioned singular structures. 

Without entering at this stage into the details of different renormalization schemes, of their relations and of the ambiguities which they bring forth, we wish to highlight two key aspects.
The first is the most common way of dealing with such singularities in the theory of SPDEs, {\it e.g} \cite{Bruned, Hairer14}.
To be more concrete, we outline it via a prototypical example of what can occur.
Suppose we consider on $\mathbb{R}^d$, $d\geq 2$, the Laplace-Beltrami operator $\Delta=\sum\limits_{i=1}^d\frac{\partial^2}{\partial x_i^2}$ and its fundamental solution $G\in\mathcal{D}^\prime(\mathbb{R}^d\times\mathbb{R}^d)$.
In the analysis of non-linear stochastic PDEs whose linear contribution is given by $-\Delta$, it becomes necessary to consider powers of $G$, which are in general ill-defined.
To bypass this hurdle, one considers a one-parameter family of smooth integral kernels, whose associated distributions $G_\varepsilon$ converges weakly to $G$.
At this stage all powers of $G^n_\varepsilon$ are well-defined but their limit as $\varepsilon\to 0^+$ is no longer meaningful.
This can be ascribed to the presence of terms proportional either to the logarithm or to suitable inverse powers of the parameter $\varepsilon$.
A common way of dealing with such pathologies is to define the sought powers of $G^n$ by removing the singular contribution from $G^n_\varepsilon$ taking subsequently the limit as $\varepsilon$ tends to $0$.
This gives rise to well-defined distributions which are usually referred to as the {\em renormalized (or regularized)} $G^n_{\mathrm{ren}}\in\mathcal{D}^\prime(\mathbb{R}^d\times\mathbb{R}^d)$.
Such procedure, which is often pictorially indicated as the subtraction of infinities, is manifestly non-unique and leads therefore to introducing ambiguities known as renormalization constants.

The second aspect that we wish to highlight is that this kind of difficulties and problems is by far not unique to SPDEs, but it is a common feature in many other subjects. For example, within the framework of theoretical physics, the same exact difficulties occur in quantum field theory when dealing with interacting models at a perturbative level. While the renormalization scheme heuristically outlined above is usually adopted also in these cases, it presents several difficulties and limitations in its use when the underlying background is no longer $\mathbb{R}^n$ but it is replaced by a smooth manifold $M$ endowed with a Lorentzian or a Riemannian metric $g$, depending on the specific case under analysis. A very successful, mathematically rigorous approach to tackle these problems exists and goes under the name of {\em algebraic quantum field theory}, see \cite{BFDY15} for a recent review. 

In these scenarios, these problems can be successfully tackled thanks to two notable frameworks, {\em algebraic quantum field theory}, see \cite{BFDY15} for a recent review, and {\em renormalization in position space}. The former is a mathematically rigorous approach to quantum field theory which highlights the algebraic structure of the algebra of observables of the underlying theory. While it is fully equivalent to the standard formulation used in high energy physics when one considers Minkowski spacetime, it is especially suitable when working on a curved background. The latter has its roots in a seminal work from Epstein and Glaser \cite{Epstein-Glaser-73} and recent results in this direction can be found in \cite{Dang-Herscovich-19,Duetsch-Fredenhagen-Keller-Rejzner-14, Herscovich-19,Keller-09}. It is important to mention that combining these two approached requires the use of techniques traded from microlocal analysis \cite{Brunetti:2009qc,Hollands-Wald-01,Hollands-Wald-02,Rejzner:2016hdj}. The net advantage of this approach is two-fold. On the one hand it avoids any $\varepsilon$-regularization and subtraction of diverging quantities to give a distributional meaning to otherwise pathological singular structures. On the other hands it allows for a direct, complete characterization of the underlying freedom, in other words of the renormalization constants. 

The goal of this paper will be to start a programme aimed at combining together the main tools proper of the algebraic approach with the theory of stochastic partial differential equations. The main merits of our approach are twofold. On the one hand, we shall develop a framework aimed at encompassing renormalization in the analysis of SPDEs without resorting to any a priori regularization scheme. On the other hand we shall show that the algebraic approach allows for a perturbative construction both of the solutions and of the correlation functions of a large class of SPDEs, taking into account contemporary all possible renormalization freedoms. In addition, at each order in perturbation theory, the renormalization has the effect of bringing a modification to the underlying dynamics. The ensuing SPDE, often referred to as the renormalized equation, can be constructed explicitly in the algebraic approach and we will show it via a prototypical example.  In the next subsections we clarify which are our specific objectives and the main results that we shall prove. Since our target audience is broad encompassing both scientists whose main interest lies in probability theory and those whose expertise is mainly in the mathematical aspects of quantum field theory, we reckon that it is desirable to follow an unconventional approach. More precisely we split the introduction in two parts, each addressing the strengths and motivations at the hearth of this work as seen from the viewpoint of the two main research communities mentioned above. Although this leads forcefully to repeating some arguments and concepts, we feel that this course of actions is in the long run more beneficial. 

\subsection{The algebraic approach to SPDEs: the viewpoint from the probability theory side}

In the 80s SPDEs were introduced in some theoretical physics models inspired by quantum field theory, e.g. \cite{JM85,PW81}. The prototypical example considered in these papers is
\[
\Delta  u -  u^3 +\eta=0, \qquad x\in\mathbb R^d/\mathbb Z^d=\mathbb T^d,
\]
where $\eta$ is a {white noise}, namely a Gaussian random {distribution} on $\mathbb R^d$ characterized by
\[
\mathbb E\left[ \exp\left( \langle \eta,f\rangle\right)\right] = \exp\left(\frac12 \|f\|^2_{L^2(\mathbb R^d)}\right),
\]
for all $f\in C^\infty_c(\mathbb R^d)$.
For example, for $d=4$, $ u$ is expected to be a random distribution, namely $ u= u(\eta)\in
\mathcal D'(\mathbb T^4)$.
More precisely $ u$ is expected to belong a.s. to some negative Sobolev space
$H^{-\kappa}(\mathbb T^4)$, $\kappa>0$, and therefore the non-linearity $ u^3$ is \emph{ill-defined}.
This is what makes the above SPDE \emph{singular}: not only existence and uniqueness are problematic,
the very same {notion of solution} is unclear.

For $d=5$, $ u$ is expected to be a random distribution and to belong to some negative Sobolev space
$H^{-1/2}(\mathbb T^5)$ and again $ u^3$ is {ill-defined}.
The higher the dimension $d$, the lower the regularity of $ u$.

In the last 8 years there have been many progresses in the study of singular SPDEs, mainly driven by the theory of regularity structures \cite{Hairer13,Hairer14,Hairer15} and of paradifferential calculus \cite{GIP}. The result is that we have 
\begin{itemize}
	\item a proper notion of solutions,
	\item existence and uniqueness of such solutions for a class of equations,
	\item a non-perturbative approach.
\end{itemize}
With reference to the example of non linear SPDE considered above, the outcome is a map $\eta \mapsto  u= u(\eta)\in\mathcal D'(\mathbb T^d)$: {a random variable} 
with values in the {space of distributions}.

The construction of solutions works as follows: one considers a regularisation $\eta_\varepsilon=\rho_\varepsilon*\eta$ of the white noise, $*$ indicating the stochastic convolution;
then one considers the random regularised PDE (for $d=5$)
\begin{align}\label{Eq: prob-intro-renormalized-SPDE}
	\Delta \widehat  u_\varepsilon -\left(\widehat  u_\varepsilon\right)^3 \, {+\left(\frac{C_1}\varepsilon+C_2\log\varepsilon+R\right)\widehat  u_\varepsilon}+\eta_\varepsilon=0.
\end{align}
The constants $C_1,C_2$ are {fixed}, while $R\in\mathbb R$ can vary.
The theory shows that $$\widehat  u_\varepsilon\to\hat  u, \qquad\varepsilon\to 0$$ as a distribution, and 
$\widehat  u= \widehat u(R)$ is \emph{the renormalised solution} to the equation.
This approach is \emph{non-perturbative}. We have at the same time 
\begin{itemize}
	\item a finite family of {renormalisation parameters} -- $R\in\mathbb{R}$ in the case of Equation \eqref{Eq: prob-intro-renormalized-SPDE}.
	\item a {uniqueness statement} once the renormalisation parameters are fixed.
\end{itemize}

In particular:
\begin{itemize}
	\item for every choice of the renormalisation parameters, there is a unique solution. 
	\item for different choices of the renormalisation parameters, the renormalised solutions can possibly differ.
\end{itemize}

This approach works \emph{pathwise}, namely for a fixed realisation of the noise $\eta$. However,
it is not so simple to compute efficient formulae for the expectations of interesting functionals of the solution $\widehat u$. As a typical example, we would like to compute \emph{correlation functions} of $\widehat u$
\[
f(x_1,\ldots,x_n)={\mathbb E}[\widehat u(x_1)\cdots\widehat u(x_n)], \qquad x_i\in\mathbb T^d.
\]
There are two problems at this stage, namely, since $\widehat u$ is a (random) distribution, 
\begin{itemize}
	\item $f$ may also be not better than a distribution
	\item $f$ may be ill-defined since we are dealing with products of $\widehat{u}$'s.
\end{itemize}

\noindent It is worth noticing that, in a recent remarkable paper, \cite{Barashkov}, a formula has been shown for the {Laplace transform} of the Stochastic Quantization of $\Phi^4_3$ in $d=3$ in the non-perturbative setting (related to $ u$ for $d=5$)
\[
\Lambda(h)=\mathbb E[\exp(\langle \Phi^4_3,h\rangle)], \qquad h\in \mathcal D(\mathbb T^3).
\]
Although, in principle, this formula might allow to compute correlations functions, in practice, it does not seem so simple, due to the complexity of the functional $\Lambda$.

In this paper, we use techniques borrowed from {Algebraic QFT} in order to compute the correlation functions
which are expected for the solution to singular SPDEs. In the remainder of this section, we consider only for expository purposes the concrete example of 
\[
\Delta  u - {\lambda}  u^3 +\eta=0, \qquad x\in\mathbb R^d/\mathbb Z^d=\mathbb T^d.
\]
Also within this approach one has to deal with ill-defined distributions and {ambiguities} due to renormalisation.
At this stage we can treat only the {perturbative} approach, also known as pAQFT, namely the Taylor expansion with respect to the (small) parameter ${\lambda}$. More precisely, we write the above equation in its mild formulation
\[
 u=G*\eta-\lambda\, G* u^3
\]
where $G$ is the Green function of $\Delta$.  Subsequently we consider the $C^\infty(\mathbb{T}^d)$-valued functional $C^\infty(\mathbb T^d)\ni\varphi\mapsto\Psi=\Psi(\varphi)\in C^\infty(\mathbb T^d)$ 
\[
\Psi=\varphi-\lambda\, G*\Psi^3.
\]
Note that this is a completely deterministic setting. 
Moreover $G*\eta\in\mathcal D'(\mathbb T^d)$ has been replaced by $\varphi\in C^\infty(\mathbb T^d)$. To realize how this different problem helps in describing the original one, let us write $\Psi=\Psi(\lambda)$ as a {formal} power series
\[
\Psi\llbracket\lambda\rrbracket=\sum_{j\geq0}\lambda^j F_j\,,
\]
with $ F_j=F_j(\varphi)$ a {$C^\infty$-valued, polynomial functional} on $C^\infty(\mathbb T^d)$. 
The equation yields the explicit expressions
\[
\begin{split}	
	F_0&=\varphi,\qquad
	F_1=-G*\varphi^3,\qquad
	F_2=3G*(\varphi^2\, G*\varphi^3)\,,\\
	F_j&=-\sum_{j_1+j_2+j_3=j-1}G*(F_{j_1}F_{j_2}F_{j_3})\,,\qquad j\geq 3\,.
\end{split}
\]
The series may not be convergent, and this is the well-known conundrum of the perturbative approach.
However this is not the only problem: if we want to go back and replace $\varphi\in C^\infty(\mathbb T^d)$ with
$G*\eta\in\mathcal D'(\mathbb T^d)$, then all terms bar $F_0$ are ill-defined since they contain powers of $G*\eta\in\mathcal D'$.

Note that $G*\eta$ is Gaussian, centered and with {covariance}
\[
Q(x,y)=\int G(x,z)\, G(z,y)\,{\rm d}z
\]
where $Q\in\mathcal{D}^\prime(\mathbb T^d\times\mathbb T^d)$. Observe that, for $d\leq 3$ one can give a meaning to an algebra of polynomial functionals of $\varphi$, e.g.
\[
\Phi \cdot_Q \Phi (f;\varphi)= \int_{\mathbb T^d} \left(\varphi^2(x)+Q(x,x)\right) f(x)\,{\rm d}x
\]
where $\Phi(f;\varphi)=\int_{\mathbb T^d} f(x)\,\varphi(x)\,{\rm d}x$ and $f\in\mathcal D(\mathbb T^d)$.

However for $d\geq 4$ this expression is ill-defined since $Q(x,x)$ corresponds formally to the integral kernel of an element of $\mathcal{D}^\prime(\mathbb{T}^d)$ which should be realized as $Q\delta_2$, where $\delta_2$ is the Dirac delta supported on $\textrm{Diag}_2(\mathbb{T}^d)$ the diagonal of $\mathbb{T}^d\times\mathbb{T}^d$. Yet, techniques of {microlocal analysis} allow to find and characterize all $\widehat Q\in\mathcal D'(\mathbb T^d\times\mathbb T^d)$
which extend $Q$, giving thus a meaning to the restriction to $\textrm{Diag}_2(\mathbb{T}^d)$. This is based on the study of the {wavefront set} and of the {scaling degree} of the relevant
distributions, see \cite{Hormander-I-03} and \cite{Brunetti-Fredenhagen-00}.
We obtain \emph{existence} of a well defined product
\[
\Phi \cdot_Q \Phi (f;\varphi)= \int_{\mathbb T^d} \left(\varphi^2(x)+\widehat P(x)\right) f(x)\,{\rm d}x
\]
with $\widehat P\in\mathcal D'(\mathbb T^d)$. This is like a formal \emph{chaos expansion}.
More generally, we can define an algebra of
{distribution-valued polynomial functionals} of $\varphi$.

One might wonder about the \emph{uniqueness} of the above product. If $d\geq 4$ then there is a {family} of possible choices for $\cdot_Q$. These are the ambiguities due to {renormalisation} and they can be fully characterized.
For example any other possible choice of $\widehat Q\in\mathcal D'(\mathbb T^d)$ must satisfy
$\widehat Q-\widehat P\in\mathbb R{\bf 1}$.

It is very important to stress that the construction of the $\cdot_Q$ product is not merely algebraic but combines analytical and combinatorial tools: combinatorics shows which counterterms are needed
for renormalisation, analysis guarantees that the renormalised objects are well-defined as distributions.
In particular, in full analogy with the approaches of \cite{Hairer14,Hairer15} and of \cite{GIP}, if $d\leq 5$, it is sufficient to renormalize only a finite number of distributions to control that no divergences occur at any order in the perturbative series.  

Focusing once more on the problem in hand
\[
\Psi=\varphi-\lambda\, G*\Psi^3,
\qquad	\Psi\llbracket\lambda\rrbracket=\sum_{j\geq0}\lambda^j\, F_j\,,
\]
\[
\begin{split}	
{\textstyle\rm with} \qquad		F_0&=\varphi,\qquad
	F_1=-G*\varphi^3,\qquad
	F_2=3G*(\varphi^2\, G*\varphi^3)\,,\\
	F_j&=-\sum_{j_1+j_2+j_3=j-1}G*(F_{j_1}F_{j_2}F_{j_3})\,,\qquad j\geq 3\,.
\end{split}
\]
In order to give a proper description of the underlying renormalised SPDE, we renormalise by writing
\[
\widehat\Psi=\Phi-\lambda\, G*(\widehat\Psi\cdot_Q\widehat\Psi\cdot_Q\widehat\Psi),
\qquad\qquad
\widehat\Psi\llbracket\lambda\rrbracket=\sum_{j\geq0}\lambda^j\, \widehat F_j\,,
\]
\[
\begin{split}	
{\textstyle\rm with} \qquad	\widehat F_0&=\Phi,\qquad
	\widehat F_1=-G*(\Phi\cdot_Q\Phi\cdot_Q\Phi),\qquad
	\widehat F_2=3G*(\Phi\cdot_Q\Phi\cdot_Q G*(\Phi\cdot_Q\Phi\cdot_Q\Phi))\,,\\
	\widehat F_j&=-\sum_{j_1+j_2+j_3=j-1}G*(\widehat F_{j_1}\cdot_Q \widehat F_{j_2}\cdot_Q \widehat F_{j_3})\,,\qquad j\geq 3\,.
\end{split}
\]
The results of this paper show that all $\widehat F_j$'s are well-defined as distribution-valued polynomial functionals of $\varphi$, namely $\widehat F_j(\,\cdot\, ; \varphi)\in \mathcal D'$ and $C^\infty\ni\varphi\mapsto
\widehat F_j(f ; \varphi)$ is a polynomial functional for all $f\in\mathcal D$.
We also expect the following result: if $\widehat u=\widehat u(\lambda)$ is the renormalised solution to
\[
\Delta  u - \lambda u^3 +\eta=0
\]
then for all $f\in\mathcal D$ 
\[
\left.\frac{{\rm d}^k}{{\rm d} \lambda^k}\mathbb E[\langle\widehat u(\lambda),f\rangle]\right|_{\lambda=0} = k! \, \widehat F_k(f;0).
\]
We have an analogous construction which allows to compute correlation functions of $\widehat u$.
We also plan to study the connection between the two different sets of renormalisation constants in a future work.

\subsection{The algebraic approach to SPDEs: the viewpoint from the algebraic quantum field theory side}

In the following we outline how, to which extent and for which purpose one can analyze SPDEs starting from the standard viewpoint of the algebraic approach to quantum field theory. We indicate with $M$ a smooth, connected manifold of dimension $\dim M=d\geq 2$, while with $E$ we refer to an element of a class of microhypoelliptic linear partial differential operators, {\it cf.} \cite{Hormander-VolIII}, which we shall characterize in Section \ref{Sec: Basic definitions}. We denote $\mathcal{D}(M):=C^\infty_c(M)$ and $\mathcal{E}(M):=C^\infty(M)$.

For the sake of simplicity, at this stage we also assume that $E$ possesses a fundamental solution $G\in\mathcal{D}^\prime(M\times M)$. This hypothesis will be subsequently relaxed and it will suffice to consider parametrices 
associated to $E$. Concrete settings compatible with our hypotheses are 
\begin{itemize}
	\item the Laplace-Beltrami operator $-\Delta_g$ if $M$ is endowed with a Riemannian structure $g$ 
	\item the heat operator $-\partial_t+\Delta_h$ if $M$ is diffeomeorphic to $\mathbb{R}\times\Sigma$ and $\Sigma$ is endowed with a Riemannian structure $h$. 
\end{itemize}
In addition we consider the following class of semi-linear SPDEs
\begin{align}\label{Eq: intro differential equation}
	E\widehat{\psi}
	=\widehat{\xi}+ F[\widehat{\psi}]\,,
	%	\qquad\mu>0\,,
\end{align}
where $F:\mathbb{R}\to\mathbb{R}$ is a smooth non-linear potential. In Equation \eqref{Eq: intro differential equation} $\widehat{\psi}$ has to be interpreted as a distribution on $M$ which, once smeared with a test function, yields a random variable. At the same time $\widehat{\xi}$ denotes the standard realization of Gaussian white noise with vanishing mean and covariance $\mathbb{E}(\widehat{\xi}(x)\,\widehat{\xi}(y))=\delta(x-y)$. For more details on the formulation of stochastic PDEs we  refer for example to \cite{Chandra,Hairer09}. 

\begin{remark}\label{Rem: generality_1}
	Observe that we have chosen to work with Equation \eqref{Eq: intro differential equation} for simplicity of presentation. There is a priori no obstruction to considering more general non-linearities, where $F$ might depend 
	for example on the gradient of $\widehat{\psi}$. Yet this would lead to a more complicated analysis especially from the viewpoint of the notation and therefore we decided to focus our attention on a restricted class of SPDEs.
\end{remark}

Equation \eqref{Eq: intro differential equation} can be written in its integral form
\begin{align}\label{Eq: intro integral equation}
	\widehat{\psi}
	=G\textproto{s}\widehat{\xi}
	-G\textproto{s}F[\widehat{\psi}]
	=\widehat{\varphi}
	-G\textproto{s}F[\widehat{\psi}]\,,
\end{align}
where $\textproto{s}$ indicates the stochastic convolution, {\it cf.} \cite{Chandra}, while $\widehat{\varphi}(x):=(G\textproto{s}\widehat{\xi})(x)$ is a distribution with values in random variables. Notice that, once smeared against any test-function $f\in\mathcal{D}(M)$, $\widehat{\varphi}(x)$ is still Gaussian with mean and covariance 
\begin{align}\label{Eq: expectation value intro}
	\mathbb{E}(\widehat{\varphi}(x))=0\,,\qquad
	\mathbb{E}(\widehat{\varphi}(x)\,\widehat{\varphi}(y))=Q(x,y):=(G\circ G)(x,y)\,,
\end{align}
where $\circ$ is the composition of distribution, {\it cf.} Equation \eqref{Eq: distribution composition}. Note that if
$F_1,F_2\in \mathcal{D}(M\times M)$ then $F_1\circ F_2\in  \mathcal{D}(M\times M)$ is given by
\begin{equation}\label{Eq:F1F2}
	F_1\circ F_2\,(x,y) =\int_M F_1(x,z)\, F_2(z,y)\, \mathrm{d}\mu(z)\, ,
\end{equation}
where $d\mu(z)$ is a suitable integration measure. If $F\in \mathcal{D}(M\times M)$ and $F(x,y)=F(y,x)$ then we have moreover
\begin{equation}\label{Eq:FF}
	F\circ F\, (x,x)=\int_M F^2(x,z)\, \mathrm{d}z\, .
\end{equation}
However in the general case where $G\in\mathcal{D}'(M\times M)$, $G\circ G$ and $G^2$ may be ill-defined.

\begin{remark}\label{Rem: awareness of IR problems}
	The stochastic convolution between $G$ and $\widehat{\psi}$ might be a priori ill-defined unless for example $G\in\mathcal{E}^\prime(M\times M)$. Yet this occurs only in special circumstances, {\it e.g.} if $M$ is a compact manifold and $E$ is a suitable elliptic operator such as $-\Delta_g+m^2$ where $-\Delta_g$ is the Laplace-Beltrami operator built out of a Riemannian metric $g$ on $M$, while $m^2>0$ is constant, which, in concrete models, is interpreted as a mass term. Observe that a similar problem occurs in the definition of $G\circ G$. These hurdles, often referred to as infrared problems in the mathematical physics literature, are bypassed using suitable cut-off functions. In this section we avoid introducing these and we are implicitly assuming that we are considering only those scenarios where such problems do not occur. In the main body of this work, we shall instead address the issue when necessary. 
\end{remark}

Observe that both in Equation \eqref{Eq: intro differential equation} and in Equation \eqref{Eq: intro integral equation} the term $F[\widehat{\psi}]$ might not be well-defined as it can contain powers of an a priori generic distribution. As a concrete example consider $F[\widehat{\psi}]=\widehat{\psi}^3$ and try to solve iteratively Equation \eqref{Eq: intro differential equation}. In other words, if we consider a formal power series, $\widehat{\psi}\llbracket\lambda\rrbracket=\sum\limits_{j\geq 0}\lambda^j \widehat{\psi}_j$, one obtains 
\begin{align}\label{Eq: intro perturbative sequence}
	\widehat{\psi}_0:=\widehat{\varphi}\,,\qquad
	\widehat{\psi}_1=-G\textproto{s}\widehat{\varphi}^3\,,\qquad
	\widehat{\psi}_j=-G\textproto{s}\sum_{j_1+j_2+j_3=j-1}\widehat{\psi}_{j_1}\widehat{\psi}_{j_2}\widehat{\psi}_{j_3}\qquad
	j\geq 1\,.
\end{align}
Hence one has to cope with the problem that $\widehat{\varphi}^3$ is not a priori well defined. Similarly, if one tries to look for solutions via a fixed point argument, formally one considers the sequence 
\begin{align}\label{Eq: intro fixed point sequence}
	\widehat{\psi}_0:=\widehat{\varphi}\,,\qquad
	\widehat{\psi}_k=	\widehat{\varphi}-G\textproto{s}(\widehat{\psi}_{k-1}^3)\qquad k\geq 1\,,
\end{align}
where the same kind of problem as in Equation \eqref{Eq: intro perturbative sequence} occurs. 

The obstruction in defining the powers of $\widehat{\varphi}$ are referred to in the theoretical and mathematical physics literature as ultraviolet singularities. It is instructive to observe how this problem manifests itself at the level of expectation values. For example, using Equations \eqref{Eq: expectation value intro}-\eqref{Eq:F1F2}-\eqref{Eq:FF}
\begin{align}\label{Eq: singular expectation value intro}
	\mathbb{E}(\widehat{\varphi}^2(f))
	=(G\circ G)(f\delta_{\mathrm{Diag}_2})=G^2(f\otimes 1)\,,
\end{align}
where $\delta_{\mathrm{Diag}_2}\in\mathcal{D}'(M\times M)$ is the Dirac delta distribution centred on the diagonal $\mathrm{Diag}_2$ of $M\times M$. The latter expression is ill-defined as a distribution due to the presence of the square $G^2$.
On the one hand it is interesting to observe that, in this last formula, the difficulties arise only from the singular behaviour of the fundamental solution of the underlying linear operator $E$.
On the other hand, $G^2(x,y)$ is well-defined whenever $x\neq y$ and, therefore $G^2\in\mathcal{D}^\prime(M\times M\setminus\textrm{Diag}_2)$ where $\textrm{Diag}_2=\{(x,x)\;|\;x\in M \}$.
Hence the question whether one can give a coherent meaning to Equation \eqref{Eq: singular expectation value intro} is tantamount to asking whether it is possible to extend $G^2$ to the diagonal $\textrm{Diag}_2$ as a distribution.
A structurally identical problem occurs if one considers the expectation value of different functions of $\widehat{\varphi}$ and the key ingredient of renormalization consists of addressing whether such extensions do exist and, if so, whether they are unique or not.

\vskip .2cm 

A completely different scenario, where the same class of problems occurs, is relativistic quantum field theory. Especially the past twenty years have seen the development of a mathematical framework which also allows to encompass at a structural level the role and the features of renormalization. This setting goes under the name of algebraic quantum field theory (AQFT). For reasons of practicality and conciseness we cannot enter into a detailed explanation of AQFT and we refer to a few recent reviews for further details, {\it cf.} \cite{BFDY15, Fredenhagen:2014lda}. Yet we stress that the main rationale behind the algebraic approach to quantum field theory is the identification of two key elements. The first is a $*$-algebra $\mathcal{A}_{\mathrm{rel}}$ whose elements encompass the structural properties of the underlying theory, in particular dynamics. The second is a state, that is a positive, normalized linear functional over $\mathcal{A}_{\mathrm{rel}}$ out of which one can recover the underlying expectation values and correlation functions. Once more we stress that this language in combination with tools of microlocal analysis allows for a rigorous implementation of renormalization, {\it cf.} \cite{Brunetti:2009qc,Hollands-Wald-01,Hollands-Wald-02}. 

Although AQFT has been mainly developed with the application to relativistic quantum field theories in mind, whose dynamics is ruled by hyperbolic partial differential equations on Lorentzian backgrounds, such framework is very versatile and it can be adapted to rather different contexts, such as Riemannian manifolds endowed with elliptic partial differential operators -- see for example \cite{CDDR20,DDR20}.

The goal of this paper is to adapt the framework of AQFT and in particular its implementation of renormalization to the study of SPDEs in the form of Equation \eqref{Eq: intro differential equation}. The first step in this quest is the identification of a suitable counterpart of the algebra $\mathcal{A}_{\mathrm{rel}}$ mentioned above. Inspired by the so-called functional formalism \cite{DDR20,Fredenhagen:2014lda}, we do not focus directly on the stochastic behaviour encoded in $\widehat{\xi}$, rather we consider $\mathcal{D}^\prime(M;\operatorname{Pol})$, namely distributions with values in polynomial functionals over $\mathcal{E}(M)$ -- see Section \ref{Sec: Basic definitions} for the precise definitions. Herein we identify two distinguished elements
\begin{equation}\label{e:Phi}
	\Phi(f;\varphi)=\int_M \varphi f\mu,\qquad
	\boldsymbol{1}(f;\varphi)=\int_M f\mu
\end{equation}
where $\varphi\in\mathcal{E}(M)$ while $f\in\mathcal{D}(M)$, whereas $\mu$ indicates a volume form over $M$, which is henceforth fixed once and for all. One can realize that $\Phi$ and $\boldsymbol{1}$ can be used as generators of an algebra $\mathcal{A}$ whose composition is the pointwise product, {\it cf.} Section \ref{Sec: Basic definitions}. Yet $\mathcal{A}$ falls short of our goal since neither the dynamics ruled by Equation \eqref{Eq: intro differential equation} nor the randomness due to $\widehat{\xi}$ are encoded in $\mathcal{A}$. Still following the rationale of AQFT, this shortcoming is overcome by deforming the product of $\mathcal{A}$ so that the new product, indicated with $\cdot_Q$, encompasses the information both of the dynamics and of the randomness of the underlying SPDE. Given any two admissible functional-valued distributions $\tau_1,\tau_2\in\mathcal{A}\subset\mathcal{D}^\prime(M;\operatorname{Pol})$, the new product is formally realized as:
\begin{align}\label{Eq: formal delta-product}
	(\tau_1\cdot_Q \tau_2)(f;\varphi)
	&=\sum_{k\geq 0}\frac{1}{k!}[(\delta_{\mathrm{Diag}_2}\otimes Q^{\otimes k})\cdot(\tau_1^{(k)}\widetilde{\otimes} \tau_2^{(k)})](f\otimes1_{1+2k};\varphi)
	=\sum_{k\geq 0}t_k(f\otimes1_{1+2k};\varphi)\,,
\end{align}
where $\tau_i^{(k)}$, $i=1,2$ indicate the $k$-th functional derivatives, while $Q=G\circ G$ -- see Section \ref{Sec: Basic definitions} and Section \ref{Sec: construction of AcdotQ} for the definitions. Furthermore the symbol $\widetilde{\otimes}$ indicates that the integral kernel of $(\delta_{\mathrm{Diag}_2}\otimes Q^{\otimes k})\cdot(\tau_1^{(k)}\widetilde{\otimes} \tau_2^{(k)})$ reads
\begin{equation}\label{Eq: tilde tensor product}
	\delta(x_1,x_2)\prod_{i=1}^kQ(z_i,y_i)\,\tau_1^{(k)}(x_1,z_1,\dots,z_k)\,\tau_2^{(k)}(x_2,y_1,\dots,y_k) .
\end{equation}
For example, we expect that
\[
[\Phi\cdot_{Q} \Phi](f;\varphi)
=\int_M f(x)[\varphi^2(x)+Q(x,x)]\mathrm{d}\mu(x)\,,
\]
where $\Phi$ has been defined in \eqref{e:Phi}. This expression is in general only formal since $Q=G\circ G$ is
ill-defined on the diagonal $\mathrm{Diag}_2$.

One can realize that Equation \eqref{Eq: formal delta-product} is a priori ill-defined since the involved functional-valued distributions $t_k$ are well-defined only as distributions on $M^{2k+2}$ up to the total diagonal $\mathrm{Diag}_{2k+2}$ of $M^{2k+2}$, that is, $t_k(\cdot\,;\varphi)\in\mathcal{D}'(M^{2k+2}\setminus\mathrm{Diag}_{2k+2})$.
Yet, adapting to the case in hand techniques from microlocal analysis similarly to the case of relativistic quantum field theory \cite{Brunetti-Fredenhagen-00}, one can prove that, under mild hypotheses on the underlying operator $E$, it is possible extend $t_k$ to a distribution $\hat{t}_k\in\mathcal{D}^\prime(M^{2k+2})$.
Such extension procedure, when existent, might not be unique, but the ambiguities in the procedure are classified.
In other words we can give a mathematically precise meaning to Equation \eqref{Eq: formal delta-product}. This is the first main result of our work, namely
\begin{enumerate}
	\item we prove existence of the algebra $\mathcal{A}_{\cdot_Q}$, {\it cf.} Theorem \ref{Thm: Gamma cdotQ existence} and Corollary \ref{Cor: construction of AcdotQ},
	\item we classify and characterize the (renormalization) freedom in constructing $\mathcal{A}_{\cdot_Q}$, {\it cf.} Theorem \ref{Thm: GammacdotQ uniqueness}.
\end{enumerate}

\begin{remark}\label{Rem: Significance for SPDEs}
	The algebra $\mathcal{A}_{\cdot Q}$ encompasses at the algebraic level the fundamental structures encoded by the white noise. As a matter of fact the expectation value of the random field $\mathbb{E}(\widehat{\varphi}^2(x))$ as per Equation \eqref{Eq: expectation value intro} corresponds to evaluating at the configuration $\varphi=0$ the product of two generators $\Phi$ of $\mathcal{A}_{\cdot Q}$, that is, for all $f\in\mathcal{D}(M)$
	$$\left(\Phi\cdot_Q\Phi\right)(f;0)=\widehat{Q\delta_{\mathrm{Diag}_2}}(f),$$
	where $=\widehat{Q\delta_{\mathrm{Diag}_2}}$ indicates a renormalized version of the otherwise ill-defined composition between the operator $Q$ in Equation \eqref{Eq: expectation value intro} and $\delta_{\mathrm{Diag}_2}$, as discussed in the proof of Theorem \ref{Thm: Gamma cdotQ existence}, see page \pageref{Qdelta}.
	With a similar calculation one can realize that $\mathcal{A}_{\cdot Q}$ encompasses all other relevant renormalized expectation values, {\it e.g.} for all $f\in\mathcal{D}(M)$
	$$\mathbb{E}(\widehat{\varphi}^k)(f)=(\underbrace{\Phi\cdot_Q\ldots\cdot_Q\Phi}_k)(f;0).$$
\end{remark}

Hence the algebraic approach has the advantage of providing a systematic way of dealing with renormalization without resorting to any regularization scheme. Yet the structures introduced above are still not sufficient to catch all the information which one can extract from an SPDE. As a matter of fact correlations between the random fields are completely neglected. For this reason and still inspired by the algebraic approach to quantum field theory, we introduce an additional structure built out of $\mathcal{A}_{\cdot Q}$, namely 
$$\mathcal{T}(\mathcal{A}_{\cdot_Q})=\mathcal{E}(M)\oplus\bigoplus_{\ell>0}\mathcal{A}_{\cdot_Q}^{\otimes\ell}$$
Similarly to the preceding part of our analysis we endow $\mathcal{T}(\mathcal{A}_{\cdot_Q})$ with a deformation of its natural product which allows to encode the structure of the correlation functions of the underlying SPDE. Focusing on a concrete example, our goal is to give a meaning to formal expression such as 
\begin{multline*}
	[\Phi^2\bullet_Q \Phi^2](f_1\otimes f_2;\varphi)
	=\int\limits_{M\times M}
	f_1(x_1)f_2(x_2)\big[
	\varphi(x_1)^2\varphi(x_2)^2
	\\+4\varphi(x_1)Q(x_1,x_2)\varphi(x_2)
	+2Q^2(x_1,x_2)
	\big]\mathrm{d}\mu(x)\mathrm{d}\mu(y),
\end{multline*}
where $f_1,f_2\in\mathcal{D}(M)$, while $\varphi\in\mathcal{E}(M)$. Also in this case one has to cope with the problem that in general $Q^2\in\mathcal{D}^\prime(M\times M\setminus\mathrm{Diag}_2)$ and, as in the construction of the algebra $\mathcal{A}_{\cdot Q},$ we need tools from microlocal analysis to prove that there exists at least one extension of $Q^2$ to the whole $M\times M$ and to classify whether ambiguities arise in this process. This analysis leads us to our second set of main results:
\begin{enumerate}
	\item We prove existence of a deformation of $\mathcal{T}(\mathcal{A}_{\cdot Q})$, referred to as $\mathcal{A}_{\bullet_Q}$, {\it cf.} Theorem \ref{Thm: GammabulletQ existence},
	\item We classify and characterize the renormalization freedom in constructing $\mathcal{A}_{\bullet_Q}$, {\it cf.} Theorem \ref{Thm: GammabulletQ uniqueness}.
\end{enumerate} 

\begin{remark}\label{Rem: significance for SPDEs correlations}
	The algebra $\mathcal{A}_{\bullet Q}$ encompasses at the algebraic level the information on the correlations between the underlying random fields. More precisely, it allows to express any expectation value such as $\mathbb{E}(\Phi(f_1)\Phi(f_2))=(\Phi\bullet_Q\Phi)(f_1\otimes f_2;0)$ where the right hand side is once more evaluated at the configuration $\varphi=0$, while renormalization enters the game in the definition of the product $\bullet_Q$.
\end{remark}

With the new structures introduced it is possible to revert the attention back to Equation \eqref{Eq: intro differential equation} and to Equation \eqref{Eq: intro integral equation}. If one focuses on the sequences as in Equation \eqref{Eq: intro perturbative sequence} or in Equation \eqref{Eq: intro fixed point sequence}, each term in the expansion can be given a meaning in terms of the functionals of the algebra $\mathcal{A}_{\cdot Q}$. Evaluation at the configuration $\varphi=0$ gives rise to the expectation value of each term in the series. As we analyze in detail in Section \ref{Sec: Application to the Phi3d Model} considering for definiteness the specific example of the $\Phi^3_d$-model, the algebraic approach has the net advantage of allowing a control of the renormalization freedoms present at each order in perturbation theory without going through any regularizing procedure. As we discuss in Section \ref{Sec: Renormalized equation} we are also able to construct the so-called renormalized equation, namely at each order in perturbation theory one needs to modify the underlying SPDE to take into account the effect of the renormalization procedure. Most notably our construction has two additional advantages. On the one hand it allows to keep track not only of the expectation value, but also of the correlations between the solutions that we individuate. On the other hand our procedure is quite explicit as we show with a concrete example in Section \ref{Sec: computations at first order}. Herein we compute explicitly the possible renormalization extensions of the $n$-th powers of the fundamental solutions of the heat operator on $\mathbb{R}^n$, which play a key role in the perturbative construction of the solutions for stochastic $\Phi^3_d$ model.

Yet one should bear in mind that the algebraic method does not offer any information on the convergence, neither of Equation \eqref{Eq: intro perturbative sequence} nor of Equation \eqref{Eq: intro fixed point sequence} but it allows only to give precise meaning to each term of the series. Any improvement in this direction goes through a better understanding of how one can intermingle the whole framework of regularity structures with the algebraic one.

\vskip .3cm

The work is organized as follows: In Section \ref{Sec: Basic definitions} we discuss the general framework we are working with. In particular we introduce the notion of functional-valued distributions, the space $\mathcal{D}^\prime_{\mathrm{C}}(M;\operatorname{Pol})$ and its main analytic properties. We discuss how to endow such space with an algebra structure via a deformation of the pointwise one by means of a smooth integral kernel. In Section \ref{Sec: construction of AcdotQ} we discuss the construction of the algebra $\mathcal{A}_{\cdot Q}$ proving an existence theorem. In Section \ref{Sec: Correlations and bulletQ product} we focus instead on the correlations between the random field discussing the algebra $\mathcal{A}_{\bullet_Q}$ and proving an existence theorem. In Section \ref{Sec: Uniqueness theorems} we discuss the freedom in the construction both of $\mathcal{A}_{\cdot Q}$ and of $\mathcal{A}_{\bullet Q}$. In Section \ref{Sec: Application to the Phi3d Model} we apply the algebraic approach to a concrete example: the $\Phi^3_d$ model. We analyze it at a perturbative level discussing up to the first order both the solutions of the underlying SPDE and their two-point correlations. Finally in Section \ref{Sec: Diverging diagrams in the sub-critical case} we show within the algebraic framework that, in the so-called sub-critical cases, {\it i.e.} $d=2,3$ the number of distributions needing to be renormalized is finite. In the appendix we outline succinctly some basic definitions and notions from microlocal analysis, in particular the notion of scaling degree which plays a key role in studying whether a distribution can be renormalized and whether some freedom in this process occurs.

\paragraph{Acknowledgements.}
We are thankful to  M. Capoferri, F. Faldino, G. Montagna and N. Pinamonti for helpful discussions on the topic and for the comments on this manuscript. N.D. is supported by a Postdoctoral Fellowship of the Alexander von Humboldt Foundation.
The work of P.R. was partly supported by a fellowship of the ``Progetto Giovani GNFM 2019" under the project ``Factorization Algebras vs AQFTs on Riemannian manifolds" fostered by the National Group of Mathematical Physics (GNFM-INdAM).

\section{General Framework}\label{Sec: Basic definitions}

The goal of this section is twofold. On the one hand we want to fix the basic notations and conventions. On the other hand, we introduce the key analytic and algebraic tools which are of relevance in our analysis. Although our main results are contained in Sections \ref{Sec: construction of AcdotQ}, \ref{Sec: Correlations and bulletQ product} and \ref{Sec: Uniqueness theorems}, here we shall prove a few structural properties of the underlying mathematical framework, which might be of independent interest. 

Throughout this paper both $M$ and $N$ will indicate a finite dimensional smooth manifold, such that $\partial N=\partial M=\emptyset$ while $\mathbb{D}(M)$ and $\mathbb{D}(N)$ denote the associated density bundle. We recall that this is smooth line bundle over $M$ whose fiber at each $p\in M$ is $\mathbb{D}_p(M)\equiv\mathbb{D}(T_pM)$, the set of densities on $T_pM$, {\it cf.} \cite[Ch. 16]{Lee}. For definiteness we shall assume that a reference top-density, say $\mu_M$ and $\mu_N$, has been chosen and, whenever one of the manifolds is endowed with a Riemannian structure, it coincides with the metric induced volume form.
On top of $M$ we consider the following specific classes of microhypoelliptic, {\it cf.} \cite[Ch. XXII]{Hormander-VolIII}, linear differential operators
\begin{enumerate}
	\item a second order, elliptic, partial differential operator $E$ on $M$,
	\item $E=\partial_t-\widetilde{E}$ if the manifold $M$ is homeomorphic to $\mathbb{R}\times\Sigma$, $t$ is the standard Euclidean coordinate over $\mathbb{R}$ while $\widetilde{E}$ is a $t$-independent, second order, differential, elliptic operator on $\Sigma$. 
\end{enumerate}
In addition we indicate with $P$ a parametrix associated to $E$, {\it cf.} \cite{Hormander-VolIII}, while with $E^*$ we refer to its formal adjoint whose parametrix is denoted by $P^*$. For the reader's convenience we recall its definition

\begin{definition}\label{Def: parametrix}
	Let $M$ be a smooth manifold and let $E$ be any scalar, pseudodifferential operator on $M$. We call {\em parametrix} a bi-distribution $P\in\mathcal{D}^\prime(M\times M)$ such that 
	$$PE-\mathrm{id}|_{\mathcal{D}(M\times M)}\in\mathcal{E}(M\times M),\quad\textrm{and}\quad EP-\mathrm{id}|_{\mathcal{D}(M\times M)}\in\mathcal{E}(M\times M).$$
\end{definition}

\begin{remark}\label{Rem: non compact M}
In the following sections, unless stated otherwise we work under the assumption that $M$ is a a compact $d$-dimensional manifold $M$, while $E$ is an elliptic operator, whose parametrix is indicated with $P\in\mathcal{D}'(M\times M)$.

Observe that the compactness assumption can be dropped at the price of introducing an arbitrary cut-off $\chi\in\mathcal{D}(M)$. As a matter of fact expressions such as $P(f\otimes\varphi)$, $f\in\mathcal{D}(M)$ and $\varphi\in\mathcal{E}(M)$ are otherwise meaningless. Hence, if $M$ is non-compact, one ought to replace everywhere $P$ with $P\cdot(1\otimes\chi)\in\mathcal{D}^\prime(M\times M)$, where $\cdot$ denotes here the pointwise product between  $P\in\mathcal{D}^\prime(M\times M)$ and the smooth function $1\otimes\chi\in\mathcal{E}(M\times M)$.
\end{remark}

\begin{remark}\label{Rem: parabolic case}
Since we are also interested in the case where $M=\mathbb{R}\times\Sigma$, while $E=\partial_t-\widetilde{E}$, we stress that all our results are valid also in this case provided that one introduces a cut-off function as per Remark \ref{Rem: non compact M} and that one makes minor modifications to the proofs. These will be highlighted in specific remarks which will be denoted throughout the text as the ``parabolic case''.
\end{remark}

\begin{remark}\label{Rem: possible extensions}
We stress that most of our analysis and results could be applied to other classes of microhypoelliptic operators, provided that the necessary changes in the proofs are implemented. Yet we decided to focus on the above two classes since, in our opinion, they are the most significant ones in concrete applications. In addition working in full generality would have required dealing with many specific, different sub-cases, with the risk of losing clarity in the presentation.	
\end{remark}

\subsection*{Analytic Tools}

In this section we introduce the key analytic tools which are at heart of this paper. In particular we discuss functional-valued distributions and how they can codify expectation values and correlations of (shifted) regularized random fields. In the following with $\mathcal{D}(N)$ ({\em resp.} $\mathcal{E}(M)$) we refer to the space of smooth and compactly supported functions on $N$ ({\em resp.} smooth functions on $M$) endowed with the standard locally convex topology. 

\begin{remark}
	In agreement with Remark \ref{Rem: non compact M} we are implicitly assuming that both $M$ and $N$ are compact manifolds. In this case there is no reason to distinguish between $\mathcal{D}(M)$ and $\mathcal{E}(M)$. Yet, since this difference becomes relevant when the compactness assumption is dropped, {\it cf.} Remark \ref{Rem: non compact M}, we feel that is convenient to keep the two symbols separate in order to highlight which is the correct one to use when $M$ is a non-compact manifold.
\end{remark}

In addition with $\mathcal{D}_\mu(N)\equiv\Gamma^\infty_0(\bigwedge^{\textrm{top}}T^*N)$ and $\mathcal{E}_\mu(M)\equiv\Gamma^\infty(\bigwedge^{\textrm{top}}T^*M)$ we denote respectively the smooth and compactly supported sections of the bundle of densities on $N$ and the smooth sections of the bundle of densities on $M$. Observe that every element of $\mathcal{D}(N)$ induces one of $\mathcal{D}_\mu(N)$ via the correspondence $\mathcal{D}(N)\ni f\mapsto f\mu_N$ and, similarly each $\varphi\in\mathcal{E}(M)$ induces $\varphi_{\mu_M}\doteq\varphi\mu_M\in\mathcal{E}_\mu(M)$.
We shall adopt the convention to omit the subscripts $M$ and $N$ whenever they are obvious from the context.

\begin{definition}\label{Def: functional-valued distributions}
	We call {\bf functional-valued distribution} $\tau\in\mathcal{D}'(N;\operatorname{Fun})$ a map
	\begin{align*}
		\tau\colon \mathcal{D}(N)\times\mathcal{E}(M)\ni(f,\varphi)\mapsto \tau(f;\varphi)\in\mathbb{C}\,,
	\end{align*}
	which is linear in the first component and continuous in the locally convex topology of $\mathcal{D}(N)\times\mathcal{E}(M)$. At the same time we indicate with $\tau^{(k)}\in\mathcal{D}'(N\times M^k;\operatorname{Fun})$ the $k$-order functional derivatives of $\tau\in\mathcal{D}'(N;\operatorname{Fun})$ such that, for all $f\in\mathcal{D}(N)$, $\psi_1,\ldots,\psi_k,\varphi\in\mathcal{E}(M)$,
	\begin{align}\label{Eq: functional derivative}
		\tau^{(k)}(f\otimes\psi_1\otimes\ldots\otimes\psi_k;\varphi)\doteq\frac{\partial^k}{\partial s_1\cdots\partial s_k}\tau(f;s_1\psi_1+\ldots +s_k\psi_k+\varphi)\bigg|_{s_1=\ldots s_k=0}\,.
	\end{align}
 Furthermore, we call $\tau\in\mathcal{D}'(M;\operatorname{Fun})$ $\varphi$-polynomial if there exists $\bar{k}\in\mathbb{N}$ such that $\tau^{(k)}=0$ for all $k\geq\bar{k}$.
	We denote by $\mathcal{D}'(N;\operatorname{Pol})$ the vector space of functional-valued $\varphi$-polynomial distributions.
\end{definition}

Notice that, in view of Definition \ref{Def: functional-valued distributions}, $\tau^{(k)}$ is symmetric and compactly supported in the last $k$ arguments.

\begin{remark}\label{Rem: Directional derivative}
	It is convenient to introduce a counterpart of a directional derivative in this framework. More precisely, for a given $\psi\in \mathcal{E}(M)$ we can define
	\begin{align}\label{Eq: functional differential}
	\delta_\psi\colon\mathcal{D}'(N;\operatorname{Fun})
	\to\mathcal{D}'(N;\operatorname{Fun})\,,\qquad
	[\delta_\psi \tau](f;\varphi)
	\doteq\tau^{(1)}(f\otimes\psi;\varphi)\,.
	\end{align}
\end{remark}

\begin{example}\label{Ex: examples of functional-valued polynomial distributions}
    For the reader's convenience we give a few notable examples of $\varphi$-polynomial, which one can expect to play a key r\^ole in concrete models. In particular, fixing $N=M$ in Definition \ref{Def: functional-valued distributions}, the functional $\Phi$ defined in \eqref{e:Phi} belongs to $\mathcal{D}'(M;\operatorname{Pol})$.
 In analogy with this notation, we set $\Phi^2\in\mathcal{D}'(M;\operatorname{Pol})$ as 
	\begin{align*}
		\Phi^2(f;\varphi)
		=\int\limits_M \varphi^2(x)f(x)\mathrm{d}\mu(x)\,,
		\quad\forall\, f\in\mathcal{D}(M)\,, \
		\forall\,\varphi\in \mathcal{E}(M)\,.
	\end{align*}
Another notable example is $\Phi P\circledast\Phi\in\mathcal{D}'(M;\operatorname{Pol})$
\begin{align*}
\left[\Phi \, P\circledast\Phi\right](f;\varphi)
=\int\limits_M\varphi(x)\, (P\circledast\varphi)(x) \, f(x)\mathrm{d}\mu(x)\,,\quad\forall\, f\in\mathcal{D}(N)\,, 
\,\forall\,\varphi\in \mathcal{E}(M)\,.
\end{align*}
where $(P\circledast\varphi)(f):=P(f\otimes\varphi)$ is well-defined since, in this case, $\mathcal{E}(M)\equiv\mathcal{D}(M)$ on account of the compactness of $M$.  
	
Denoting with $\tau^{(1)}(x,z_1;\varphi)$, $\tau^{(2)}(x,z_1,z_2;\varphi)$ the integral kernels of the first- and second-order functional derivatives of $\tau\in\mathcal{D}'(M;\operatorname{Pol})$, a direct application of Equation \eqref{Eq: functional derivative} yields
	\begin{align*}
		[\Phi^2]^{(1)}(x,z;\varphi)
		&=2\varphi(x)\delta_{\mathrm{Diag}_2}(x,z)\,,\qquad
		[\Phi P\circledast\Phi]^{(1)}(x,z;\varphi)
		=(P\circledast\varphi)(x)\delta_{\mathrm{Diag}_2}(x,z) 
		+\varphi(x) P(x,z)\,,\\
		[\Phi^2]^{(2)}(x,z_1,z_2;\varphi)
		&=2\delta_{\mathrm{Diag}_3}(x,z_1,z_2)\,,\qquad
		[\Phi P\circledast\Phi]^{(2)}(x,z_1,z_2;\varphi)
		=2\delta_{\mathrm{Diag}_2}(x,z_1)P(x,z_2)\,,
%		\equiv 2P(z_1,z_2)\,,
	\end{align*}
where $\delta_{\mathrm{Diag}_n}(x,z_1,\dots,z_{n-1})$, $n\geq 2$, is the integral kernel of the Dirac-delta distribution on $M^n$.
\end{example}

\begin{remark}\label{Rmk: Pol isomorphic space}
	Observe that, starting from Definition \ref{Def: functional-valued distributions}, for every $\tau\in\mathcal{D}'(N;\operatorname{Pol})$, $f\in\mathcal{D}(N)$ and $\varphi\in \mathcal{E}(M)$, it holds that, for $\lambda\in\mathbb{C}$,
	\begin{align}\label{Eq: rescaling of a polynomial}
		\tau(f;\lambda\varphi)
		=\sum_{k=0}^\infty\frac{1}{k!}\frac{\mathrm{d}^k}{\mathrm{d}\mu^k}\tau(f;\mu\varphi)\bigg|_{\mu=0}\lambda^k
		=\sum_{k=0}^\infty\frac{1}{k!}\tau^{(k)}(f\otimes\varphi^{\otimes k};0)\lambda^k\,,
	\end{align}
	where $\varphi^{\otimes k}=\underbrace{\varphi\otimes\dots\otimes\varphi}_k$, while
	\begin{align*}
		\frac{\mathrm{d}^k}{\mathrm{d}\mu^k}\tau(f;\mu\varphi)\bigg|_{\mu=0}
		&=\frac{\mathrm{d}^k}{\mathrm{d}\mu^k}\tau(f;(\mu+\mu_1+\ldots+\mu_k)\varphi)\bigg|_{\mu=\mu_1=\ldots=\mu_k=0}
		\\&=\frac{\partial^k}{\partial\mu_1\cdots\partial\mu_k}\tau(f;(\mu_1+\ldots+\mu_k)\varphi)\bigg|_{\mu_1=\ldots=\mu_k=0}
		=\tau^{(k)}(f\otimes\varphi^{\otimes k};0)\,.
	\end{align*}
	Observe that the right hand side of Equation \eqref{Eq: rescaling of a polynomial} is well-defined since, being $\tau$ a polynomial, only a finite number of derivatives are non-vanishing. Denoting with $\tau^{(k)}(x,z_1,\ldots,z_k)$ the integral kernel of the distribution $\tau^{(k)}(\cdot\,;0)$ it descends
	\begin{align*}
		\tau(f;\varphi)
		=\sum_{k=0}^\infty\frac{1}{k!}\int\limits_{N\times M^k} \tau^{(k)}(x,z_1,\ldots,z_k)f(x)\varphi(z_1)\cdots\varphi(z_k)
		\mathrm{d}\mu_N(x)\mathrm{d}\mu_M(z_1)\cdots\mathrm{d}\mu_M(z_k)\,.
	\end{align*}
	This realizes an isomorphism of topological vector spaces
	\begin{align*}
		\mathcal{D}'(N;\operatorname{Pol})
		\simeq\bigoplus_{k\geq 0}\mathcal{D}'(N\times M^k)\,.
	\end{align*}
\end{remark}

\noindent In the forthcoming analysis we will exploit extensively the properties of the following distinguished subspaces of $\mathcal{D}'(N;\operatorname{Pol})$. Before moving to the actual definition, we need some preparatory notation. For any but fixed $k\in\mathbb{N}$, we denote by $I_1\uplus\ldots\uplus I_\ell$ an arbitrary partition of $\{1,\ldots,k\}$ into $\ell$ disjoint non-empty subsets and we call $|I_i|$, $i=1,\dots,\ell$, the cardinality of the set $I_i$. In addition we indicate with $\delta_{\mathrm{Diag}_{|I_i|}}$ the Dirac delta distribution supported on the submanifold $\mathrm{Diag}_{|I_i|}=\{(x,\dots,x)\in M^{|I_i|}\;|\;x\in M\}$ and we adopt the convention that, if the cardinality of $|I_i|=1$ then $\mathrm{Diag}_{|I_i|}=\emptyset$.

\begin{definition}\label{Def: Dc functionals}
	For $k\in\mathbb{N}$ we call
	\begin{align}\label{Eq: WF-compatible functional-valued distributions}
		\mathcal{D}_{\mathrm{C}}'(M^k;\operatorname{Pol})
		:=\{\tau\in\mathcal{D}'(M^k;\operatorname{Pol})\,|\,
		\operatorname{WF}(\tau^{(n)})\subseteq \mathrm{C}_{k+n}\,\forall n\geq 0\}\,,
	\end{align}
	where adopting the notation $\widehat{x}_{k}=(x_1,\dots,x_k)\in M^k$,
	\begin{multline}\label{Eq: C-set}
		\mathrm{C}_{k}
		:=\{(\widehat{x}_{k},\widehat{\xi}_{k})\in T^*M^{k}\setminus\{0\}\,|\,\\
		\exists\ell\in\{1,\ldots,k-1\}\,,\,
		\{1\ldots,k\}=I_1\uplus\ldots\uplus I_\ell\,,\;\textrm{such that}\\
		\forall i\neq j\,,\,
		\forall (a,b)\in I_i\times I_j\,,\;\textrm{then}\;
		x_a\neq x_b\,,\\
		\textrm{and}\;\forall j\in\{1,\ldots,\ell\}\,,\,
		(\widehat{x}_{I_j},\widehat{\xi}_{I_j})\in\operatorname{WF}(\delta_{\mathrm{Diag}_{|I_j|}})
		\}\,,
	\end{multline}
where $(\widehat{x}_{I_j},\widehat{\xi}_{I_j})=(x_i,\xi_i)_{i\in I_j}\in T^*M^{|I_j|}$ for all $j\in\{1,\ldots,\ell\}$ and where $\operatorname{WF}$ stands for the wavefront set -- \textit{cf.} Appendix \ref{App: Wave Front Set kurzgesagt}.
\end{definition}
	
%	  -- where $|I|$ denotes the cardinality of $I$.
%	Recalling equation \eqref{Eq: WF of Dirac delta on total diagonal} we have that $(\widehat{x}_{I_j},\widehat{\xi}_{I_j})\in\operatorname{WF}(\delta_{\mathrm{Diag}_{|I_j|}})$ entails $\widehat{x}_{I_j}\in\mathrm{Diag}_{|I_j|}$ as well as $\sum_{i\in I_j}\xi_i=0$.
	
%Notice that the conditions $\widehat{x}_{I_j}\in\mathrm{Diag}_{I_j}$ for all $j\in\{1,\ldots,\ell\}$ and $x_a\neq x_b$ for all $a\in I_i\neq I_j\ni b$ entail that there exists a unique partition $I_1,\ldots,I_\ell$ with the properties as above.
%Moreover, for $k=0$ we find $\mathrm{C}_0=\emptyset$. \clacomment{lo $0$ non ci azzecca nella definizione e non si capiscono le stringhe con i cappucci.}

\begin{remark}\label{Rmk: Dc smoothness and functional derivative}
	Notice in particular that $\tau\in\mathcal{D}_{\mathrm{C}}'(M;\operatorname{Pol})$ is a (functional-valued) distribution generated by a smooth function. As a matter of fact, setting $k=1$ in Equation \eqref{Eq: C-set}, $C_1=\emptyset$.
\end{remark}
\begin{remark}\label{Rmk: additivity property of Ck}
	By direct inspection Equation \eqref{Eq: C-set} entails that $\mathrm{C}_{k_1}\otimes \mathrm{C}_{k_2}\subseteq \mathrm{C}_{k_1+k_2+1}$.
	Moreover, the following property also holds true:
	\begin{align}
		(x,\widehat{u}_{k_1},\xi_1,\widehat{\nu}_{k_1})\in\mathrm{C}_{k_1}\,,\,
		(x,\widehat{u}_{k_2},\xi_2,\widehat{\nu}_{k_2})\in\mathrm{C}_{k_2}\Longrightarrow
		(x,\widehat{u}_{k_1},\widehat{u}_{k_2},\xi_1+\xi_2,\widehat{\nu}_{k_1},\widehat{\nu}_{k_2})\in\mathrm{C}_{k_1+k_2}\,.
	\end{align}
\end{remark}

\begin{remark}\label{Rmk: Dc tensor product}
	Notice that $\mathcal{D}_{\mathrm{C}}'(M;\operatorname{Pol})$ is stable under $\delta_\psi$ for all $\psi\in \mathcal{E}(M)$ -- \textit{cf.} Equation \eqref{Eq: functional differential}.
	
	Moreover, given $n_1,\ldots,n_m\in\mathbb{N}$ and $\tau_j\in\mathcal{D}'(M^{n_j};\operatorname{Pol})$ for all $j\in\{1,\ldots,m\}$ we define the functional-valued tensor product distribution $\tau_1\otimes\ldots\otimes \tau_m\in\mathcal{D}'(M^n;\operatorname{Pol})$, where $n=n_1+\ldots+n_m$, by
	\begin{align}\label{Eq: Dc tensor product}
		(\tau_1\otimes\ldots\otimes \tau_m)(f_1\otimes\ldots\otimes f_m;\varphi)
		=\tau_1(f_1;\varphi)\cdots \tau_m(f_m;\varphi)\,,
	\end{align}
	for all $f_j\in\mathcal{D}(M^{n_j})$, $j\in\{1,\ldots,m\}$, and $\varphi\in \mathcal{E}(M)$.
	A direct computation shows that, for all $k\geq 0$,
	\begin{align*}
		\left(\tau_1\otimes\ldots\otimes \tau_m\right)^{(k)}
		=\sum_{\substack{k_1,\dots,k_m\\ k_1+\ldots+k_m=k}}\tau_1^{(k_1)}\otimes\ldots\otimes \tau_m^{(k_n)}\,.
	\end{align*}
	As a consequence, if $\tau_j\in\mathcal{D}_{\mathrm{C}}'(M;\operatorname{Pol})$ for all $j\in\{1,\ldots,m\}$ then $\tau_1\otimes\ldots\otimes \tau_m\in\mathcal{D}_{\mathrm{C}}'(M^n;\operatorname{Pol})$.
\end{remark}
%\begin{example}
%	The functional-valued distributions $\Phi^2,\PhiP\circledast\Phi$ introduced in example \ref{Ex: examples of functional-valued polynomial distributions} are in fact elements of $\mathcal{D}_{\mathrm{C}}'(M;\operatorname{Pol})$.
%\end{example}

%The model we shall construct in the forthcoming section will heavily rely on the above notion of functional-valued $\varphi$-polynomial distributions.
%In fact, our plan is to define, as a first step, a linear map $\Pi\colon\mathcal{A}\to\mathcal{D}_{\mathrm{C}}'(M;\operatorname{Pol})$ -- \textit{cf.} definition \ref{Def: pointwise algebra} and theorem \ref{Thm: Gamma cdotQ existence}.
\noindent For later convenience we discuss the following result. As stressed before, we make use of the concept of scaling degree associated to distributions which is outlined in Appendix \ref{App: Scaling degree}.

\begin{lemma}\label{Lem: if it holds for tau it holds for Ptau}
	Let $\tau\in\mathcal{D}_{\mathrm{C}}^\prime(M;\operatorname{Pol})$.
	Then $P\circledast\tau\in\mathcal{D}_{\mathrm{C}}'(M;\operatorname{Pol})$ where
	\begin{align*}
		[P\circledast \tau](f;\varphi)
		:=\tau(P\circledast f;\varphi)\,,
		\qquad\forall\, f\in\mathcal{D}(M)\,,\,
		\forall\,\varphi\in \mathcal{E}(M)\,,
	\end{align*}
	where $P\circledast f\in\mathcal{D}'(M)$ is defined by $(P\circledast f)(h):=P(h\otimes f)$.
%	 indicates the partial evaluation of $P\in\mathcal{D}^\prime(M\times M)$ against $f\in\mathcal{D}(M)$.
	Moreover, if $\operatorname{sd}_{\mathrm{Diag}_{p+1}}(\tau)^{(p)}<\infty$, for $p\in\mathbb{N}$, it holds
	\begin{align}\label{Eq: P convolution functional scaling degree}
	\operatorname{sd}_{\mathrm{Diag}_{p+1}}(P\circledast\tau)^{(p)}
	<\infty\,,
	\end{align}
	where $\operatorname{sd}_{\mathrm{Diag}_{p+1}}$ is the scaling degree along the submanifold $\mathrm{Diag}_{p+1}=\{(x,\dots,x)\in M^{p+1}\;|\;x\in M\}$, {\it cf.} Appendix \ref{App: Scaling degree}.
\end{lemma}
\begin{proof}
	%Since $\tau\in\mathcal{D}_{\mathrm{C}}'(M;\operatorname{Pol})$ is a (functional-valued) distribution generated by a smooth function, {\it cf.} Remark \ref{Rmk: convolution with smooth function} ensures that 
	Since $\operatorname{WF}(P)=\operatorname{WF}(\delta_{\mathrm{Diag}_2})$ and since $\tau$ is generated by a smooth function, it descends that $P\circledast \tau\in\mathcal{D}'(M;\operatorname{Fun})$ is well-defined.
	Moreover, assuming without loss of generality, $\psi_1=\ldots=\psi_k=\psi\in \mathcal{E}(M)$ in Equation \eqref{Eq: functional derivative}, it descends
	\begin{align*}
	(P\circledast \tau)^{(k)}(f\otimes\psi^{\otimes k};\varphi)
	=\tau^{(k)}((P\circledast f)\otimes\psi^{\otimes k};\varphi)
	=[(P\otimes \delta_{\mathrm{Diag}_2}^{\otimes k})\circledast \tau^{(k)}](f\otimes\psi^{\otimes k};\varphi)\,,
	\end{align*}
	which guarantees that $P\circledast \tau\in\mathcal{D}'(M;\operatorname{Pol})$ since $(P\otimes \delta_{\mathrm{Diag}_2}^{\otimes k})$ is a compactly supported distribution.
	%-- here $\circ$ denotes convolution of distributions
	%\footnote{
	%	To avoid confusion, the convolution is meant between the distribution $(P\otimes \delta_{\mathrm{Diag}_2}^{\otimes k})$ with integral kernel $P(x_1,x_2)\prod_{\ell=1}^k\delta_{\mathrm{Diag}_2}(z_\ell,y_\ell)$ and the distribution $\tau^{(k)}$ with integral kernel $\tau^{(k)}(x_2,\widehat{y}_k)$.
	%}.
	Furthermore Equations \eqref{Eq: WF of Dirac delta on total diagonal}, and \eqref{Eq: WF of parametrix} together with item \textit{5}. of Theorem \ref{Thm: WF results} entail
	\begin{align*}
	\operatorname{WF}(P\otimes\delta_{\mathrm{Diag}_2}^{\otimes k})=
%	&=\{
%	(x_1,x_2,\widehat{z}_k,\widehat{y}_k,\xi_1,\xi_2,\widehat{\zeta}_k,\widehat{\eta}_k)\in T^*M^{2+2k}\setminus\{0\}\,|\,
%	\\&(x_1,x_2,\xi_1,\xi_2)\in\operatorname{WF}(\delta_{\mathrm{Diag}_2})\,,\,
%	(\widehat{z}_k,\widehat{y}_k,\widehat{\zeta}_k,\widehat{\eta}_k)\in\operatorname{WF}(\delta_{\mathrm{Diag}_2}^{\otimes k})
%	\}
%	\\&=
    \{(x,x,\widehat{z}_k,\widehat{z}_k,\xi,-\xi,\widehat{\zeta}_k,-\widehat{\zeta}_k)\in T^*M^{2+2k}\setminus\{0\}
	\}\,.
	\end{align*}
	Equation \eqref{Eq: WF 2-component projection} yields
	\begin{multline*}
	\operatorname{WF}_2(P\otimes \delta_{\mathrm{Diag}_2}^{\otimes k})
	=\{
	(x_2,\widehat{y}_k,\xi_2,\widehat{\eta}_k)\in T^*M^{1+k}\setminus\{0\}\,|\,
	\\\exists x_1\in M\,,\,\widehat{z}_k\in M^k\,,
	(x_1,x_2,\widehat{z}_k,\widehat{y}_k,0,\xi_2,0,\widehat{\eta}_k)\in\operatorname{WF}(P\otimes \delta_{\mathrm{Diag}_2}^{\otimes k})
	\}
	=\emptyset\,,
	\end{multline*}
	while $\operatorname{WF}_1(P\otimes \delta_{\mathrm{Diag}_2}^{\otimes k})=\emptyset$. This entails that the condition spelled in Equation \eqref{Eq: WF convolution condition} is satisfied so that $(P\otimes\delta_{\mathrm{Diag}_2}^{\otimes k})\circledast \tau^{(k)}$ is a well-defined functional-valued distribution.
	Moreover, Theorem \ref{Thm: WF results} and Equation \eqref{Eq: WF convolution bound} imply
	\begin{align*}
	\operatorname{WF}((P\otimes \delta_{\mathrm{Diag}_2}^{\otimes k})\circledast \tau^{(k)})
	&\subseteq\operatorname{WF}(P\otimes \delta_{\mathrm{Diag}_2}^{\otimes k})\circ\operatorname{WF}(\tau^{(k)})
	\\&=\{
	(x_1,\widehat{z}_k,\xi_1,\widehat{\zeta}_k)\in T^*M^{1+k}\setminus\{0\}\,|\,
	\\&\exists(x_2,\widehat{y}_k,\xi_2,\widehat{\eta}_k)\in\operatorname{WF}(\tau^{(k)})\,,\,
	(x_1,x_2,\widehat{z}_k,\widehat{y}_k,\xi_1,\xi_2,\widehat{\zeta}_k,\widehat{\eta}_k)\in\operatorname{WF}(P\otimes \delta_{\mathrm{Diag}_2}^{\otimes k})
	\}
	\\&\subseteq\mathrm{C}_{k+1}\,,
	\end{align*}
	where we used the bound $\operatorname{WF}(\tau^{(k)})\subseteq\mathrm{C}_{k+1}$ as well as the explicit form of $\operatorname{WF}(P\circledast\delta_{\mathrm{Diag}_2}^{\otimes k})$.
	
	To conclude, Equation \eqref{Eq: P convolution functional scaling degree} follows from Lemma \ref{Lemma: finite sd of convolution} applied to $P\circledast\tau^{(k)}$ which entails
	\begin{align*}
	\operatorname{sd}_{\mathrm{Diag}_{k+1}}((P\circledast \tau)^{(k)})
	<\infty\,.
	\end{align*}
\end{proof}

%\begin{remark}[Parabolic Case] 
%This result holds also in this scenario, the only difference being that we have to consider the effective dimension of the manifold $M$ in Equation \eqref{Eq: P convolution functional scaling degree}, see Remark \ref{Rmk: weighted scaling degree}. This is a consequence of $\operatorname{wsd}_{\mathrm{Diag}_2}(\delta)=d+1$. 
%\end{remark}

In order to connect the above structures to a regularized version of white noise, we need to equip $\mathcal{D}_{\mathrm{C}}'(M;\operatorname{Pol})$ with a suitable algebraic structure.
In the following with $\circ$ we indicate the composition of distributions as per \cite[Sec. 8.2]{Hormander-I-03}, see also item {\em 4.} of Theorem \ref{Thm: WF results} while $\cdot$ stands for the product of distribution, {\it cf.} item {\em 2.} of Theorem \ref{Thm: WF results}.
Recall that $P$ and $P^*$ are the parametrices respectively of the underlying microhypoelliptic differential operator $E$ and of its formal adjoint $E^*$.

\begin{proposition}\label{Prop: Dc functionals algebraic structure}
	Let $M$ be a compact manifold and let $P_\varepsilon\in \mathcal{E}(M^2)$ be a family of smooth maps such that $\lim_{\varepsilon\to 0^+}P_\varepsilon=P$ in $\mathcal{D}'(M^2)$ and let $Q_\varepsilon:=P_\varepsilon\circ P^*_\varepsilon\in \mathcal{E}(M^2)$.
	For $\tau_1,\tau_2\in\mathcal{D}_{\mathrm{C}}'(M;\operatorname{Pol})$ let $\tau_1\cdot_{Q_\varepsilon} \tau_2\in\mathcal{D}'(M;\operatorname{Pol})$ be defined by	
	\begin{align}\label{Eq: regularized delta-product}
		(\tau_1\cdot_{Q_\varepsilon} \tau_2)(f;\varphi)
		&=\sum_{k\geq 0}\frac{1}{k!}[(\delta_{\mathrm{Diag}_2}\otimes Q_\varepsilon^{\otimes k})\cdot(\tau_1^{(k)}\widetilde{\otimes}\tau_2^{(k)})](f\otimes1_{1+2k};\varphi)\,,
	\end{align}
	for all $f\in\mathcal{D}(M)$ and $\varphi\in \mathcal{E}(M)$ while, for all $\ell\in\mathbb{N}$, $1_\ell\in \mathcal{E}(M^\ell)$ denotes the function $1_\ell(\widehat{x}_\ell)=1$.	Here $\widetilde{\otimes}$ is defined as in Equation \eqref{Eq: tilde tensor product}. Then $\tau_1\cdot_{Q_\varepsilon} \tau_2\in\mathcal{D}'_{\mathrm{C}}(M;\operatorname{Pol})$ and  $(\mathcal{D}_{\mathrm{C}}'(M;\operatorname{Pol}),\cdot_{Q_\varepsilon})$ is a unital commutative and associative algebra.
\end{proposition}
\begin{proof}
	We divide the proof in three steps. We first show that Equation \eqref{Eq: regularized delta-product} defines a functional-valued distribution $\tau_1\cdot_{Q_\varepsilon}\tau_2\in\mathcal{D}'(M;\operatorname{Pol})$. In the second part we check the properties of the wavefront set showing that $\tau_1\cdot_{Q_\varepsilon}\tau_2\in \mathcal{D}_{\mathrm{C}}'(M;\operatorname{Pol})$. Finally we discuss the ensuing algebraic structure.
	
	\paragraph{{\bf Proof that $\bm{\tau_1\cdot_{Q_\varepsilon}\tau_2\in\mathcal{D}'(M;}\operatorname{{\bf Pol}}\bm{)}$.}}
	Since $\tau_1,\tau_2\in\mathcal{D}_{\mathrm{C}}'(M;\operatorname{Pol})$, the right hand side of Equation \eqref{Eq: regularized delta-product} includes only a finite number of non-vanishing terms. It suffices to show that each of these terms is well-defined.
	Let $k\geq 0$ and consider 
	\begin{align*}
		[(\delta_{\mathrm{Diag}_2}\otimes Q_\varepsilon^{\otimes k})\cdot(\tau_1^{(k)}\widetilde{\otimes} \tau_2^{(k)})]\,.
	\end{align*}
	On account of Equation \eqref{Eq: WF of Dirac delta on total diagonal} and since $Q_\varepsilon\in \mathcal{E}(M^2)$, it holds
	\begin{align}\label{Eq: WF of delta times Qepsilon}
		\operatorname{WF}(\delta_{\mathrm{Diag}_2}\otimes Q_\varepsilon^{\otimes k})
		&
		=\{(x,x,\widehat{z}_k,\widehat{y}_k,\xi,-\xi,0,0)\in T^*M^{2+2k}\setminus\{0\}\}\,,
	\end{align}
	where $\widehat{z}_k,\widehat{y}_k\in M^k$ while $(x,\xi)\in T^*M$. In addition Remark \ref{Rmk: Dc tensor product} entails that $\tau_1^{(k)}\widetilde{\otimes}\tau_2^{(k)}\in\mathcal{D}_{\mathrm{C}}'(M^{2+2k};\operatorname{Pol})$ and
	\begin{align*}
		\operatorname{WF}(\tau_1^{(k)}\widetilde{\otimes}\tau_2^{(k)})
		&\subseteq
		\{(x_1,x_2,\widehat{z}_k,\widehat{y}_k,\xi_1,\xi_2,\widehat{\zeta}_k,\widehat{\eta}_k)\in T^*M^{2+2k}\setminus\{0\}\,|\,
		(x_1,\widehat{z}_k,\xi_1,\widehat{\zeta}_k),(x_2,\widehat{y}_k,\xi_2,\widehat{\eta}_k)\in\mathrm{C}_k
		\}\,.
%		\\&=\mathrm{C}_k\otimes\mathrm{C}_k\,.
	\end{align*}
	As a consequence Equation \eqref{Eq: WF product condition} holds true since 
	\begin{multline*}
		(x_1,x_2,\widehat{z}_k,\widehat{y}_k,0,0,0,0)\notin
		\{
		(x_1,x_2,\widehat{z}_k,\widehat{y}_k,\xi_1+\xi_1',\xi_2+\xi_2',\widehat{\zeta}_k+\widehat{\zeta}_k',\widehat{\eta}_k+\widehat{\eta}_k')\in T^*M^{2+2k}\,|\,
		\\(x_1,x_2,\widehat{z}_k,\widehat{\eta}_k,\xi_1,\xi_2,\widehat{\zeta}_k,\widehat{\eta}_k)\in\operatorname{WF}(\delta_{\mathrm{Diag}_2}\otimes Q_\varepsilon^{\otimes k})\,,\,
		\\(x_1,\widehat{z}_k,\xi_1',\hat{\zeta}_k')\in\operatorname{WF}(\tau_1^{(k)})\,,\,
		(x_2,\widehat{y_k},\xi_2',\widehat{\eta}_k')\in\operatorname{WF}(\tau_2^{(k)})
		\}\,.
	\end{multline*}

\noindent Hence item {\em 2.} of Theorem \ref{Thm: WF results} implies that	
	\begin{align*}
		[(\delta_{\mathrm{Diag}_2}\otimes Q_\varepsilon^{\otimes k})\cdot(\tau_1^{(k)}\widetilde{\otimes} \tau_2^{(k)})]\in\mathcal{D}'(M^{2+2k};\operatorname{Fun})\,.
	\end{align*}
	Moreover, for all $\ell\in\mathbb{N}$ it holds
	\begin{align*}
		[(\delta_{\mathrm{Diag}_2}\otimes Q_\varepsilon^{\otimes k})\cdot(\tau_1^{(k)}\widetilde{\otimes} \tau_2^{(k)})]^{(\ell)}
		=\sum\limits_{\substack{\ell_1,\ell_2\\\ell_1+\ell_2=\ell}}[(\delta_{\mathrm{Diag}_2}\otimes Q_\varepsilon^{\otimes k}\otimes 1_{\ell})\cdot(\tau_1^{(k+\ell_1)}\widetilde{\otimes}\tau_2^{(k+\ell_2)})]\,,
	\end{align*}
	where analogously to the preceding case and taking into account that both $\tau_1$ and $\tau_2$ are polynomial functionals it descends that $[(\delta_{\mathrm{Diag}_2}\otimes Q_\varepsilon^{\otimes k}\otimes 1_\ell)\cdot (\tau_1^{(k+\ell_1)}\widetilde{\otimes}\tau_2^{(k+\ell_2)})]\in\mathcal{D}^\prime(M^{2+2k+\ell};\operatorname{Pol})$.	
	
\paragraph{Proof that $\bm{\tau_1\cdot_{Q_\varepsilon}\tau_2\in\mathcal{D}_{\mathrm{C}}'(M;}\operatorname{{\bf Pol}}\bm{)}$.}
	We need to show that, for all $\ell\geq 0$ it holds
	\begin{align*}
		\operatorname{WF}([\tau_1\cdot_{Q_\varepsilon} \tau_2]^{(\ell)})\subseteq\mathrm{C}_{\ell+1}\,.
	\end{align*}
	A direct computation yields
	\begin{align*}
		[\tau_1\cdot_{Q_\varepsilon} \tau_2]^{(\ell)}(f\otimes\psi_\ell;\varphi)
		&=\sum_{k\geq 0}\frac{1}{k!}\sum_{\ell_1+\ell_2=\ell}
		[(\delta_{\mathrm{Diag}_2}\otimes Q_\varepsilon^{\otimes k}\otimes 1_\ell)\cdot (\tau_1^{(k+\ell_1)}\widetilde{\otimes} \tau_2^{(k+\ell_2)})]
		(f\otimes 1_{1+2k}\otimes\psi_\ell;\varphi)
		\\&=\sum_{k\geq 0}\frac{1}{k!}\sum_{\ell_1+\ell_2=\ell}
		T_{\ell_1,\ell_2}(f\otimes\psi_\ell;\varphi)\,.
	\end{align*}
	where $f\in\mathcal{D}(M)$, $\psi_\ell\in \mathcal{E}(M^\ell)$ while $\varphi\in \mathcal{E}(M)$.	
%	We shall prove the wave front set bound for each term $T_{\ell_1,\ell_2}$ in the above sum separately.
	Let $\ell_1,\ell_2\in\mathbb{N}\cup\{0\}$ be such that $\ell_1+\ell_2=\ell$ and consider
	\begin{align*}
		[(\delta_{\mathrm{Diag}_2}\otimes Q_\varepsilon^{\otimes k}\otimes 1_\ell)\cdot (\tau_1^{(k+\ell_1)}\widetilde{\otimes} \tau_2^{(k+\ell_2)})]\,,
	\end{align*}
whose wavefront set reads, \textit{cf.} Equation \eqref{Eq: WF of delta times Qepsilon},

	\begin{multline*}
		\operatorname{WF}([(\delta_{\mathrm{Diag}_2}\otimes Q_\varepsilon^{\otimes k}\otimes 1_\ell)\cdot(\tau_1^{(k+\ell_1)}\widetilde{\otimes}\tau_2^{(k+\ell_2)})])
		\subseteq
		\\\{(x,x,\widehat{z}_k,\widehat{y}_k,\widehat{u}_\ell,\xi_1+\xi_1',-\xi_1+\xi_2',\widehat{\zeta}_k,\widehat{\eta}_k,\widehat{\nu}_\ell)\in T^*M^{2+2k+\ell}\setminus\{0\}\,|\,
		\\(x,\widehat{z}_k,\widehat{u}_{\ell_1},\xi_1',\widehat{\zeta}_k,\widehat{\nu}_{\ell_1})\in\mathrm{C}_{k+\ell_1}\,,\,
		(x,\widehat{y}_k,\widehat{u}_{\ell_2},\xi_2',\widehat{\eta}_k,\widehat{\nu}_{\ell_2})\in\mathrm{C}_{k+\ell_2}
		\}\,.
	\end{multline*}
	For our purposes we need to consider the functional-valued distribution obtained from $[(\delta_{\mathrm{Diag}_2}\otimes Q_\varepsilon^{\otimes k}\otimes 1_\ell)\cdot(\tau_1^{(k+\ell_1)}\widetilde{\otimes}\tau_2^{(k+\ell_2)})]$ after its partial evaluation against the constant function $1_{1+2k}$:
	\begin{align*}
		T_{\ell_1,\ell_2}(f\otimes\psi_\ell;\varphi)
		:=[(\delta_{\mathrm{Diag}_2}\otimes Q_\varepsilon^{\otimes k}\otimes 1_\ell)\cdot(\tau_1^{(k+\ell_1)}\otimes \tau_2^{(k+\ell_2)})](f\otimes 1_{1+2k}\otimes\psi_\ell;\varphi)\,.
	\end{align*}
On account of Remark \ref{Rmk: convolution with smooth function} it holds
	\begin{multline*}
		\operatorname{WF}(T_{\ell_1,\ell_2})
		\subseteq\{
		(x_1,\widehat{u}_\ell,\xi_1,\widehat{\nu}_\ell)\in T^*M^{1+\ell}\setminus\{0\}\;|\;
		\exists x_2\in M\;\textrm{and}\;
		\exists\widehat{z}_k,\widehat{y}_k\in M^k\,,\,
		\\(x_1,x_2,\widehat{z}_k,\widehat{y}_k,\widehat{u}_\ell,\xi_1,0,0,0,\widehat{\nu}_\ell)
		\in\operatorname{WF}((\delta_{\mathrm{Diag}_2}\otimes Q_\varepsilon^{\otimes k})\cdot(\tau_1^{(k+\ell_1)}\widetilde{\otimes}\tau_2^{(k+\ell_2)}))
		\}
		\subseteq\mathrm{C}_{\ell}\,,
	\end{multline*}
	where the last inclusion follows from the bound on the wave front set of $(\delta_{\mathrm{Diag}_2}\otimes Q_\varepsilon^{\otimes k}\otimes 1_\ell)\cdot(\tau_1^{(k+\ell_1)}\widetilde{\otimes}\tau_2^{(k+\ell_2)})$ as well as from Equation \eqref{Eq: C-set} and Remark \ref{Rmk: additivity property of Ck}. This established that $\tau_1\cdot_{Q_\varepsilon}\tau_2\in\mathcal{D}_{\mathrm{C}}'(M;\operatorname{Pol})$.

	\paragraph{Algebraic properties of $\bm{\mathcal{D}_{\mathrm{C}}'(M;}\operatorname{{\bf Pol}}\bm{)}$.}
	To conclude the proof we need to show that $\cdot_{Q_\varepsilon}$ induces a unital and commutative algebra structure on $\mathcal{D}^\prime_{\mathrm{C}}(M;\operatorname{Pol})$. In the preceding steps we have shown that
	\begin{align*}
		\cdot_{Q_\varepsilon}\colon
		\mathcal{D}_{\mathrm{C}}'(M;\operatorname{Pol})\times\mathcal{D}_{\mathrm{C}}'(M;\operatorname{Pol})
		\to\mathcal{D}_{\mathrm{C}}'(M;\operatorname{Pol})\,,
	\end{align*}
	which entails that $(\mathcal{D}^\prime_{\mathrm{C}}(M;\operatorname{Pol}),\cdot_{Q_\varepsilon})$ is an algebra. It is unital, since, calling $\boldsymbol{1}(f;\varphi):=1(f)$, then, $\tau\cdot_{Q_\varepsilon}\boldsymbol{1} =\boldsymbol{1}\cdot_{Q_\varepsilon} \tau=\tau$ for all $\tau\in\mathcal{D}_{\mathrm{C}}'(M;\operatorname{Pol})$. Furthermore Equation \eqref{Eq: regularized delta-product} is clearly symmetric in $\tau_1$ and $\tau_2$, since $Q_\varepsilon=P_\varepsilon\circ P^*_\varepsilon$. It remains to be shown that $\cdot_{Q_\varepsilon}$ is associative.
	To this end, we introduce the following notation:
	\begin{enumerate}[(i)]
		\item 
		for $\tau_1,\tau_2\in\mathcal{D}_{\mathrm{C}}'(M;\operatorname{Pol})$ we denote by $\tau_1\widehat{\otimes}\tau_2\in\mathcal{D}'(M^2;\operatorname{Pol}^{\otimes 2})$ the functional-valued distribution
		\begin{align*}
			[\tau_1\widehat{\otimes}\tau_2](f_1\otimes f_2;\varphi_1\otimes\varphi_2)
			:=\tau_1(f_1;\varphi_1)\tau_2(f_2;\varphi_2)\,,
		\end{align*}
		for all $f_1,f_2\in\mathcal{D}(M)$ and $\varphi_1,\varphi_2\in\mathcal{E}(M)$. 
		%-- notice that $\tau_1\widehat{\otimes}\tau_2$ is now a functional on pairs of configurations $\varphi_1,\varphi_2$.
		\item
		We call $\mathsf{m}\colon\mathcal{D}'(M^2;\operatorname{Pol}^{\otimes 2})\to\mathcal{D}'(M^2;\operatorname{Pol})$
		\begin{align*}
			\mathsf{m}(\tau_1\widehat{\otimes}\tau_2):=\tau_1\otimes \tau_2\,,
		\end{align*} 
		where $\tau_1\otimes \tau_2\in\mathcal{D}'(M^2;\operatorname{Pol})$ has been defined in Remark \eqref{Rmk: Dc tensor product}.
		\item
		We call $\Upsilon_{Q_\varepsilon}\colon\mathcal{D}_{\mathrm{C}}'(M^2;\operatorname{Pol}^{\otimes 2})\to\mathcal{D}_{\mathrm{C}}'(M^2;\operatorname{Pol}^{\otimes 2})$ the linear map
		\begin{align*}
			[\Upsilon_{Q_\varepsilon}(\tau_1\widehat{\otimes}\tau_2)](f_1\otimes f_2;\varphi_1\otimes\varphi_2)
			:=[(1_2\otimes Q_\varepsilon)\cdot(\tau_1^{(1)}\widehat{\otimes} \tau_2^{(1)})](f_1\otimes f_2\otimes 1_2;\varphi_1\otimes\varphi_2)\,,
		\end{align*}
		for all $f_1,f_2\in\mathcal{D}(M)$ and $\varphi_1,\varphi_2\in \mathcal{E}(M)$.
		%, being $\cdot$ the product of distributions
		%\footnote{
		%	To avoid confusion we mean the product between the distribution	$1_2\otimes Q_\varepsilon$ with integral kernel $1_2(x_1,x_2)Q_\varepsilon(z,y)$ with the (functional-valued) distribution $\tau_1^{(1)}\widehat{\otimes} \tau_2^{(1)}$ with integral kernel $\tau_1^{(1)}(x_1,z) \tau_2^{(1)}(x_2,y)$.}.
	\end{enumerate}
	Following these notations, Equation \eqref{Eq: regularized delta-product} can be written as
	\begin{align*}
		[\tau_1\cdot_{Q_\varepsilon}\tau_2](f;\varphi)
		=[\delta_{\mathrm{Diag}_2}\cdot\mathsf{m}\circ\exp[\Upsilon_{Q_\varepsilon}](\tau_1\widehat{\otimes}\tau_2)](f\otimes 1;\varphi)\,,
	\end{align*}
	where $\circ$ stands here for the composition of maps while $\exp[\Upsilon_{Q_\varepsilon}]=\sum_{k\geq 0}\frac{1}{k!}\Upsilon_{Q_\varepsilon}^{k}$. Observe that only a finite number of terms in the sum are non vanishing since we consider polynomial functionals. If $\tau_1,\tau_2,\tau_3\in\mathcal{D}_{\mathrm{C}}'(M;\operatorname{Pol})$ it descends, for all $f\in\mathcal{D}(M)$,
	\begin{align*}
		[(\tau_1\cdot_{Q_\varepsilon} \tau_2)\cdot_{Q_\varepsilon}\tau_3](f;\varphi)
		&=[\delta_{\mathrm{Diag}_2}\cdot\mathsf{m}\exp[\Upsilon_{Q_\varepsilon}]([\tau_1\cdot_{Q_\varepsilon}\tau_2]\widehat{\otimes}\tau_3)](f\otimes 1;\varphi)
		\\&=
		[\delta_{\mathrm{Diag}_2}\cdot\mathsf{m}\exp[\Upsilon_{Q_\varepsilon}]
		[(\delta_{\mathrm{Diag}_2}\cdot\mathsf{m}\exp[\Upsilon_{Q_\varepsilon}])\otimes\operatorname{Id}]
		(\tau_1\widehat{\otimes}\tau_2\widehat{\otimes}\tau_3)](f\otimes 1_2;\varphi)\,,
	\end{align*}
	where $\operatorname{Id}$ denotes the identity operator on $\mathcal{D}_{\mathrm{C}}'(M;\operatorname{Pol})$.
	Denoting with $\Upsilon_{12}$ the linear map acting as
	\begin{align*}
		\Upsilon_{12}(\tau_1\widehat{\otimes}\tau_2\widehat{\otimes}\tau_3)
		:=\Upsilon_{Q_\varepsilon}(\tau_1\widehat{\otimes}\tau_2)\widehat{\otimes}\tau_3\,,
	\end{align*}
	and similarly $\Upsilon_{13},\Upsilon_{23}$, it holds that, for all $f\in\mathcal{D}(M)$,
	\begin{align*}
		[(\tau_1\cdot_{Q_\varepsilon} \tau_2)\cdot_{Q_\varepsilon}\tau_3](f;\varphi)
		&=
		\delta_{\mathrm{Diag}_3}\cdot\mathsf{m}(\mathsf{m}\otimes\operatorname{Id})
		\exp[\Upsilon_{13}+\Upsilon_{23}]\exp[\Upsilon_{12}](\tau_1\widehat{\otimes}\tau_2\widehat{\otimes}\tau_3)(f\otimes 1_2;\varphi)
		\\&=
		\delta_{\mathrm{Diag}_3}\cdot\mathsf{m}(\operatorname{Id}\otimes\mathsf{m})
		\exp[\Upsilon_{12}+\Upsilon_{13}]\exp[\Upsilon_{23}](\tau_1\widehat{\otimes}\tau_2\widehat{\otimes}\tau_3)(f\otimes 1_2;\varphi)
		\\&=[\tau_1\cdot_{Q_\varepsilon}(\tau_2\cdot_{Q_\varepsilon}\tau_3)](f;\varphi)\,.
	\end{align*}
\end{proof}

\begin{remark}\label{Rmk: interpretation of Dc local product}
	Proposition \ref{Prop: Dc functionals algebraic structure} and Equation \eqref{Eq: regularized delta-product} in particular codify the expectation values of polynomial expressions in the (shifted) regularized random field $\widehat{\varphi}_\varepsilon(x)=P\circledast\widehat{\xi}_\varepsilon$ discussed in Section \ref{Sec: introduction}.
	For concreteness, let $\Phi\in\mathcal{D}_{\mathrm{C}}'(M;\operatorname{Pol})$ be
	\begin{align*}
		\Phi(f;\varphi)=\int_M f\varphi\mu\,.
	\end{align*}
	A direct application of Equation \eqref{Eq: regularized delta-product} yields
	\begin{align*}
		[\Phi\cdot_{Q_\varepsilon} \Phi](f;\varphi)
		=\int_M f(x)[\varphi^2(x)+Q_\varepsilon(x,x)]\mathrm{d}\mu(x)
		=\Phi^2(f;\varphi)
		+\int_M f(x)Q_\varepsilon(x,x)\mathrm{d}\mu(x)\,,
	\end{align*}
	where $\Phi^2$ has been defined in Example \ref{Ex: examples of functional-valued polynomial distributions}.
	The last expression coincides with the expectation value $\mathbb{E}(\widehat{\varphi}_\varepsilon^2(f))$ of the regularized random field $\widehat{\varphi}_\varepsilon(x)^2=(P\circledast\widehat{\xi}_\varepsilon+\varphi)^2(x)$ smeared against a test function $f\in\mathcal{D}(M)$.
\end{remark}

The product $\cdot_{Q_\varepsilon}$ captures the information on the expectation value of polynomial expressions in the (shifted) regularized random field $\widehat{\varphi}_\varepsilon$.
However, it does not contain any information about the correlations. Such datum can be codified via a different product structure on a suitable tensor algebra built out of $\mathcal{D}_{\mathrm{C}}'(M;\operatorname{Pol})$. The following proposition makes this statement precise and, since the proof is similar to that of Proposition \ref{Prop: Dc functionals algebraic structure}, we omit it.

\begin{proposition}\label{Prop: Dc tensor algebra algebraic structure}
	We call $\mathcal{T}_{\mathrm{C}}'(M;\operatorname{Pol})$ the vector space
	\begin{align}\label{Eq: Tc algebra}
		\mathcal{T}_{\mathrm{C}}'(M;\operatorname{Pol}):=
		\mathbb{C}\oplus
		\bigoplus_{n\geq 1}\mathcal{D}'_{\mathrm{C}}(M^n;\operatorname{Pol})\,.
	\end{align}
	This is a commutative and associative algebra if endowed with the product $\bullet_{Q_\varepsilon}$ which is completely and unambiguously specified as follows: For all $\tau_1\in\mathcal{D}'_{\mathrm{C}}(M^{n_1};\operatorname{Pol})$ and $\tau_2\in\mathcal{D}'_{\mathrm{C}}(M^{n_2};\operatorname{Pol})$ with $n_1,n_2\in\mathbb{N}\cup\{0\}$
	\begin{align}\label{Eq: regularized multilocal delta-product}
		(\tau_1\bullet_{Q_\varepsilon} \tau_2)(f_1\otimes f_2;\varphi)
		&=\sum_{k\geq 0}\frac{1}{k!}
		[(1_{n_1+n_2}\otimes Q_\varepsilon^{\otimes k})\cdot(\tau_1^{(k)}\widetilde{\otimes}\tau_2^{(k)})]
		(f_1\otimes f_2\otimes 1_{2k};\varphi)\,,
	\end{align}
	for all $f_1\in\mathcal{D}(M^{n_1})$, $f_2\in\mathcal{D}(M^{n_2})$ and $\varphi\in \mathcal{E}(M)$, where $\cdot$ denotes once more the product of distributions.
\end{proposition}

\begin{remark}\label{Rem: tilde tensor product disambiguation}
	Observe that in Equation \eqref{Eq: regularized multilocal delta-product} we are employing with a slight abuse of notation the symbol $\widetilde{\otimes}$ which first appeared in Equation \eqref{Eq: tilde tensor product}. Here it still indicates that a reordering of the underlying variables, namely the integral kernel of $(1_{n_1+n_2}\otimes Q_\varepsilon^{\otimes k})\cdot(\tau_1^{(k)}\widetilde{\otimes}\tau_2^{(k)})$ reads
	$$\prod\limits_{i=1}^kQ_\varepsilon(z_i,y_i)\tau_1^{(k)}(x_1,\dots,x_{n_1},z_1,\dots,z_k)\tau_2^{(k)}(x_{n_1+1},\dots,x_{n_2},y_1,\dots,y_k).$$
\end{remark}

\begin{remark}\label{Rmk: interpretation of Dc tensor product}
	The product $\bullet_{Q_\varepsilon}$ codifies the correlation between polynomial expressions in the (shifted) regularized random field $\widehat{\varphi}_\varepsilon$.
	For concreteness and as an example observe that, for $\tau\in\mathcal{D}_{\mathrm{C}}'(M;\operatorname{Pol})$ and for $\Phi$ defined as per Remark \ref{Rmk: interpretation of Dc local product}, Equation \eqref{Eq: regularized multilocal delta-product} yields
	\begin{align*}
		[\Phi\bullet_{Q_\varepsilon} \Phi](f_1\otimes f_2;\varphi)
		=\int_{M\times M} f_1(x_1)f_2(x_2)[\varphi(x_1)\varphi(x_2)+Q_\varepsilon(x_1,x_2)]\mathrm{d}\mu(x_1)\mathrm{d}\mu(x_2)\,,
	\end{align*}
	which coincides with $\mathbb{E}(\widehat{\varphi}_\varepsilon(f_1)\widehat{\varphi}_\varepsilon(f_2))$ where $\widehat{\varphi}_\varepsilon(x)=(P\circledast\widehat{\xi}_\varepsilon+\varphi)(x)$.
\end{remark}

\begin{remark}\label{Rmk: Necessity of renormalization}
Propositions \ref{Prop: Dc functionals algebraic structure} and \ref{Prop: Dc tensor algebra algebraic structure} clarify how $\mathcal{D}_{\mathrm{C}}'(M;\operatorname{Pol})$ encodes the information on the moments of any polynomial whose variable is a regularized random field $\widehat{\varphi}_\varepsilon$.

Yet, one should pay attention to the fact that, for $\dim M\geq 4$ (or $\dim(\Sigma)\geq 2$ in the parabolic case with $M=\mathbb{R}\times\Sigma$), neither $\widehat{\varphi}_\varepsilon$ nor $\cdot_{Q_\varepsilon},\bullet_{Q_\varepsilon}$ do converge when taking the weak limit as $\varepsilon\to 0^+$. 
As a matter of fact, if one starts from Equations \eqref{Eq: regularized delta-product} and \eqref{Eq: regularized multilocal delta-product} and if one replaces formally $Q_\varepsilon$ with $Q:=P\circ P^*$, the ensuing expressions are ill-defined. 
This should not come as a surprise since similar features occur in other analysis of stochastic PDEs. 
\end{remark}

In the rest of the paper we shall discuss how, using suitable analytic techniques traded from the theory of distributions and of renormalization, one can give a well-defined meaning to both $\cdot_Q$ and $\bullet_Q$ on a suitable subset of $\mathcal{D}_{\mathrm{C}}'(M;\operatorname{Pol})$.
As discussed in the introduction, we follow the strategy which is inspired by the algebraic approach to quantum field theory and to renormalization.
As a first step this calls for the identification of a suitable algebra of functional-valued distributions.

In what follows we shall denote by $\mathcal{E}[\mathcal{O}]$, where $\mathcal{O}\subseteq\mathcal{D}'(M;\operatorname{Fun})$ is any subset of $\mathcal{D}'(M;\operatorname{Fun})$, the smallest $\mathcal{E}(M)$-ring containing $\mathcal{O}$.
We shall refer to $\mathcal{E}[\mathcal{O}]$ as the polynomial ring on $\mathcal{E}(M)$ generated by elements in $\mathcal{O}$ -- though $\mathcal{O}$ is not required to be countable.

\begin{definition}\label{Def: pointwise algebra}
	Let $\boldsymbol{1},\Phi\in\mathcal{D}'(M;\operatorname{Pol})$ be the functional-valued distributions defined by
	\begin{align}\label{Eq: Phi,1 functional}
		\Phi(f;\varphi)
%		:=\varphi(f)
		:=\int_M f\varphi\mu\,,\qquad
		\boldsymbol{1}(f;\varphi)
%		:=1(f)
		=\int_M f\mu\,.
	\end{align}
	We then define a unital, commutative $\mathbb{C}$-algebra as follows.
	We set recursively
	\begin{align}\label{Eq: recursive algebras}
		\mathcal{A}_0
		:=\mathcal{E}[\boldsymbol{1},\Phi]\,,\qquad
		\mathcal{A}_j
		:=\mathcal{E}[\mathcal{A}_{j-1}\cup P\circledast\mathcal{A}_{j-1}]\,,\qquad\forall j\in\mathbb{N}\,,
	\end{align}
	where $P\circledast\mathcal{A}_{j-1}:=\{P\circledast\tau\,|\,\tau\in\mathcal{A}_{j-1}\}$.
	All these algebras are ordered by inclusion, {\it i.e.} $\mathcal{A}_{j_1}\subseteq \mathcal{A}_{j_2}$ if $j_1\leq j_2$, therefore we can introduce the direct limit
	\begin{align}\label{Eq: pointwise algebra}
		\mathcal{A}=\varinjlim \mathcal{A}_j\,,
	\end{align}
	which is thus a commutative and associative $\mathbb{C}$-algebra -- as well as an $\mathcal{E}(M)$-module.
	
	The $\mathbb{C}$-algebra structure is codified by the pointwise product
	\begin{align}\label{Eq: pointwise product}
		[\tau_1\tau_2](f;\varphi)
		:=(\tau_1\otimes \tau_2)(f\delta_{\mathrm{Diag}_2};\varphi)\,,\qquad\forall\,\tau_1,\tau_2\in\mathcal{A}\, .
	\end{align}
\end{definition}

\begin{remark}\label{Rem: graded algebra}
	It is important to highlight that $\mathcal{A}$ ({\em resp.} each $\mathcal{A}_j$, $j\geq 0$) is a positively bigraded algebra over the ring $\mathcal{E}(M)$, namely
	$$\mathcal{A}=\bigoplus_{l,k\in\mathbb{N}_0}\mathcal{M}_{l,k},\qquad \mathcal{A}_j=\bigoplus_{l,k\in\mathbb{N}_0}\mathcal{M}^j_{l,k}$$
	where $\mathcal{M}_{l,k}$ is the $\mathcal{E}(M)$-module generated by the elements of $\mathcal{A}$ in which the parametrix $P$ acts $l$-times while $\Phi$ appears only in degree $k$, {\it e.g.} $P\circledast\left(\Phi^2 P\circledast\Phi^3\right)\in\mathcal{M}_{2,5}$.
	At the same time $\mathcal{M}^j_{l,k}\doteq\mathcal{M}_{l,k}\cap\mathcal{A}_j$.
	Observe that, for later convenience, we introduce $\mathcal{M}_k\doteq\bigoplus_{\substack{l\in\mathbb{N}_0\\p\leq k}}\mathcal{M}_{l,p}$ and $\mathcal{M}^j_k\doteq\bigoplus_{\substack{l\in\mathbb{N}_0\\p\leq k}}\mathcal{M}^j_{l,p}$. 
	In addition it holds that 
	\begin{equation}\label{Eq: module structure of the algebra}
		\mathcal{A}
		=\varinjlim\mathcal{M}_k\,,
		\qquad\mathcal{M}_k\mathcal{M}_{k^\prime}
		=\mathcal{M}_{k+k^\prime},\;\forall k,k^\prime\in\mathbb{N}_0\,.
	\end{equation}
%	with a similar expression for $\mathcal{M}_l(P)$ being valid.
\end{remark}

\begin{remark}
	Observe that the pointwise product on $\mathcal{A}$ is well-defined because any $\tau\in\mathcal{A}$ is a functional-valued distribution generated by a functional-valued smooth function. In addition, notice that $\mathcal{A}$ is stable under the action of $\delta_\psi$ for all $\psi\in\mathcal{E}(M)$ -- \textit{cf.} Equation \eqref{Eq: functional differential}.
\end{remark}

Although the algebra $\mathcal{A}$ does not carry information about the expectation values of $\widehat{\varphi}$, it plays nevertheless an important r\^ole in the construction of a counterpart, which we shall indicate as $\mathcal{A}_{\cdot_Q}$, whose elements do carry information about the expectation values of the polynomial expressions of $\widehat{\varphi}$.

As next step we focus on the functional derivatives of elements lying in $\mathcal{A}$, which are distributions whose wavefront set is controlled as per Definition \ref{Def: Dc functionals}.
In particular, for all $\tau\in\mathcal{A}$ and $p\in\mathbb{N}$, $\tau^{(p)}$ may be singular only on the full diagonal $\mathrm{Diag}_{p+1}\subset M^{p+1}$. In the following lemma we prove a bound on the scaling degree of $\tau^{(p)}$ with respect to $\mathrm{Diag}_{p+1}$ -- \textit{cf.} Definition \ref{Def: scaling degree}.

\begin{lemma}\label{Lem: sd bound in the pointwise algebra}
	Let $\mathcal{A}$ be the algebra introduced in Definition \ref{Def: pointwise algebra}.
	For all $\tau\in\mathcal{A}$ and $p\in\mathbb{N}$ let $\sigma_p(\tau):=\operatorname{sd}_{\mathrm{Diag}_{p+1}}(\tau^{(p)})$ -- \textit{cf.} Definition \ref{Def: scaling degree}.
	Then,
	\begin{align}\label{Eq: sd bound in the pointwise algebra}
		\sigma_p(\tau)<\infty\,,\qquad\forall\tau\in\mathcal{A}\,,\quad\forall p\in\mathbb{N}\,.
	\end{align}
\end{lemma}
\begin{proof}
	As starting point observe that, if $\tau=\boldsymbol{1}$, Equation \eqref{Eq: sd bound in the pointwise algebra} holds true by direct inspection, while, for $\tau=\Phi$, {\it cf.} Equation \eqref{Eq: Phi,1 functional}, it suffices to consider the identities
	\begin{align*}
		\Phi^{(1)}
		=\delta_{\mathrm{Diag}_2}\,,\qquad
		\Phi^{(p)}
		=0,\;\;\forall p>1\,,
	\end{align*}
	together with Example \ref{Ex: sd of delta on total diagonal}. As second step we show that, whenever Equation \eqref{Eq: sd bound in the pointwise algebra} holds true for $\tau_1,\tau_2\in\mathcal{A}$, then it is also verified for $\tau_1\tau_2$. As a matter of fact, Leibniz rule entails that 
	\begin{align*}
		(\tau_1\tau_2)^{(p)}
		=\sum_{\substack{p_1,p_2\\ p_1+p_2=p}}\tau_1^{(p_1)}\otimes\tau_2^{(p_2)}\,,
	\end{align*}
	which, in turn, implies the inequality
	\begin{align}\label{Eq: scaling degree of a product}
		\sigma_p(\tau_1\tau_2)
		\leq\max_{p_1+p_2=p}\sigma_p(\tau_1^{(p_1)}\otimes\tau_2^{(p_2)})
		\leq\max_{p_1+p_2=p}\left(\sigma_{p_1}(\tau_1^{(p_1)})
		+\sigma_{p_2}(\tau_2^{(p_2)})\right)
		<\infty\,.
	\end{align}
	To conclude the proof it suffices to observe that, whenever Equation \eqref{Eq: sd bound in the pointwise algebra} holds true for $\tau\in\mathcal{A}$, then this is the case also for $\psi\tau$ and $P\circledast\tau$ where $\psi\in \mathcal{E}(M)$.
	While the first statement is straightforward, the second is a consequence of Lemma \ref{Lemma: finite sd of convolution}. Indeed we have $(P\circledast\tau)^{(k)}=P\circledast\tau^{(k)}$ and we notice that the hypotheses of Lemma \ref{Lemma: finite sd of convolution} are met on account of the microlocal behaviour of $\tau^{(k)}$ which is codified by the space $C_{k+1}$ (\textit{cf}. Equation \eqref{Eq: C-set}) together with the properties of the parametrix $P$. Putting together these data, one obtains 
	\begin{align*}
		\sigma_p(P\circledast\tau)
		<\infty\,.
	\end{align*}
In view of Definition \ref{Def: pointwise algebra} we can cover inductively all cases, hence proving the sought statement.
\end{proof}
\begin{remark}[Parabolic Case]
This lemma holds true also in the parabolic case, mutatis mutandis. 
%the only difference being that in Equation \eqref{Eq: sd bound in the pointwise algebra} we must use the effective dimension of $M=\mathbb{R}\times\Sigma$, {\it cf.} Remark \ref{Rmk: weighted scaling degree}, \emph{i.e.}, the claim becomes 
%$$\sigma_p(\tau)\leq p(d+1),\quad\forall\,\tau\in\mathcal{A}\;\textrm{and}\;\forall\,p\in\mathbb{N}.$$
\end{remark}

\begin{remark}\label{Rem: equivalence between two algebras}
	Observe that we can also endow $\mathcal{A}$ with the product $\cdot_{Q_\varepsilon}$, {\it cf.} Proposition \ref{Prop: Dc functionals algebraic structure} turning it into a subalgebra of $\mathcal{D}^\prime_{\mathrm{C}}(M;\operatorname{Pol})$.
	Since the integral kernel of $Q_\varepsilon$ is smooth, one can prove by means of a standard argument, {\it e.g.} \cite[Ch. 5]{book}, that  $(\mathcal{A},\cdot_{Q_\varepsilon})$ is isomorphic to $\mathcal{A}$ endowed with the pointwise product. The isomorphism is implemented by
	$$\tau\tau^\prime=\alpha_{Q_\varepsilon}\left(\alpha^{-1}_{Q_\varepsilon}(\tau)\cdot_{Q_\varepsilon}\alpha^{-1}_{Q_\varepsilon}(\tau')\right),$$
	where 
	\begin{equation}\label{Eq: deformedA}
		\alpha_{Q_\varepsilon}\colon\mathcal{A}\to\mathcal{A},\qquad
		\alpha_{Q_\varepsilon}(\tau)(f;\varphi)
		:=\sum_{n=0}^\infty\frac{1}{n!}[(1\otimes Q_\varepsilon^{\otimes k})\cdot\tau^{(2k)}](f\otimes 1_{2k};\varphi)\,,
%		\tau\mapsto\alpha_{Q_\varepsilon}(\tau):=\sum\limits_{n=0}^\infty\frac{1}{n!}\tau^{(2k)}\langle Q_\varepsilon^{\otimes n},\tau^{(2n)}\rangle.
\end{equation}
	where the sum is finite because $\tau$ is a polynomial functional-valued distribution.
%	Here $\langle,\rangle$ is a shortcut to indicate the pairing between compactly supported distributions and their test functions.
	Notice that $\alpha_{Q_\varepsilon}^{-1}=\alpha_{-Q_\varepsilon}$.
\end{remark}

\section{Construction of $\mathcal{A}_{\cdot_Q}$}\label{Sec: construction of AcdotQ}

In the previous section we have constructed the algebra $\mathcal{A}$, {\it cf.} Definition \ref{Def: pointwise algebra} endowed with the pointwise product as well as with the $\cdot_{Q_\varepsilon}$-product, {\it cf.} Remark \ref{Rem: equivalence between two algebras}. Although the two algebras are isomorphic, the latter has greater significance since it allows to encode the information on the expectation value of polynomial expressions in the shifted regularized random field $\widehat{\varphi}_\varepsilon$. As mentioned in Remark \ref{Rmk: Necessity of renormalization}, the weak limit as $\varepsilon\to 0^+$ is not a priori well-defined. The goal of this section is to discuss a way to bypass this hurdle by defining $\cdot_Q$ as a suitably renormalized version of the limit as $\varepsilon\to 0^+$ of the product $\cdot_{Q_\varepsilon}$. We tackle the problem in two steps. In the first one we introduce a suitable deformation of $\mathcal{A}$ similar in spirit to the map $\alpha_{Q_\varepsilon}$ of Remark \ref{Rem: equivalence between two algebras}. In the second one, we endow the deformed version of $\mathcal{A}$ with a new algebra structure by means of an expression similar to that of Equation \eqref{Eq: deformedA}.

\begin{theorem}\label{Thm: Gamma cdotQ existence}
	Let $\mathcal{A}$ be the the algebra as per Definition \ref{Def: pointwise algebra}. Then there exists a linear map $\Gamma_{\cdot_Q}\colon\mathcal{A}\to\mathcal{D}_{\mathrm{C}}'(M;\operatorname{Pol})$ with the following properties:
	\begin{enumerate}
		\item
		for all $\tau\in\mathcal{M}_1$ -- \textit{cf.} Remark \ref{Rem: graded algebra} -- it holds
		\begin{align}\label{Eq: Gamma on M1}
			\Gamma_{\cdot_Q}(\tau)=\tau\,.
		\end{align}
		\item
		for all $\tau\in\mathcal{A}$ it holds
		\begin{align}\label{Eq: Gamma on Ptau}
			\Gamma_{\cdot_Q}(P\circledast\tau)
			=P\circledast\Gamma_{\cdot_Q}(\tau)\,.
		\end{align}
		\item
		for all $\psi\in \mathcal{E}(M)$ it holds
		\begin{align}\label{Eq: Pi-functional derivative}
			\Gamma_{\cdot_Q}\circ\delta_\psi
			=\delta_\psi\circ\Gamma_{\cdot_Q}\,,\qquad
			\Gamma_{\cdot_Q}(\psi\tau)
			=\psi\Gamma_{\cdot_Q}(\tau)\,.
		\end{align}
		\item
		For all $\tau\in\mathcal{A}$ and $p\geq 1$ and in view of Lemma \ref{Lem: sd bound in the pointwise algebra}
		\begin{align}\label{Eq: Pi scaling degree bound}
			\sigma_p(\Gamma_{\cdot_Q}(\tau))
			<\infty\,,
		\end{align}
		where $\mathrm{Diag}_{p+1}\subset M^{p+1}$ is the total diagonal of $M^{p+1}$.
	\end{enumerate}
\end{theorem}

\begin{proof}
	The proof is lengthy, hence we divide it in several steps. 
	
	\paragraph{Proof for $d\in\{2,3\}$.}
	For the particular case of $d\in\{2,3\}$ we set $\Gamma_{\cdot_Q}(\tau)=\tau$ for all $\tau\in\mathcal{M}_1$ and
	\begin{align*}
		\Gamma_{\cdot_Q}(\tau_1\cdots\tau_\ell)
		:=\tau_1\cdot_Q\ldots\cdot_Q\tau_\ell\,,
	\end{align*}
	for all $\tau_1,\ldots,\tau_\ell\in\mathcal{M}_1$, where $\cdot_Q$ is defined by giving meaning to an expression similar to Equation \eqref{Eq: regularized delta-product} with $Q_\varepsilon$ replaced by $Q$.
	
	Notice that, under the assumption that $d\in\{2,3\}$, $\cdot_Q$ is actually well-defined on account of the logarithmic divergence of $Q$ at the total diagonal of $M^2$.
	All other properties required by $\Gamma_{\cdot_Q}$ are verified by direct inspection.
			
	\paragraph{Strategy of the proof for $d\geq 4$.}
	The main idea of the proof is to construct $\Gamma_{\cdot_Q}$ inductively exploiting Equation \eqref{Eq: module structure of the algebra}, in particular that $\mathcal{A}=\bigoplus\limits_{k\in\mathbb{N}_0}\mathcal{M}_k$.
	One would start from Equation \eqref{Eq: Gamma on M1} and, for $\tau=\tau_1\cdots\tau_n$, $\tau_i\in\mathcal{A}$, $i=1,\dots,n$, one would set
	\begin{align*}
		\Gamma_{\cdot_Q}(\tau)
		:=\Gamma_{\cdot_Q}(\tau_1)\cdot_Q\ldots\cdot_Q\Gamma_{\cdot_Q}(\tau_n)\,.
	\end{align*}
	However, contrary to the case $d\in\{2,3\}$, the product $\cdot_Q$ is ill-defined, on account of the more singular behaviour of $Q$ on the total diagonal.
	To cope with this hurdle we shall proceed by renormalizing the ill-defined expressions appearing in the product $\cdot_Q$ in a way consistent with the grading of $\mathcal{A}$.
	In particular, we shall prove that, whenever $\Gamma_{\cdot_Q}$ has been defined on the submodule $\mathcal{M}_k$, then it can be extended (non-uniquely) to $\mathcal{M}_{k+1}$.
	The extension procedure will require not only an induction over $k$, but also over the index $j$ controlling the direct limit $\mathcal{A}=\varinjlim\mathcal{A}_j$ where $\mathcal{A}_j=\bigoplus_{k\in\mathbb{N}_0}\mathcal{M}_k^j$.
	
	\paragraph{Step 1.} As starting point it is convenient to show that, assuming $\Gamma_{\cdot_Q}$ has been assigned for $\tau\in\mathcal{A}$ fulfilling properties {\em 3.} and {\em 4.}, then Equations \eqref{Eq: Pi-functional derivative} and \eqref{Eq: Pi scaling degree bound} must hold true also for $P\circledast\tau$.
	First of all notice that, for any $\tau\in\mathcal{A}$, $\Gamma_{\cdot_Q}(P\circledast\tau)$ is completely defined via Equation \eqref{Eq: Gamma on Ptau} as 
	\begin{align*}
		\Gamma_{\cdot_Q}(P\circledast\tau)(f;\varphi)
		:=P\circledast \Gamma_{\cdot_Q}(\tau)(f;\varphi)
		=\Gamma_{\cdot_Q}(\tau)(P\circledast f;\varphi)\,.
	\end{align*}
	Since $\Gamma_{\cdot_Q}(\tau)\in\mathcal{D}_{\mathrm{C}}^\prime(M;\operatorname{Pol})$, Lemma \ref{Lem: if it holds for tau it holds for Ptau} entails that $P\circledast\Gamma_{\cdot_Q}(\tau)\in\mathcal{D}_{\mathrm{C}}'(M;\operatorname{Pol})$.
	In addition, for all $k\geq 0$ and for all $\psi\in \mathcal{E}(M)$ it holds
	\begin{align*}
		\Gamma_{\cdot_Q}(P\circledast\tau)^{(k)}(f\otimes\psi^{\otimes k};\varphi)
		&=[(P\otimes\delta_{\mathrm{Diag}_2}^{\otimes k})\circledast\Gamma_{\cdot_Q}(\tau)^{(k)}](f\otimes\psi^{\otimes k};\varphi)\,.
	\end{align*}
	Equation \eqref{Eq: Pi-functional derivative} is a direct consequence of the following chain of identities
	\begin{align*}
		[P\circledast\Gamma_{\cdot_Q}(\tau)]^{(1)}(f\otimes\psi;\varphi)
		&=\Gamma_{\cdot_Q}(\tau)^{(1)}(P\circledast f\otimes\psi;\varphi)\\
		&=\Gamma_{\cdot_Q}(\delta_\psi\tau)(P\circledast f;\varphi)
		=\Gamma_{\cdot_Q}(P\circledast\delta_\psi\tau)(f;\varphi)
		=\Gamma_{\cdot_Q}(\delta_\psi P\circledast\tau)(f;\varphi)\,.
	\end{align*}
At the same time Equation \eqref{Eq: Pi scaling degree bound} follows from Lemma \ref{Lem: if it holds for tau it holds for Ptau} -- \textit{cf.} Equation \eqref{Eq: P convolution functional scaling degree} in particular -- together with the assumption that Equation \eqref{Eq: Pi scaling degree bound} holds true for $\Gamma_{\cdot_Q}(\tau)$: 
	\begin{align*}
		\sigma_p(\Gamma_{\cdot_Q}(P\circledast\tau))
		<\infty\,.
	\end{align*}
	
	\paragraph{Step 2 -- First induction procedure: $k=1,2$.}
	We focus on defining inductively $\Gamma_{\cdot_Q}$ on $\mathcal{M}_k$ -- \textit{cf.} Remark \eqref{Rem: graded algebra}. The case $k=1$ is ruled by Equation \eqref{Eq: Gamma on M1} and it represents the first step in the induction procedure. We can focus on case $k=2$. We will discuss it thoroughly and eventually we shall generalize the same procedure to arbitrary $k$.

	In order to extend $\Gamma_{\cdot_Q}$ from $\mathcal{M}_1$ to $\mathcal{M}_2$, we exploit that $\mathcal{M}_2=\bigcup_{j\in\mathbb{N}_0}\mathcal{M}_2^j$ -- \textit{cf.} Remark \ref{Rem: graded algebra}. Here we proceed inductively over $j$.
	For $j=0$, recall that $\mathcal{M}_2^0$ is the $\mathcal{E}(M)$-module 
	\begin{align*}
		\mathcal{M}_2^0
		=\textrm{span}_{\mathcal{E}(M)}(\boldsymbol{1},\Phi,\Phi^2)\,.
	\end{align*}
	In view of Equation \eqref{Eq: Gamma on M1}, the only unknown is $\Gamma_{\cdot_Q}(\Phi^2)$. Afterwards $\Gamma_{\cdot_Q}$ can be extended per linearity to the whole $\mathcal{M}_2^0$. Hence, for all $f\in\mathcal{D}(M)$ and for all $\varphi\in\mathcal{E}(M)$, recalling that  $\Phi^2$ has been defined in Example \ref{Ex: examples of functional-valued polynomial distributions}, we set formally
	\begin{align}\label{Eq: first formal product}
		\Gamma_{\cdot_Q}(\Phi^2)(f;\varphi)\doteq\Gamma_{\cdot_Q}(\Phi)\cdot_Q\Gamma_{\cdot_Q}(\Phi)(f;\varphi)
		=\Phi^2(f;\varphi)
		+P^2(f\otimes 1)\,.
	\end{align}
	The second equality is nothing but Equation \eqref{Eq: regularized delta-product} with $Q_\varepsilon$ replaced by $Q=P\circ P^*$, where we have also used both that $\Gamma_{\cdot_Q}^{(1)}(\Phi)=\delta_{\mathrm{Diag}_2}$ and that
	\begin{align*}
		(\delta_{\mathrm{Diag}_2}\otimes Q)\cdot(\Gamma_{\cdot_Q}(\Phi)^{(1)}\otimes\Gamma_{\cdot_Q}(\Phi)^{(1)})(f\otimes 1_3;\varphi)
		&=\int_M Q(x,x)f(x)\mathrm{d}\mu(x)
		\\&=\int_M P^2(x,y)f(x) 1(y)\mathrm{d}\mu(x)\mathrm{d}\mu(y)
		=P^2(f\otimes 1)\,.
	\end{align*}
 Observe that, since we are working with parametrices of elliptic operators $P^*(x,y)=P(y,x)$ and this justifies why in the last formula $P^2(x,y)$ is present. Both this last expression and Equation \eqref{Eq: first formal product} are purely formal since $P^2$, the square of the parametrix $P$, is ill-defined. To cope with this issue we start by observing that $P^2\in\mathcal{D}'(M^2\setminus\mathrm{Diag}_2)$ is a well-defined distribution because $\operatorname{WF}(P)=\operatorname{WF}(\operatorname{\delta_{\mathrm{Diag}_2}})$ -- \textit{cf.} Example \ref{Ex: WF of parametrix}.
	In addition, using Equation \eqref{Eq: sd of product} and Example \ref{Ex: sd of parametrix}, $\operatorname{sd}_{\mathrm{Diag}_2}(P^2)\leq 2(d-2)<+\infty$ 	and, therefore,  on account of Theorem \ref{Thm: extension with scaling degree} there exists an extension $\widehat{P}_2\in\mathcal{D}'(M^2)$ of $P^2$ such that $\operatorname{sd}_{\mathrm{Diag}_2}(\widehat{P}_2)=\operatorname{sd}_{\mathrm{Diag}_2}(P^2)$.
	Moreover $\operatorname{WF}(\widehat{P}_2)=\operatorname{WF}(\delta_{\mathrm{Diag}_2})$.
	
	We assume that one such extension, $\widehat{P}_2$, has been chosen once and for all. Notice that the said extension is unique when $\dim M=d\in\{2,3\}$, consistently with the definition of $\cdot_Q$. Accordingly and in view of Equation \eqref{Eq: first formal product}, we define
	\begin{align}\label{Eq: definition of Phi2}
		\Gamma_{\cdot_Q}(\Phi^2)(f;\varphi)
		:=\Phi^2(f;\varphi)
		+\widehat{P}_2(f\otimes 1)\,.
	\end{align}
	In addition, we can infer that  $\Gamma_{\cdot_Q}(\Phi^2)\in\mathcal{D}'_{\mathrm{C}}(M;\operatorname{Pol})$.
	In particular Remark \ref{Rmk: convolution with smooth function} yields that $\widehat{P}_2\circledast 1\in \mathcal{E}(M)$ while the bound $\operatorname{WF}([\Gamma_{\cdot_Q}(\Phi)^2]^{(\ell)})\subseteq \mathrm{C}_{\ell+1}$ holds true combining Example \ref{Ex: examples of functional-valued polynomial distributions} with the wavefront set of $\widehat{P}_2$.
	Moreover it descends that
	\begin{align*}
		\Gamma_{\cdot_Q}(\Phi^2)^{(1)}(f\otimes\psi;\varphi)
		=[\Phi^2]^{(1)}(f\otimes\psi;\varphi)\,,\qquad
		\Gamma_{\cdot_Q}(\Phi^2)^{(2)}(f\otimes\psi_1\otimes\psi_2;\varphi)
		=[\Phi^2]^{(2)}(f\otimes\psi_1\otimes\psi_2;\varphi)\,,
	\end{align*}
	from which Equations \eqref{Eq: Pi-functional derivative} and \eqref{Eq: Pi scaling degree bound} hold true.
	
	This completes the definition of $\Gamma_{\cdot_Q}$ on $\mathcal{M}_2^0$.
	Proceeding inductively with respect to the index $j$, we assume that $\Gamma_{\cdot_Q}$ has been defined on $\mathcal{M}_2^j$ iterating the procedure used in Equation \eqref{Eq: first formal product} and we extend it to $\mathcal{M}_2^{j+1}$.
	Observe that, given any $\tau\in\mathcal{M}^{j+1}_2$, it suffices to prove the induction step for those elements either of the form $P\circledast\tau^\prime$, $\tau^\prime\in\mathcal{M}^j_2$ or of the form $\tau=\tau_1\tau_2$ with $\tau_1,\tau_2\in\mathcal{M}_1^j\cup P\circledast\mathcal{M}^j_1$ -- \textit{cf.} Definition \ref{Def: pointwise algebra}.
	In the first case, it suffices to invoke the induction hypothesis and the first step of the proof.
	
	In the second case, we consider the formal expression
	\begin{align*}
		\Gamma_{\cdot_Q}(\tau_1)\cdot_Q\Gamma_{\cdot_Q}(\tau_2)(f;\varphi)
		&=\left(\Gamma_{\cdot_Q}(\tau_1)\Gamma_{\cdot_Q}(\tau_2)\right)(f;\varphi)+
		\left[(\delta_{\mathrm{Diag}_2}\otimes Q)\cdot(t_1^{(1)}\widetilde{\otimes} t_2^{(1)})\right](f\otimes 1_3;\varphi)\,,
	\end{align*}
	where $t_1^{(1)}:=\Gamma_{\cdot_Q}(\tau_1)^{(1)}$ and similarly $t_2^{(1)}$.
	The above formula is a priori not well-defined on account of
	\begin{align*}
		T:=(\delta_{\mathrm{Diag}_2}\otimes Q)\cdot(t_1^{(1)}\widetilde{\otimes} t_2^{(1)})\,.
	\end{align*}\label{Qdelta}
	To bypass this hurdle, recall that $\Gamma_{\cdot_Q}(\tau_1),\Gamma_{\cdot_Q}(\tau_2)\in\mathcal{D}'_{\mathrm{C}}(M;\operatorname{Pol})$. Hence it holds that
	\begin{align*}
		\operatorname{WF}(t_1^{(1)})\cup\operatorname{WF}(t_2^{(1)})
		\subseteq \mathrm{C}_2
		=\operatorname{WF}(\delta_{\mathrm{Diag}_2})
		=\operatorname{WF}(Q)\,.
	\end{align*}
	It follows that $T$ identifies an element of $\mathcal{D}'(M^4\setminus\mathrm{Diag}_4^{\mathrm{big}})$ where
	\begin{align*}
		\mathrm{Diag}_4^{\mathrm{big}}
		:=\{(x_1,\dots,x_4)\in M^4\,|\,
		\exists a,b\in\{1,2,3,4\}\,,\,x_a=x_b\}\,.
	\end{align*}
	Moreover observe that, whenever $x\in\mathrm{Diag}_4^{\mathrm{big}}\setminus\mathrm{Diag}_4$ one of the factors between $\delta_{\mathrm{Diag}_2}\otimes Q$, $t_1^{(1)}$, $t_2^{(1)}$ is smooth while the product of the other two is well-defined.
	
	This entails that $T\in\mathcal{D}'(M^4\setminus\mathrm{Diag}_4)$. Furthermore, on account of the inductive hypothesis over $j$ it holds that
	\begin{align*}
		\operatorname{sd}_{\mathrm{Diag}_4}(T)
		\leq\operatorname{sd}_{\mathrm{Diag}_4}(\delta_{\mathrm{Diag}_2}\otimes Q)
		+\operatorname{sd}_{\mathrm{Diag}_2}(t_1^{(1)})
		+\operatorname{sd}_{\mathrm{Diag}_2}(t_2^{(1)})
		<\infty\,,
	\end{align*}
	where in the last inequality we also used Corollary \ref{Cor: sd of composition} to prove finiteness of $\mathrm{sd}_{\mathrm{Diag}_2}(Q)$.
	Thanks to Theorem \ref{Thm: extension with scaling degree} we can conclude that there exists a possibly non-unique extension $\widehat{T}\in\mathcal{D}'(M^4)$ of $T$ such that $\operatorname{sd}_{\mathrm{Diag}_4}(\widehat{T})=\operatorname{sd}_{\mathrm{Diag}_4}(T)$. In addition it holds that $\operatorname{WF}(\widehat{T})=\operatorname{WF}(T)$. 	Choosing an extension $\widehat{T}$, we define for all $f\in\mathcal{D}(M)$ and for all $\varphi\in\mathcal{E}(M)$,
	\begin{align*}
		\Gamma_{\cdot_Q}(\tau)(f;\varphi)
		:=\left(\Gamma_{\cdot_Q}(\tau_1)\Gamma_{\cdot_Q}(\tau_2)\right)(f;\varphi)
		+\widehat{T}(f\otimes 1_3;\varphi)\,.
	\end{align*}
	By direct inspection we see that $\Gamma_{\cdot_Q}(\tau)\in\mathcal{D}'_{\mathrm{C}}(M;\operatorname{Pol})$.
	Equation \eqref{Eq: Pi-functional derivative} is satisfied by construction, while, to check the inequality in Equation \eqref{Eq: Pi scaling degree bound} we observe that
	\begin{gather*}
		\Gamma_{\cdot_Q}(\tau)^{(1)}(f\otimes\psi;\varphi)
		=\left(\Gamma_{\cdot_Q}(\tau_1)\Gamma_{\cdot_Q}(\tau_2)^{(1)}+\Gamma_{\cdot_Q}(\tau_1)^{(1)}\Gamma_{\cdot_Q}(\tau_2)\right)(f\otimes\psi;\varphi)\,,\\
		\Gamma_{\cdot_Q}(\tau)^{(2)}(f\otimes\psi_1\otimes\psi_2;\varphi)
		=\left(\Gamma_{\cdot_Q}^{(1)}(\tau_1)\Gamma_{\cdot_Q}(\tau_2)^{(1)}\right)		(f\otimes\psi_1\otimes\psi_2;\varphi)\,.
	\end{gather*}
	It descends that, 
	\begin{align*}
		\sigma_1(\Gamma_{\cdot_Q}(\tau))
%		\operatorname{sd}_{\mathrm{Diag}_2}(\Gamma_{\cdot_Q}(\tau)^{(1)})
		&\leq\max_{i=1,2}
		\sigma_1(\Gamma_{\cdot_Q}(\tau_i))
%		\operatorname{sd}_{\mathrm{Diag}_2}\Gamma_{\cdot_Q}(\tau_{k_p})^{(1)}
		<\infty\,,\\
		\sigma_2(\Gamma_{\cdot_Q}(\tau))
		&\leq\sigma_1(\Gamma_{\cdot_Q}(\tau_1))+\sigma_1(\Gamma_{\cdot_Q}(\tau_2))
%		\operatorname{sd}_{x[3]}(\Gamma_{\cdot_Q}(\tau)^{(2)})
%		&\leq\max\left[\sigma_1(\tau_1^{(1)})+\sigma_1(\tau_2^{(1)})\right]
%		\operatorname{sd}_{\mathrm{Diag}_2}\Gamma_{\cdot_Q}(\tau_{k_a})^{(1)}
%		+\operatorname{sd}_{\mathrm{Diag}_2}\Gamma_{\cdot_Q}(\tau_{k_b})^{(1)}
		<\infty\,,
	\end{align*}
	where we used Equation \eqref{Eq: Pi scaling degree bound} applied to $\Gamma_{\cdot_Q}(\tau_i)$, $i=1,2$. Its validity is guaranteed by the induction step.
	
	\paragraph{Step 3 -- Second induction procedure.}	
	In the preceding step we have proven the sought after statement for $\mathcal{M}_k$ with $k=1,2$. We proceed by induction over $k$. In other words we assume that $\Gamma_{\cdot_Q}$ has been defined on $\mathcal{M}_k$ for an arbitrary but fixed $k$ and we show that the same holds true for $\mathcal{M}_{k+1}=\bigcup_{j\in\mathbb{N}_0}\mathcal{M}^j_{k+1}$, {\it cf.} Remark \ref{Rem: graded algebra}.

	\paragraph{Step 3a -- The special case $\mathcal{M}^0_{k+1}$.}
	If we set $j=0$ we are considering the $\mathcal{E}(M)$-module
	\begin{align*}
		\mathcal{M}^0_{k+1}
		=\operatorname{span}_{\mathcal{E}(M)}(\boldsymbol{1},\Phi,\ldots,\Phi^{k+1})\,.
	\end{align*}
	By the inductive hypothesis over $k$ and on account of Equation \eqref{Eq: Gamma on Ptau}, we are left with defining $\Gamma_{\cdot_Q}(\Phi^{k+1})$. Following the same strategy of Equation \eqref{Eq: first formal product} and bearing in mind the identities $$\Gamma_{\cdot_Q}(\Phi)^{(1)}=\delta_{\mathrm{Diag}_2},\qquad\Gamma_{\cdot_Q}(\Phi)^{(k)}=0,\quad\forall\,k\geq2$$ 
	we consider the formal expression
	\begin{align}\label{Eq: Phik ill-defined expression}
		\Gamma_{\cdot_Q}(\Phi^{k+1})=\underbrace{\Gamma_{\cdot_Q}(\Phi)\cdot_Q\ldots\cdot_Q\Gamma_{\cdot_Q}(\Phi)}_{k+1}(f;\varphi)
		=\sum_{\ell=0}^{\lfloor \frac{k+1}{2}\rfloor}\frac{(2\ell)!}{(2\ell)!!}{k+1\choose 2\ell}(Q_{2\ell}\cdot\Gamma_{\cdot_Q}(\Phi)^{k+1-2\ell})(f;\varphi)\,,
	\end{align}
	where $\lfloor \frac{k+1}{2}\rfloor\leq\frac{k+1}{2}$ denotes the integer part of $\frac{k+1}{2}$ while $n!!:=n(n-2)(n-4)\cdots $ stands for the double factorial.
	Here $Q_{2\ell}$ is given by
	%\footnote{
	%	More explicitly we have
	%	\begin{align*}
	%		Q_{2\ell}(f)
	%		=\int Q(x,x)^\ell f(x)\mathrm{d}x
	%		=\int \prod_{j=1}^\ell P^2(x,y_j)f(x)1_\ell(\widehat{y}_\ell)\mathrm{d}x\mathrm{d}\widehat{y}_\ell\,.
	%	\end{align*}
	%	}
	\begin{align*}
		Q_{2\ell}(f)
		=(P^2)^{\otimes\ell}\cdot(\delta_{\mathrm{Diag}_{\ell}}\otimes 1_\ell)(f\otimes 1_{2\ell-1})\,.
	\end{align*}
	The expression in Equation \eqref{Eq: Phik ill-defined expression} is formal due to the presence of $Q_{2\ell}$, which is built out of $P^2$, the square of the parametrix $P$.
	Nevertheless, we have already shown that $P^2$ admits at least one extension $\widehat{P}_2\in\mathcal{D}'(M^2)$ with $\operatorname{sd}_{\mathrm{Diag}_2}(\widehat{P}_2)=\operatorname{sd}_{\mathrm{Diag}_2}(P^2)$.
	Given an arbitrary but fixed choice for $\widehat{P}_2$ we denote by
	\begin{align*}
		\widehat{Q}_{2\ell}(f)
		:=\widehat{P}_2^{\otimes\ell}\cdot(\delta_{\mathrm{Diag}_{\ell}}\otimes 1_\ell)(f\otimes 1_{2\ell-1})\,,
	\end{align*}
	the corresponding extension of $Q_{2\ell}$. Notice that the product
	%\footnote{
	%	The product is among the distribution $\widehat{P}_2^{\otimes\ell}$ with integral kernel $\prod_{j=1}^\ell\widehat{P}_2(y_j,z_j)$ and the distribution $\delta_{\mathrm{Diag}_{\ell}}\otimes 1_\ell$ with integral kernel $\delta_{\mathrm{Diag}_{\ell}}(\widehat{y}_\ell)1_\ell(\widehat{z}_\ell)$.
	%}
	$\widehat{P}_2^{\otimes\ell}\cdot(\delta_{\mathrm{Diag}_{\ell}}\otimes 1_\ell)$ is well-defined on account of Theorem \ref{Thm: WF results} -- \textit{cf.} Equation \eqref{Eq: WF product condition}.
	We consider the extension $\widehat{Q}_{2\ell}$ associated with the choice of $\widehat{P}_2$ which has been outlined in Step 2. Hence we can set
	\begin{align}\label{Eq: definition of Phik}
		\Gamma_{\cdot_Q}(\Phi^{k+1})(f;\varphi)
		:=\sum_{\ell=0}^{\lfloor \frac{k+1}{2}\rfloor}\frac{(2\ell)!}{(2\ell)!!}{k+1\choose 2\ell}
		[\widehat{Q}_{2\ell}\cdot\Gamma_{\cdot_Q}(\Phi)^{k+1-2\ell}](f;\varphi)\,.
	\end{align}
	Since $\Gamma_{\cdot_Q}(\Phi)$ is a functional-valued distribution generated by a smooth function, the product $\widehat{Q}_{2\ell}\cdot\Gamma_{\cdot_Q}(\Phi)^{k+1-2\ell}$ is well-defined.
	As a matter of fact, a direct application of Equation \eqref{Eq: WF product bound} and of Remark \ref{Rmk: convolution with smooth function} shows that $[\widehat{Q}_{2\ell}\cdot\Gamma_{\cdot_Q}(\Phi)^{k+1-2\ell}]$ is a functional-valued distribution generated by a smooth function.

	Equation \eqref{Eq: definition of Phik} defining $\Gamma_{\cdot_Q}(\Phi^{k+1})$ is compatible with Equation \eqref{Eq: Pi-functional derivative} since, for all $f\in\mathcal{D}(M)$ and for all $\varphi,\psi\in\mathcal{E}(M)$
	\begin{align*}
		\Gamma_{\cdot_Q}(\Phi^{k+1})^{(j)}(f\otimes\psi^{\otimes j};\varphi)
		&=\sum_{\ell=0}^{\lfloor \frac{k+1-j}{2}\rfloor}
		\frac{(k+1)!}{(2\ell)!!}(k+1-2\ell-j)!
		[\widehat{Q}_{2\ell}\cdot\Gamma_{\cdot_Q}(\Phi)^{k+1-2\ell-j}](f\psi^j;\varphi)
		\\&=\Gamma_{\cdot_Q}(\delta^j_{\psi^{\otimes j}}\Phi^k)(f;\varphi)\,,
	\end{align*}
	where we used that $[\Gamma_{\cdot_Q}(\Phi)^{a}]^{(b)}=\frac{a!}{(a-b)!}\Gamma_{\cdot_Q}(\Phi)^{a-b}\cdot\delta_{\mathrm{Diag}_{b+1}}$, while $\delta^j_{\psi^{\otimes j}}:=\delta_\psi\circ\cdots\circ\delta_\psi$.
	On account of Equation \eqref{Eq: functional differential}, it holds 
	\begin{align*}
		\Gamma_{\cdot_Q}(\Phi^{k+1})^{(j)}(f\otimes\psi^{\otimes j};\varphi)
		&=\frac{(k+1)!}{(k+1-j)!}\Gamma_{\cdot_Q}(\underbrace{\psi\cdots\psi}_j\Phi^{k+1-j})(f;\varphi)
		\\&=\frac{(k+1)!}{(k+1-j)!}\Gamma_{\cdot_Q}(\Phi^{k+1-j})(\underbrace{\psi\cdots\psi}_j f;\varphi)\,.
	\end{align*}
	Due to the arbitrariness of both $\psi$, $f$ and $\varphi$, this entails that
	\begin{align*}
		\Gamma_{\cdot_Q}(\Phi^{k+1})^{(j)}
		=\frac{(k+1)!}{(k+1-j)!}\Gamma_{\cdot_Q}(\Phi^{k+1-j})\cdot\delta_{\mathrm{Diag}_{j+1}}\,,
	\end{align*}
	which implies in turn $\operatorname{WF}(\Gamma_{\cdot_Q}(\Phi^{k+1})^{(j)})\subseteq \mathrm{C}_{j+1}$. Therefore $\Gamma_{\cdot_Q}(\Phi^{k+1})\in\mathcal{D}'_{\mathrm{C}}(M;\operatorname{Pol})$.
	
	Finally Equation \eqref{Eq: Pi scaling degree bound} is a direct consequence of the bound
	\begin{align*}
		\sigma_p(\Gamma_{\cdot_Q}(\Phi^{k+1}))
%		\operatorname{sd}_{\mathrm{Diag}_{p+1}}(\Gamma_{\cdot_Q}(\Phi^{k+1})^{(p)})
		\leq\operatorname{sd}_{\mathrm{Diag}_{p+1}}(\Gamma_{\cdot_Q}(\Phi^{k+1-p})\cdot\delta_{\mathrm{Diag}_{p+1}})
		\leq\operatorname{sd}_{\mathrm{Diag}_{p+1}}(\delta_{\mathrm{Diag}_{p+1}})
		=pd\,.
	\end{align*}

	\paragraph{Step 3b: The general case $\mathcal{M}^j_{k+1}$.}
	The preceding step allows us to proceed in the inductive construction of $\Gamma_{\cdot_Q}$. In particular we assume that the sought after result is known for $\mathcal{M}^j_{k+1}$ and we show that the same holds true for $\mathcal{M}^{j+1}_{k+1}$. 
	
	Hence, let us consider a generic $\tau\in\mathcal{M}^{j+1}_{k+1}$.
	In view of Definition \ref{Def: pointwise algebra} and Remark \ref{Rem: graded algebra} and of the linearity of $\Gamma_{\cdot_Q}$, it suffices to focus our attention on those elements which are either of the form $P\circledast\tau^\prime$, with $\tau^\prime\in\mathcal{M}^j_{k+1}$, or 
	$$\tau=\tau_{k_1}\cdots\tau_{k_\ell}\,,\qquad\tau_{k_n}\in\mathcal{M}^j_{k_n}\cup P\circledast\mathcal{M}^j_{k_n}\,,$$
	where $\ell\in\mathbb{N}\cup\{0\}$, $k_n\in\mathbb{N}$ for all $n\in\{1,\ldots,\ell\}$, $\sum\limits_{n=1}^\ell k_n=k+1$.
	While in the first case we can resort to Step 1 of this proof, in the second one we can start be recalling that the inductive hypothesis entails that $\Gamma_{\cdot_Q}(\tau_{k_n})$ is known for all $n$.
	
\noindent Following the same strategy of Step 2. and of Step 3a, we wish to set
	\begin{align*}
	\Gamma_{\cdot_Q}(\tau)
	:=\Gamma_{\cdot_Q}(\tau_{k_1})\cdot_Q\cdots_Q\Gamma_{\cdot_Q}(\tau_{k_\ell})\,.
	\end{align*}
Using the notation $t_{k_i}:=\Gamma_{\cdot_Q}(\tau_{k_i})$ and using Equation \eqref{Eq: regularized delta-product} with $Q_\varepsilon$ replaced by $Q$, we obtain
	\begin{gather*}
		[t_{k_1}\cdot_Q\cdots_Q t_{k_\ell}](f;\varphi)
		=\\=\sum_{\substack{N\geq 0\\N_1+\ldots+ N_\ell=2N}}
		\frac{1}{(2N)!!}\frac{(2N)!}{N_1!\cdots N_\ell!}
		[(\delta_{\mathrm{Diag}_\ell}\otimes Q^{\otimes N})\cdot
		(t_{k_1}^{(N_1)}\widetilde{\otimes}\ldots\widetilde{\otimes} t_{k_\ell}^{(N_\ell)})]
		(f\otimes 1_{\ell-1+2N};\varphi)\,,
	\end{gather*}
	where $\cdot$ indicates once more the product between distributions. As in the preceding cases, this is a formal expression and, to make it well-defined, we start by noticing that, being all functionals polynomial, then $N_i\leq k_i$ for all $i\in\{1,\ldots,\ell\}$ and thus $2N\leq k+1$. For later convenience we set
	\begin{align}\label{Eq: to be extended distribution}
		T_N
		:=[(\delta_{\mathrm{Diag}_\ell}\otimes Q^{\otimes N})\cdot
		(t_{k_1}^{(N_1)}\widetilde{\otimes}\ldots\widetilde{\otimes} t_{k_\ell}^{(N_\ell)})]\,.
	\end{align}
	In addition it holds that
	\begin{multline*}
		\operatorname{WF}(\delta_{\mathrm{Diag}_\ell}\otimes Q^{\otimes N})
		\subseteq\{
		(\widehat{x}_\ell,\widehat{z}_{2N},\widehat{\xi}_\ell,\widehat{\zeta}_{2N})\in T^*M^{\ell+2N}\setminus\{0\}\,|\,
		\\(\widehat{x}_\ell,\widehat{\xi}_\ell)\in\operatorname{WF}(\delta_{\mathrm{Diag}_\ell})\,,\,
		(\widehat{z}_{2N},\widehat{\zeta}_{2N})\in\operatorname{WF}(Q^{\otimes N})
		\}\,.
	\end{multline*}
	The inductive hypothesis entails that $\operatorname{WF}(t_{k_i}^{(N_i)})\subseteq\mathrm{C}_{N_i+1}$ for all $i\in\{1,\ldots,\ell\}$ and, therefore, by applying Theorem \ref{Thm: WF results} -- \textit{cf.} Equation \eqref{Eq: WF product condition} -- $T_N\in\mathcal{D}'(M^{\ell+2N}\setminus\mathrm{Diag}_{\ell+2N}^{\mathrm{big}})$, where
	\begin{align}
		\mathrm{Diag}_{\ell+2N}^{\mathrm{big}}
		:=\{x\in M^{\ell+2N}\,|\, \exists i,j\in\{1,\ldots,\ell+2N\}\,,\, x_i=x_j\}\,.
	\end{align}	
	These data in combination with Equation \eqref{Eq: WF product bound} yield
	\begin{multline}\label{Eq: WF of to be extended distribution}
		\operatorname{WF}(T_N)
		\subseteq\{
		(\widehat{x}_\ell,\widehat{z}_{2N},\widehat{\xi}_\ell+\widehat{\xi}_\ell',\widehat{\zeta}_{N_1}+\widehat{\zeta}_{N_1}',\ldots,\widehat{\zeta}_{N_\ell}+\widehat{\zeta}_{N_\ell}^\prime)\in T^*M^{\ell+2N}\setminus\{0\}\,|\,
		\\(\widehat{x}_\ell,\widehat{\xi}_\ell)\in\operatorname{WF}(\delta_{\mathrm{Diag}_\ell})\,,\,
		(\widehat{z}_{2N},\widehat{\zeta}_{2N})\in\operatorname{WF}(Q^{\otimes N})\,,\,
		\\\forall p\in\{1,\ldots,\ell\}\,
		(x_p,\widehat{z}_{N_p},\xi_p',\widehat{\zeta}_{N_p}')\in\mathrm{C}_{N_p}
		\}\,,
	\end{multline}
	where $\widehat{z}_{2N}=(\widehat{z}_{N_1},\ldots,\widehat{z}_{N_\ell})$.
%	We now proceed as in the discussion of the extension of $P^2$, however, this time we shall use the inductive hypothesis in order to show that $T_J$ is actually well-defined as a distribution in $\mathcal{D}'(M^{\ell+2J}\setminus\mathrm{Diag}_{\ell+2J})$.
	Consider now $\{A,B\}$, a partition of $\{1,\ldots,\ell+2N\}$ -- that is, $\{1,\ldots,\ell+2N\}=A\cup B$, $A\cap B=\emptyset$ -- such that, if $(x_1,\ldots,x_{\ell+2N})=(\widehat{x}_A,\widehat{x}_B)$, then $x_a\neq x_b$ for all $x_a\in \widehat{x}_A$ and for all $x_b\in \widehat{x}_B$. As a consequence, the integral kernel of $T_N$ decomposes as
	\begin{align*}
		T_N(\widehat{x}_A,\widehat{x}_B)
		=K_{N,1}(\widehat{x}_A)S_N(\widehat{x}_A,\widehat{x}_B)K_{J,N}(\widehat{x}_B)\,,
	\end{align*}
	where $S_N$ is a smooth kernel while $K_{N,1}$, $K_{N,2}$ are the integral kernels of distributions appearing in the definition of $\Gamma_{\cdot_Q}$ on $\mathcal{M}^b_a$ for $b<j+1$ and $a\leq k+1$.
	By the inductive hypothesis $K_{N,1}$, $K_{N,2}$ are well-defined and the same holds true for $(K_{J,1}\otimes K_{J,2})\cdot S_J$ on account of Theorem \ref{Thm: WF results} -- \textit{cf.} Equation \eqref{Eq: WF product condition}.
	Accordingly we can conclude that $T_N\in\mathcal{D}'(M^{\ell+2N}\setminus\mathrm{Diag}_{\ell+2N})$. In addition, because of Lemma \ref{Lemma: finite sd of convolution}, of Corollary \ref{Cor: sd of composition} and of the inductive hypothesis on $\Gamma_{\cdot_Q}(\tau_{k_1}),\ldots,\Gamma_{\cdot_Q}(\tau_{k_\ell})$, it holds that
	\begin{align*}
		\operatorname{sd}_{\mathrm{Diag}_{\ell+2N}}(T_N)
		&=\operatorname{sd}_{\mathrm{Diag}_{\ell+2N}}((\delta_{\mathrm{Diag}_\ell}\otimes Q^{\otimes N})\cdot
		(t_{k_1}^{(N_1)}\otimes\ldots\otimes t_{k_\ell}^{(N_\ell)}))
		\\&\leq
		\operatorname{sd}_{\mathrm{Diag}_{\ell+2N}}((\delta_{\mathrm{Diag}_\ell}\otimes Q^{\otimes N})
		+\sum_{i=1}^\ell\operatorname{sd}_{\mathrm{Diag}_{N_i+1}}(t_{k_i}^{(N_i)})<\infty
		\,,
	\end{align*}
	Hence, on account of Theorem \ref{Thm: extension with scaling degree} there exists an extension $\widehat{T}_{N}\in\mathcal{D}'(M^{\ell+2N})$ of $T_N$ with $\operatorname{sd}_{\mathrm{Diag}_{\ell+2N}}(\widehat{T}_N)=\operatorname{sd}_{\mathrm{Diag}_{\ell+2N}}(T_N)$ and with $\operatorname{WF}(\widehat{T}_N)=\operatorname{WF}(T_N)$. Consequently, for all $f\in\mathcal{D}(M)$ and for all $\varphi\in\mathcal{E}(M)$, we can set

	\begin{align}
		\Gamma_{\cdot_Q}(\tau)(f;\varphi)
%		:=\Gamma_{\cdot_Q}(\tau_{k_1})\cdot_Q\cdots_Q\Gamma_{\cdot_Q}(\tau_{k_\ell}) 
		:=	\sum_{\substack{N\geq 0\\N_1+\ldots+ N_\ell=2J}}
		\frac{1}{(2N)!!}\frac{(2N)!}{N_1!\cdots N_\ell!}
		\widehat{T}_N(f\otimes 1_{\ell-1+2N};\varphi)\,.
	\end{align}
	As in the preceding steps, once an extension $\widehat{T}_N$ is chosen, this last formula is well-defined and Remark \ref{Rmk: convolution with smooth function} entails that $\mathcal{D}(M)\ni f\mapsto\widehat{T}_N(f\otimes 1_{\ell+2N};\varphi)\in \mathcal{E}(M)$ -- \textit{cf.} Equations \eqref{Eq: WF of to be extended distribution} and \eqref{Eq: WF convolution bound}. 
	
	Moreover, for all $p\in\mathbb{N}\cup\{0\}$ and $\psi\in\mathcal{E}(M)$, it holds
	\begin{align*}
		\Gamma_{\cdot_Q}(\tau)^{(p)}(f\otimes\psi^{\otimes p};\varphi)
		:=\sum_{\substack{N\geq 0\\N_1+\ldots+ N_\ell=2N}}
		\frac{1}{(2N)!!}\frac{(2N)!}{N_1!\cdots N_\ell!}
		\widehat{T}_N^{(p)}(f\otimes 1_{\ell-1+2N}\otimes\psi^{\otimes p};\varphi)\,,
	\end{align*}
	where $0\leq 2N\leq k+1-p$.
	To ensure that Equation \eqref{Eq: Pi-functional derivative} holds true, consider the formal expression
	\begin{align*}
		&[\Gamma_{\cdot_Q}(\tau_{k_1})\cdot_Q\cdots_Q\Gamma_{\cdot_Q}(\tau_{k_\ell})]^{(p)}(f\otimes\psi^{\otimes p};\varphi)
		\\&:=\sum_{\substack{N\geq 0\\N_1+\ldots+ N_\ell=2N\\p_1+\ldots +p_\ell=p}}\frac{1}{(2N)!!}\frac{(2N)!p!}{N_1!p_1!\cdots N_\ell!p_\ell!}
		[(\delta_{\mathrm{Diag}_\ell}\otimes Q^{\otimes N}\otimes 1_p)
		\cdot( t_{k_1}^{(N_1+p_1)}\otimes\ldots\otimes t_{k_\ell}^{(N_\ell+p_\ell)})]
		(f\otimes 1_{\ell-1+2N}\otimes\psi^{\otimes p};\varphi)
		\\&=
		\sum_{\substack{N\geq 0\\N_1+\ldots+N_\ell=2N\\p_1+\ldots +p_\ell=p}}\frac{1}{(2N)!!}\frac{(2N)!p!}{N_1!p_1!\cdots N_\ell!p_\ell!}
		T_N^{[\widehat{p}_\ell]}
		(f\otimes 1_{\ell-1+2N}\otimes\psi^{\otimes p};\varphi)\,,
	\end{align*}
	where $T_N^{[\widehat{p}_\ell]}$ is a functional-valued distribution on $M^{\ell+2N+p}\setminus\mathrm{Diag}_{\ell+2N+p}$ with finite scaling degree at $\mathrm{Diag}_{\ell+2N+p}$.
	It follows that Equation \eqref{Eq: Pi-functional derivative} is satisfied provided we choose the extension $\widehat{T}_N$ so that
	\begin{align}\label{Eq: condition for Pi-functional derivative}
		\widehat{T}_N^{(p)}
		=\sum_{p_1+\ldots+p_\ell=p}
		\frac{p!}{p_1!\cdots p_\ell!}
		\widehat{T_N^{[\widehat{p}_\ell]}}\,,
	\end{align}
	where $\widehat{T_N^{[\widehat{p}_\ell]}}\in\mathcal{D}'(M^{\ell+2N+p})$ is a scaling degree preserving extension of $T_N^{[\widehat{p}_\ell]}$, whose existence is guaranteed by the finiteness of $\operatorname{sd}_{\mathrm{Diag}_{\ell+2N+p}}(T_N^{[\widehat{p}_\ell]})$.
	Notice that we can impose Equation \eqref{Eq: condition for Pi-functional derivative} on account of the fairly explicit construction of $\widehat{T}_N$ -- \textit{cf.} Theorem \ref{Thm: extension with scaling degree}.
	The proof that $\Gamma_{\cdot_Q}(\tau)\in\mathcal{D}'_{\mathrm{C}}(M;\operatorname{Pol})$ follows by estimating the wave front set of the distribution
	\begin{align*}
		\mathcal{D}(M)\otimes\mathcal{E}(M)^{\otimes p}\ni f\otimes\psi^{\otimes p}
		\to\widehat{T}_N^{(p)}(f\otimes 1_{\ell-1+2N}\otimes\psi^{\otimes p};\varphi)\,,
	\end{align*}
	which is achieved using Theorem \ref{Thm: WF results} -- \textit{cf.} Remark \eqref{Rmk: convolution with smooth function} -- together with the estimate
	\begin{multline*}
		\operatorname{WF}(\widehat{T}_N^{(p)})
		=\{(\widehat{x}_\ell,\widehat{z}_{2N},\widehat{y}_p,\widehat{\xi}_\ell+\widehat{\xi}_\ell',\widehat{\zeta}_{2N}+\widehat{\zeta}_{2N}^\prime,\widehat{\eta}_p+\widehat{\eta}_p^\prime)
		\in T^*M^{\ell+2N+p}\setminus\{0\}\,|\,\\
		(\widehat{x}_\ell,\widehat{\xi}_\ell)\in\operatorname{WF}(\delta_{\mathrm{Diag}_\ell})\,,\,
		(\widehat{z}_{2N},\widehat{\zeta}_{2N})\in\operatorname{WF}(Q^{\otimes N})\,,\\
		\forall h\in\{1,\ldots,\ell\}\,,\,
		(x_h,\widehat{z}_{N_h},\widehat{y}_{p_h},\xi_h,\widehat{\zeta}_{n_h}^\prime,\widehat{\eta}_{p_h}^\prime)\in\mathrm{C}_{N_h+p_h+1}
		\}\,,
	\end{multline*}
	where we wrote $\widehat{\zeta}_{2N}=(\widehat{\zeta}_{N_1},\ldots,\widehat{\zeta}_{N_\ell})$ and similarly for $\widehat{y}_p$, $\widehat{\zeta}_{2N}$ and $\widehat{\eta}_p$.
	Finally the bound \eqref{Eq: Pi scaling degree bound} is satisfied by direct inspection.
	As a matter of fact 
	\begin{align*}
		\sigma_p(\Gamma_{\cdot_Q}(\tau))
		&\leq\max_{0\leq 2N\leq k+1-p}\operatorname{sd}_{\mathrm{Diag}_{\ell+2N+p}}(T_N^{(p)})(\cdot\otimes 1_{\ell-1+2N}\otimes\cdot;\varphi)
		\\&\leq\operatorname{sd}_{\mathrm{Diag}_\ell}(\delta_{\mathrm{Diag}_\ell})
		+\max_{0\leq 2N\leq k+1-p}\left[N\operatorname{sd}_{\mathrm{Diag}_2}(Q)
		+\sum_{i=1}^q\sigma_{1+N_i+p_i}(\Gamma_{\cdot_Q}(\tau_{k_i}))\right]<\infty\,,
	\end{align*}
	where we used the inductive hypothesis on $\mathcal{M}_k$, Examples \ref{Ex: sd of delta on total diagonal} and \ref{Ex: sd of parametrix} as well as (a slight generalization of) Lemma \ref{Lemma: finite sd of convolution}.
%	\begin{align}\label{Eq: sd estimate}
%		\sigma_p(\Gamma_{\cdot_Q}(\tau))
%		&\leq\max_{0\leq 2N\leq k+1-p}\operatorname{sd}_{\mathrm{Diag}_{\ell+2N+p}}(T_N^{(p)})
%		-(\ell-1+2N)d\,,
%	\end{align}
%	where the last factor arises from the integration of $1_{\ell-1+2N}$ inserted in $T_N^{(p)}$.
%	Using the inductive hypothesis on $\mathcal{M}_k$ as well as Examples \ref{Ex: sd of delta on total diagonal} and \ref{Ex: sd of parametrix} as well as Lemma \ref{Lemma: finite sd of convolution}, it holds
%	\begin{align*}
%		\operatorname{sd}_{\mathrm{Diag}_{\ell+2J+p}}(T_N^{(p)})
%		\leq\operatorname{sd}_{\mathrm{Diag}_\ell}(\delta_{\mathrm{Diag}_\ell})
%		+N\operatorname{sd}_{\mathrm{Diag}_2}(Q)
%		+\sum_{i=1}^q\sigma_{1+N_i+p_i}(\Gamma_{\cdot_Q}(\tau_{k_i}))
%		\\&\leq
%		(\ell-1)d
%		+N(d-4)
%		+\sum_{i=1}^\ell\left[
%		(N_i+p_i)d
%		+\frac{k_i-N_i-p_i}{2}(d-4)
%		\right]\,.
%	\end{align*}
%	Inserting this result in Equation \eqref{Eq: sd estimate}, one obtains Equation \eqref{Eq: Pi scaling degree bound} for $\Gamma_{\cdot_Q}(\tau)$.
	This concludes the induction procedure and the proof.
\end{proof}

\begin{remark}[Parabolic Case]\label{Remark: parabolic case of AcdotQ}
	The result of Theorem \ref{Thm: Gamma cdotQ existence} holds true, \emph{mutatis mutandis}, also in the parabolic case $M=\mathbb{R}\times\Sigma$ and $E=\partial_t-\widetilde{E}$. Notice that also in this case renormalization enters the game if $\dim(\Sigma)=d\geq2$.
	The proof goes along the same lines, the main point being the renormalization procedure, namely the extension of the singular distributions involved in the definition of $\cdot_Q$.
	This is based on the finiteness of the weighted scaling degree -- \textit{cf.} Remark \ref{Rmk: weighted scaling degree} and Example \ref{Ex: sd of parametrix}.
\end{remark}

\noindent To conclude the section, we show how we can use $\Gamma_{\cdot_Q}$ to deform the algebra structure of $\mathcal{A}$.

\begin{corollary}\label{Cor: construction of AcdotQ}
	Let $\Gamma_{\cdot_Q}\colon\mathcal{A}\to\mathcal{D}_{\mathrm{C}}'(M;\operatorname{Pol})$ be defined as per Theorem \ref{Thm: Gamma cdotQ existence}.
	Then the vector space $\mathcal{A}_{\cdot_Q}:=\Gamma_{\cdot_Q}(\mathcal{A})\subseteq\mathcal{D}_{\mathrm{C}}'(M;\operatorname{Pol})$ is a unital, commutative and associative $\mathbb{C}$-algebra with respect to the product
	\begin{align}\label{Eq: cdot GammaQ product}
		\tau_1\cdot_{\Gamma_{\cdot_Q}}\tau_2
		:=\Gamma_{\cdot_Q}[\Gamma_{\cdot_Q}^{-1}(\tau_1)\Gamma_{\cdot_Q}^{-1}\tau_2]\,,\qquad
		\forall\tau_1,\tau_2\in\mathcal{A}_{\cdot_Q}\,.
	\end{align}
\end{corollary}

\begin{proof}
	First of all observe that any map $\Gamma_{\cdot_Q}\colon\mathcal{A}\to\mathcal{D}_{\mathrm{C}}'(M;\operatorname{Pol})$ built as per Theorem \ref{Thm: Gamma cdotQ existence} is such that $\ker\Gamma_{\cdot_Q}=\{0\}$.
	Indeed, let $\tau\in\ker\Gamma_{\cdot_Q}\setminus\{0\}$ be of polynomial degree $k$ in $\Phi$, that is $\delta^k_\psi\tau\neq 0$ while $\delta^{k+1}_\psi\tau=0$ -- {\it cf.} Remark \ref{Rem: graded algebra}.
	Equation \eqref{Eq: Gamma on M1} entails that $\mathbf{1}\notin\ker\Gamma_{\cdot_Q}$, so that $k>0$. In turn this implies that for all $\psi\in\mathcal{E}(M)$ there exists $0\neq f_\psi\in\mathcal{E}(M^k)$ such that
	\begin{align*}
		\delta^k_\psi\tau
		=f_\psi\mathbf{1}\,.
	\end{align*}	
	Moreover Equation \eqref{Eq: Pi-functional derivative} implies that, if $\tau\in\ker\Gamma_{\cdot_Q}$, then also $\delta_\psi\tau\in\ker\Gamma_{\cdot_Q}$. This entails that $f_\psi\mathbf{1}\in\ker\Gamma_{\cdot_Q}$ which implies $f_\psi=0$. This is a contradiction. Thus $\Gamma_{\cdot_Q}$ is injective and therefore $\cdot_{\Gamma_{\cdot_Q}}$ is well-defined. All remaining properties are straightforwardly verified.
	In particular $\boldsymbol{1}\cdot_{\Gamma_{\cdot_Q}}\tau=\tau\cdot_{\Gamma_{\cdot_Q}}\boldsymbol{1}=\tau$ for all $\tau\in\mathcal{A}_{\cdot_{\Gamma_{\cdot_Q}}}$ while associativity and commutativity of $\mathcal{A}_{\cdot_{\Gamma_{\cdot_Q}}}$ are inherited from the corresponding properties of $\mathcal{A}$.
\end{proof}

We refer to \cite{EFPTZ} for a Hopf-algebraic approach to deformations of the form \eqref{Eq: cdot GammaQ product}, based on moment-cumulants relations.

\begin{remark}\label{Rmk: co(variance)mment}
	Theorem \ref{Thm: Gamma cdotQ existence} and Corollary \ref{Cor: construction of AcdotQ} identify a suitable algebra $\mathcal{A}_{\cdot_Q}$ whose physical relevance has already been explained.
	One may wonder which is the interplay between this construction and diffeomorphism invariance, namely to which extent the algebra $\mathcal{A}_{\cdot_Q}=\mathcal{A}_{\cdot_Q}(M,E)$ depends on the geometrical data $M,E$.
	
	In particular it would be desirable that the assignment $(M,E)\to\mathcal{A}_{\cdot_Q}(M,E)$ satisfies the following property: Whenever $\iota\colon M_1\to M_2$ is a smooth map such that $E_1\circ\iota^*=\iota^*\circ E_2$ then there exists a corresponding injective isomorphism of algebras $\mathcal{A}_{\cdot_Q}(\iota)\colon\mathcal{A}_{\cdot_Q}(M_1,E_1)\to\mathcal{A}_{\cdot_Q}(M_2,E_2)$.  
	
	This statement can be read as the translation to this setting of the principle of general covariance which is adopted in algebraic quantum field theory, {\it cf.} \cite[Ch. 4]{BFDY15} and which is often stated using the language of category theory. Here, we avoid entering into the details, and we just observe that, similarly to the analysis in \cite{DDR20}, such kind of property is hard to implement as our construction depends on the choice of a parametrix $P$ for the operator $E$. It is well-known that the choice of $P$ is not covariant and this has a repercussion in the failure of $\mathcal{A}_{\cdot_Q}$ being invariant under the induced action of the diffeomorphism group of the underlying manifold $M$. Following the same rationale of \cite{DDR20} a possible way out from this quandary consists of working with all possible parametrices $P$ at once. We feel that, pursuing such path in this paper would only be an unnecessary detour from our real goal and we shall address it in a future work.
\end{remark}

\section{Correlations and $\bullet_Q$ product}\label{Sec: Correlations and bulletQ product}

In the previous section Theorem \ref{Thm: Gamma cdotQ existence} together with Corollary \ref{Cor: construction of AcdotQ} laid the foundation of the algebra $\mathcal{A}_{\cdot_Q}$, whose elements can be interpreted as the expectation value $\mathbb{E}(\widehat{\tau}(f))$ of the ($\varphi$-shifted) random variable $\widehat{\tau}(f)$, \textit{cf.} Proposition \ref{Prop: Dc functionals algebraic structure} and Remark \ref{Rmk: interpretation of Dc local product}. This construction falls short of the goal of computing also the correlations of $\widehat{\tau}$ in terms of $\Gamma_{\cdot_Q}(\tau)$.

Proposition \ref{Prop: Dc tensor algebra algebraic structure} solves this quandary in the case of a regularized counterpart for $Q$, \textit{cf.} Remark \ref{Rmk: interpretation of Dc tensor product}. In this section we adopt a strategy similar to that of the preceding one showing that, barring a suitable renormalization procedure, it is possible to replace $Q_\varepsilon$ with $Q$. Observe that the main hurdles arise from the singular behaviour of $Q$ which leads to a generally ill-defined product $(1_2\otimes Q^{\otimes k})\cdot(\tau_1^{(k)}\otimes \tau_2^{(k)})$, $\tau_1,\tau_2\in\mathcal{D}^\prime(M;\operatorname{Pol})$. The following examples illustrate in a concrete case the problem that we have outlined.

\begin{example}\label{Exam: problem with correlations}
	Consider $\Phi^2\in\mathcal{D}_{\mathrm{C}}'(M;\operatorname{Pol})$ as per Example \ref{Ex: examples of functional-valued polynomial distributions}. Equation \eqref{Eq: regularized multilocal delta-product} with $Q_\varepsilon$ formally replaced with $Q$ leads to
	\begin{gather*}
		[\Phi^2\bullet_Q \Phi^2](f_1\otimes f_2;\varphi)
		=\\=\int\limits_{M\times M}
		f_1(x_1)f_2(x_2)\big[
		\varphi(x_1)^2\varphi(x_2)^2
		+4\varphi(x_1)Q(x_1,x_2)\varphi(x_2)
		+2Q(x_1,x_2)^2
		\big]\mathrm{d}\mu(x_1)\mathrm{d}\mu(x_2),
	\end{gather*}
	where $f_1,f_2\in\mathcal{D}(M)$, while $\varphi\in\mathcal{E}(M)$. One can realize by direct inspection that the last term on the right hand side of the previous identity is a priori well-defined only outside the total diagonal of $M\times M$. In other words $Q^2\in\mathcal{D}'(M\times M\setminus\mathrm{Diag}_2)$.
	
Yet, observe that, on account of Remark \ref{Rmk: sd product estimate}, Lemma \ref{Lemma: finite sd of convolution} and Corollary \ref{Cor: sd of composition},
	\begin{align*}
		\operatorname{sd}_{\mathrm{Diag}_2}(Q^2)
		\leq 2\operatorname{sd}_{\mathrm{Diag}_2}(Q)<\infty\,.
	\end{align*}
	Hence, in view of Theorem \ref{Thm: extension with scaling degree}, there exists $\widehat{Q}_2\in\mathcal{D}^\prime(M\times M)$, an extension of $Q^2$ with $\operatorname{sd}_{\mathrm{Diag}_2}(\widehat{Q}_2)=\operatorname{sd}_{\mathrm{Diag}_2}(Q^2)$. As a matter of fact, exploiting the local behaviour of the parametrix $P$ as well as Example \ref{Ex: sd of parametrix}, one has $\mathrm{sd}_{\mathrm{Diag}_2}(Q)\leq d-4$. Thus if $\dim M=d<8$ the above extension is unique, {\it cf.} Theorem \ref{Thm: extension with scaling degree}. Following the same strategy as in the preceding section, we can conceive to use such extension to give meaning to $\Phi^2\bullet_Q \Phi^2$. The remainder of the section is devoted to making this idea precise. 
\end{example}

\begin{remark}[Parabolic Case]
The same strategy of Example \ref{Exam: problem with correlations} applies in the parabolic case. For the sake of clarity, we underline that the only difference is that $\operatorname{wsd}_{\mathrm{Diag}_2}(Q^2)=2(d-3)$, {\it cf.} Remark \ref{Rmk: weighted scaling degree}, where $d=\dim M$ As a consequence, for the same reason of the above example, the extension preserving the weighted scaling degree $\widehat{Q}_2\in\mathcal{D}^\prime(M\times M)$ of $Q^2\in\mathcal{D}'(M\times M\setminus\mathrm{Diag}_2)$ is unique if $\dim(\Sigma)<6$ where $M=\mathbb{R}\times\Sigma$.
\end{remark}

\begin{remark}\label{Rmk: decomposition of Pi tensor algebra}
	Notice that the decomposition $\mathcal{A}=\bigoplus\limits_{k\geq 0}\mathcal{M}_k=\varinjlim\limits_j\bigoplus\limits_{k\geq 0}\mathcal{M}_k^j$ -- \textit{cf.} Remark \ref{Rem: graded algebra} -- induces a counterpart at the level of the universal tensor module $\mathcal{T}(\mathcal{A}_{\cdot_Q})=\mathcal{E}(M)\oplus\bigoplus_{\ell>0}\mathcal{A}_{\cdot_Q}^{\otimes\ell}$, {\it i.e.}
	\begin{align*}
		\mathcal{T}(\mathcal{A}_{\cdot_Q})
		&=\mathcal{E}(M)\oplus\bigoplus_{\ell> 0}
		\bigoplus_{k=0}^\infty\bigoplus_{\substack{k_1,\dots,k_\ell\\ k_1+\dots +k_\ell=k}}
		\Gamma_{\cdot_Q}(\mathcal{M}_{k_1})\otimes\ldots\otimes\Gamma_{\cdot_Q}(\mathcal{M}_{k_\ell})
		\\ 
		&=\mathcal{E}(M)\oplus\bigoplus_{\ell> 0}
		\bigoplus_{k=0}^\infty\bigoplus_{\substack{k_1,\dots,k_\ell\\ k_1+\dots +k_\ell=k}}
		\varinjlim\limits_{j_1,\dots,j_\ell}
		\Gamma_{\cdot_Q}(\mathcal{M}_{k_1}^{j_1})\otimes\ldots\otimes\Gamma_{\cdot_Q}(\mathcal{M}_{k_\ell}^{j_\ell})\,.
	\end{align*}
\end{remark}

%Theorem \ref{Thm: GammabulletQ existence} proves the existence of a suitable product $\bullet_{\Gamma_{\bullet_Q}}$ whose definition is inspired by $\bullet_Q$.
In the following we prove a counterpart of Theorem \ref{Thm: Gamma cdotQ existence} for the universal tensor module, namely we look for a map  $\Gamma_{\bullet_Q}\colon\mathcal{T}(\mathcal{A}_{\cdot_Q})\to\mathcal{T}_{\mathrm{C}}'(M;\operatorname{Pol})$ out of which we can define a deformed algebra structure on $\mathcal{T}(\mathcal{A}_{\cdot_Q})$ along the same spirit of Corollary \ref{Cor: construction of AcdotQ}.

\begin{theorem}\label{Thm: GammabulletQ existence}
	Let $\mathcal{A}_{\cdot_Q}$ be the algebra defined in Corollary \ref{Cor: construction of AcdotQ} with the map $\Gamma_{\cdot_Q}$ built as per Theorem \ref{Thm: Gamma cdotQ existence}. Letting $\mathcal{T}_{\mathrm{C}}'(M;\operatorname{Pol})$ be defined as per Equation \eqref{Eq: Tc algebra} and calling $\mathcal{T}(\mathcal{A}_{\cdot_Q})$ the universal tensor module built out of $\mathcal{A}_{\cdot_Q}$, there exists a linear map $\Gamma_{\bullet_Q}\colon\mathcal{T}(\mathcal{A}_{\cdot_Q})\to\mathcal{T}_{\mathrm{C}}'(M;\operatorname{Pol})$ with the following properties:
	\begin{enumerate}[(i)]
		\item
		for all $\tau_1,\ldots,\tau_\ell\in\mathcal{A}_{\cdot_Q}$ with $\tau_1\in\Gamma_{\cdot_Q}(\mathcal{M}_1)$ it holds
		\begin{align}\label{Eq: GammabulletQ with one factor in M1}
			\Gamma_{\bullet_Q}(\tau_1\otimes\ldots\otimes \tau_\ell)
			:=\tau_1\bullet_Q\Gamma_{\bullet_Q}(\tau_2\otimes\ldots\otimes \tau_\ell)\,,
		\end{align}
		where $\bullet_Q$ is defined as in Equation \eqref{Eq: regularized multilocal delta-product} with $Q_\varepsilon$ replaced by $Q$.
		\item
		Let $\tau_1,\ldots,\tau_\ell\in\mathcal{A}_{\cdot_Q}$ and $f_1,\ldots,f_\ell\in\mathcal{D}(M)$.
		If there exists $I\subsetneq\{1,\ldots,\ell\}$ for which
		\begin{align*}
			\bigcup_{i\in I}\operatorname{supp}(f_i)\cap\bigcup_{j\notin I}\operatorname{supp}(f_j)
			=\emptyset\,,
		\end{align*}
		then
		\begin{align}\label{Eq: smooth factorization}
			\Gamma_{\bullet_Q}(\tau_1\otimes\ldots\otimes\tau_\ell)(f_1\otimes\ldots\otimes f_\ell)
			=\bigg[\Gamma_{\bullet_Q}\bigg(\bigotimes_{i\in I}\tau_i\bigg)
			\bullet_Q\Gamma_{\bullet_Q}\bigg(\bigotimes_{j\notin I}\tau_j\bigg)\bigg](f_1\otimes\ldots\otimes f_\ell)\,.
		\end{align}
		\item
		for all $\ell\geq 0$, $\Gamma_{\bullet_Q}\colon\mathcal{A}_{\cdot_Q}^{\otimes \ell}\to\mathcal{T}_{\mathrm{C}}'(M;\operatorname{Pol})$ is a symmetric map,
		\item\label{Item: GammabulletQ conditions}
		$\Gamma_{\bullet_Q}$ satisfies the following identities:
		\begin{subequations}\label{Eq: GammabulletQ conditions}
			\begin{align}
				\label{Eq: GammabulletQ on AcdotQ}
				&\Gamma_{\bullet_Q}(\tau)
				=\tau\,,
				\qquad\forall \tau\in\mathcal{A}_{\cdot_Q}\,,\\
				\label{Eq: GammabulletQ functional derivative requirement}
				&\Gamma_{\bullet_Q}\circ\delta_\psi
				=\delta_\psi\circ\Gamma_{\bullet_Q}\,,
				\qquad\forall \psi\in\mathcal{E}(M)\,,\\
				\label{Eq: GammabulletQ P-condition}
				&\Gamma_{\bullet_Q}(\tau_1\otimes\ldots\otimes P\circledast\tau_j\otimes\ldots\otimes \tau_\ell)
				=(\delta_{\mathrm{Diag}_2}^{\otimes j-1}\otimes P\otimes\delta_{\mathrm{Diag}_2}^{\otimes \ell-j})\circledast
				\Gamma_{\bullet_Q}(\tau_1\otimes\ldots\otimes \tau_j\otimes\ldots\otimes \tau_\ell)\,,
			\end{align}
		\end{subequations}
		for all $\tau_1,\ldots,\tau_\ell\in \mathcal{A}_{\cdot_Q}$ and for all $\ell\in\mathbb{Z}_+$.	
	\end{enumerate}
	Moreover, given any such map $\Gamma_{\bullet_Q}$ let
	\begin{align}\label{Eq: bulletQ algebra}
		\mathcal{A}_{\bullet_Q}
		:=\Gamma_{\bullet_Q}(\mathcal{A}_{\cdot_Q})
		\subseteq\mathcal{T}_{\mathrm{C}}'(M;\operatorname{Pol})\,.
	\end{align}
	Then the bilinear map $\bullet_{\Gamma_{\bullet_Q}}\colon\mathcal{A}_{\bullet_Q}\times\mathcal{A}_{\bullet_Q}\to\mathcal{A}_{\bullet_Q}$ defined by
	\begin{align}\label{Eq: extended multilocal delta product}
		 \tau\bullet_{\Gamma_{\bullet_Q}}\bar{\tau}
		:=\Gamma_{\bullet_Q}(\Gamma_{\bullet_Q}^{-1}(\tau)\otimes \Gamma_{\bullet_Q}^{-1}(\bar{\tau}))\,,
		\qquad\forall\tau,\bar{\tau}\in\mathcal{A}_{\bullet_Q}\,,
	\end{align}
	defines a unital, commutative and associative product on $\mathcal{A}_{\bullet_Q}$.
\end{theorem}

\begin{proof}
	The strategy of the proof is very much similar in spirit to that of Theorem \ref{Thm: Gamma cdotQ existence}, namely we proceed inductively. Due to the sheer length of the analysis we divide what follows in separate steps. As a preliminary observation we stress that Equation \eqref{Eq: GammabulletQ with one factor in M1} and the map $\bullet_Q$ in particular are well-defined when applied to elements lying in $\Gamma_{\cdot_Q}(\mathcal{M}_1)$, since no divergences occur. In particular all equations in item \eqref{Item: GammabulletQ conditions} are automatically fulfilled.
	In addition notice that, under the assumption made on the test functions $f_1,\ldots,f_\ell$, the product $\bullet_Q$ appearing in Equation \eqref{Eq: smooth factorization} is well-defined.
		
	\paragraph{Step 1 -- Induction over $\ell$.}
	The first step consists of controlling $\Gamma_{\bullet_Q}$ as the number of arguments in the tensor product increases -- recall that $\mathcal{T}(\mathcal{A}_{\cdot_Q})=\mathcal{E}(M)\bigoplus_{\ell=1}^\infty\mathcal{A}_{\cdot_Q}^{\otimes \ell}$. In other words we are proceeding inductively over $\ell$ and we observe that Equation \eqref{Eq: GammabulletQ on AcdotQ} defines completely $\Gamma_{\bullet_Q}$ when $\ell=1$. As per the inductive hypothesis we assume that the action of $\Gamma_{\bullet_Q}$ has been defined on $\mathcal{A}_{\cdot_Q}^{\otimes p}\subseteq\mathcal{T}_{\mathrm{C}}'(M;\operatorname{Pol})$ for all $p<\ell$ and we prove existence of $\Gamma_{\bullet_Q}$ on $\mathcal{A}_{\cdot_Q}^{\otimes \ell}$.
	Since, on account of Remark \ref{Rmk: decomposition of Pi tensor algebra},
	\begin{align*}
		\mathcal{A}_{\cdot_Q}^{\otimes \ell}
		=\bigoplus_{k\geq 0}\bigoplus_{\substack{k_1,\dots k_l\\k_1+\ldots+k_\ell= k}}
		\Gamma_{\cdot_Q}(\mathcal{M}_{k_1})\otimes\ldots\otimes\Gamma_{\cdot_Q}(\mathcal{M}_{k_\ell})\,,
	\end{align*}
	the problem reduces to constructing, for all $k\in\mathbb{N}\cup\{0\}$, $\Gamma_{\bullet_Q}(\tau_{k_1}\otimes\ldots\otimes \tau_{k_\ell})$ for all $\tau_{k_1},\ldots,\tau_{k_\ell}$ with $\tau_{k_p}\in\Gamma_{\cdot_Q}(\mathcal{M}_{k_p})$, while $k_1,\ldots,k_p\in\mathbb{N}\cup\{0\}$ are such that $k_1+\ldots+k_\ell= k$.
	
	\paragraph{Step 2 -- Induction over $k$.}
	At this stage we proceed inductively over $k\in\mathbb{N}\cup\{0\}$. The cases $k=0,1$ are readily verified since no singularity can occur in 
	\begin{align*}
		\Gamma_{\bullet_Q}(\tau_{k_1}\otimes\ldots\otimes\tau_{k_\ell})
		=\tau_{k_1}\bullet_Q\ldots\bullet_Q\tau_{k_\ell}\,.
	\end{align*}
	Furthermore both items $(ii)$ and $(iii)$ are trivially satisfied per construction. Hence we can make the inductive hypothesis assuming that $\Gamma_{\bullet_Q}(\tau_{k_1}\otimes\ldots\otimes\tau_{k_\ell})$ has been defined for all $k_1,\ldots,k_\ell$ such that $k_1+\ldots+k_\ell=k-1$. In order to extend the definition of $\Gamma_{\bullet_Q}$ to the case when $k_1+\dots k_\ell=k$ we can invoke Remark \ref{Rmk: decomposition of Pi tensor algebra} to write
	\begin{align*}
		\Gamma_{\cdot_Q}(\mathcal{M}_{k_1})\otimes\ldots\otimes\Gamma_{\cdot_Q}(\mathcal{M}_{k_\ell})
		=\varinjlim_{j_1,\ldots,j_\ell}\Gamma_{\cdot_Q}(\mathcal{M}^{k_1}_{j_1})\otimes\ldots\otimes\Gamma_{\cdot_Q}(\mathcal{M}^{k_\ell}_{j_\ell})\,,
	\end{align*}
	proceeding inductively over $j_1,\ldots,j_\ell$.
	
	\paragraph{Step 2a -- Induction over $j_1,\ldots,j_\ell$: starting case.}
	
	If $j_1=\ldots=j_\ell=0$, this amounts to considering only the case $\tau_{k_p}=\Gamma_{\cdot_Q}(\Phi^{k_p})$ for all $p\in\{1,\ldots,\ell\}$ -- \textit{cf.} Definition \ref{Def: pointwise algebra}.
	For all $f_i\in\mathcal{D}(M)$, $i=1,\dots,\ell$ with disjoint supports and for all $\varphi\in\mathcal{E}(M)$, Equation \eqref{Eq: smooth factorization} entails, together with Equation \eqref{Eq: regularized multilocal delta-product} with $Q_\varepsilon$ replaced by $Q$ reads for the case in hand,
	\begin{multline*}
		\Gamma_{\bullet_Q}(\Gamma_{\cdot_Q}(\Phi^{k_1})\otimes\ldots\otimes\Gamma_{\cdot_Q}(\Phi^{k_\ell}))(f_1\otimes\ldots\otimes f_\ell;\varphi)
		=\Gamma_{\cdot_Q}(\Phi^{k_1})\bullet_Q\ldots\bullet_Q\Gamma_{\cdot_Q}(\Phi^{k_\ell})(f_1\otimes\ldots\otimes f_\ell;\varphi)
		\\=\sum_{N=0}^\infty\frac{(2N)!}{(2N)!!}\sum\limits_{\substack{N_1,\dots N_\ell\\ N_1+\dots +N_\ell=2N}}
		{\widehat{k}_\ell \choose \widehat{N}_\ell}
		(1_\ell\otimes Q^{\otimes N})\cdot\\
\cdot	[\delta_{\mathrm{Diag}_{N_1+1}}\Gamma_{\cdot_Q}(\Phi^{k_1-N_1})\otimes\ldots\otimes\delta_{\mathrm{Diag}_{N_\ell+1}}\Gamma_{\cdot_Q}(\Phi^{k_\ell-N_\ell})]
		(f_1\otimes\ldots\otimes f_\ell\otimes 1_{2N};\varphi)\,,
	\end{multline*}
	where $\cdot$ denotes the product of distributions,
%	\footnote{
%		With unavoidable confusion we specify that we are considering the product between the (functional-valued) distribution $\delta_{\mathrm{Diag}_{j_1+1}}\Gamma_{\cdot_Q}(\Phi^{k_1-j_1})\otimes\ldots\otimes\delta_{\mathrm{Diag}_{j_1+1}}\Gamma_{\cdot_Q}^{k_\ell-j_\ell}$ with integral kernel $\delta_{\mathrm{Diag}_{j_1+1}}(x_1,\widehat{z}_{j_1})\Gamma_{\cdot_Q}(\Phi^{k_1-j_1})(x_1;\varphi)\cdots\delta_{\mathrm{Diag}_{j_1+1}}(x_\ell,\widehat{z}_{j_\ell})\Gamma_{\cdot_Q}^{k_\ell-j_\ell}(x_\ell;\varphi)$ and the distribution $1_\ell\otimes Q^{\otimes J}$ with integral kernel $1_\ell(\widehat{x}_\ell)\mathcal{Q}(\widehat{z}_{2J})$, where $\mathcal{Q}(\widehat{z}_{2J})$ denotes the integral kernel of $Q^{\otimes J}$ after a suitable rearrangement of the components in $\widehat{z}_{2J}=(\widehat{z}_{j_1},\ldots,\widehat{z}_{j_\ell})=(z_1,\ldots,z_{2J})$.
%		.},
	while ${\widehat{k}_\ell \choose \widehat{N}_\ell}:=\prod_{p=1}^\ell{k_p \choose N_p}$ and
	\begin{align*}
		\Gamma_{\cdot_Q}(\Phi^k)^{(j)}=\frac{k!}{(k-j)!}\Gamma_{\cdot_Q}(\Phi^{k-j})\delta_{\mathrm{Diag}_{j+1}}\,.
	\end{align*}
	The previous expression defines $\Gamma_{\bullet_Q}(\Gamma_{\cdot_Q}(\Phi^{k_1})\otimes\ldots\otimes\Gamma_{\cdot_Q}(\Phi^{k_\ell}))$ as a (functional-valued) distribution on $M^\ell\setminus\mathrm{Diag}_\ell$.
	This is codified by the fact that
	\begin{align*}
		S_{\ell,N}:=(1_\ell\otimes Q^{\otimes N})\cdot
		[\delta_{\mathrm{Diag}_{N_1+1}}\otimes\ldots\otimes\delta_{\mathrm{Diag}_{N_\ell+1}}]\,,
	\end{align*}
	is a well-defined functional-valued distribution only on $M^{\ell+2N}\setminus\mathrm{Diag}_{\ell+2N}$.
	To complete the definition of $\Gamma_{\bullet_Q}(\Gamma_{\cdot_Q}(\Phi^{k_1})\otimes\ldots\otimes\Gamma_{\cdot_Q}(\Phi^{k_\ell}))$ we observe that $\operatorname{sd}_{\mathrm{Diag}_{\ell+2N}}(S_{\ell,N})$ is finite and, thus, Theorem \ref{Thm: extension with scaling degree} ensures the existence of an extension $\widehat{S}_{\ell,N}$ of $S_{\ell,N}$ to the whole $M^{\ell+2N}$ such that $\operatorname{sd}_{\mathrm{Diag}_{\ell+2N}}(\widehat{S}_{\ell,N})=\operatorname{sd}_{\mathrm{Diag}_{\ell+2N}}(S_{\ell,N})$.
	Having chosen one such extension we set for all $f_i\in\mathcal{D}(M)$, $i=1,\dots,\ell$ and for all $\varphi\in\mathcal{E}(M)$
	\begin{multline*}
		\Gamma_{\bullet_Q}(\Gamma_{\cdot_Q}(\Phi^{k_1})\otimes\ldots\otimes\Gamma_{\cdot_Q}(\Phi^{k_\ell}))(f_1\otimes\ldots\otimes f_\ell;\varphi)
		\\=\sum_{\substack{N\geq 0\\N_1+\ldots+N_\ell=2N}}
		\frac{(2N)!}{(2N)!!}
		{\widehat{k}_\ell \choose \widehat{N}_\ell}
		\widehat{S}_{\ell,N}\cdot
		[\Gamma_{\cdot_Q}(\Phi^{k_1-N_1})\otimes\ldots\otimes\Gamma_{\cdot_Q}(\Phi^{k_\ell-N_\ell})]
		(f_1\otimes\ldots\otimes f_\ell\otimes 1_{2N};\varphi)\,,
	\end{multline*}
	which identifies an element in $\mathcal{D}'(M^\ell;\operatorname{Pol})$.
	Furthermore for all $p\in\mathbb{N}\cup\{0\}$ and $\psi\in\mathcal{E}(M)$, it holds
	\begin{multline*}
	\Gamma_{\bullet_Q}(\Gamma_{\cdot_Q}(\Phi^{k_1})\otimes\ldots\otimes\Gamma_{\cdot_Q}(\Phi^{k_\ell}))^{(p)}
	(f_1\otimes\ldots\otimes f_\ell\otimes\psi^{\otimes p};\varphi)=
	\sum_{\substack{N\geq 0\\N_1+\ldots+N_\ell=2N\\p_1+\ldots+p_\ell=p}}
	\frac{(2N)!p!}{(2N)!!}
	\prod_{h=1}^qC(k_h,N_h,p_h)\cdot\\\cdot
	\widehat{S}_{\ell,N}^{[p]}\cdot
	[\Gamma_{\cdot_Q}(\Phi^{k_1-N_1-p_1})\otimes\ldots\otimes\Gamma_{\cdot_Q}(\Phi^{k_\ell-N_\ell-p_\ell})]
	(f_1\otimes\ldots\otimes f_\ell\otimes 1_{2N}\otimes\psi^{\otimes p};\varphi)\,,
\end{multline*}
where $C(k_h,N_h,p_h)=\frac{k_h!}{p_h!N_h!(k_h-N_h-p_h)!}$ while we denoted by $\widehat{S}_{\ell,N}^{[p]}$ the distribution
	\begin{align*}
		\widehat{S}_{\ell,N}^{[p]}
		:=(\widehat{S}_{\ell,N}\otimes 1_p)\cdot[1_N\otimes\delta_{\mathrm{Diag}_{p_1+1}}\otimes\ldots\otimes\delta_{\mathrm{Diag}_{p_\ell+1}}]
		\supseteq\\
		\supseteq(1_\ell\otimes Q^{\otimes N}\otimes 1_p)\cdot[\delta_{\mathrm{Diag}_{p_1+N_1+1}}\otimes\ldots\otimes\delta_{\mathrm{Diag}_{p_\ell+N_\ell+1}}]\,.
	\end{align*}
	Here $\supseteq$ means that the two distributions coincide on $\mathcal{D}(M^{\ell+2N+p}\setminus\mathrm{Diag}_{\ell+2N+p})$.
	In addition it holds that $\Gamma_{\bullet_Q}(\Gamma_{\cdot_Q}(\Phi^{k_1})\otimes\ldots\otimes\Gamma_{\cdot_Q}(\Phi^{k_\ell}))\in\mathcal{D}_{\mathrm{C}}'(M^\ell;\operatorname{Pol})$. This follows by estimating the wave front set of the distribution
	\begin{align*}
	\mathcal{D}(M)^{\otimes\ell}\otimes\mathcal{E}(M)^{\otimes p}\ni f_1\otimes\ldots\otimes f_\ell\otimes\psi^{\otimes p}\mapsto
		\widehat{S}_{\ell,N}^{[p]}(f_1\otimes\ldots\otimes f_\ell\otimes 1_{2N}\otimes\psi^{\otimes p})\,,
	\end{align*}
	by means of Theorem \ref{Thm: WF results} and Remark \ref{Rmk: convolution with smooth function}.
Furthermore a direct application of Equation \eqref{Eq: Pi-functional derivative} entails that, for all $f_i\in\mathcal{D}(M)$, $i=1,\dots,\ell$ and for all $\varphi\in\mathcal{E}(M)$
	\begin{multline*}
		\Gamma_{\bullet_Q}(\Gamma_{\cdot_Q}(\Phi^{k_1})\otimes\ldots\otimes\Gamma_{\cdot_Q}(\Phi^{k_\ell}))^{(p)}
		(f_1\otimes\ldots\otimes f_\ell\otimes\psi^{\otimes p};\varphi)
		\\=\sum_{\substack{N\geq 0\\N_1+\ldots+N_\ell=2N\\p_1+\ldots+p_\ell=p}}
		\frac{p!}{p_1!\cdots p_\ell!}
		{\widehat{k}_\ell \choose \widehat{N}_\ell}
		\widehat{S}_{\ell,N}\cdot
		[\Gamma_{\cdot_Q}(\delta_\psi^{p_1}\Phi^{k_1-N_1})\otimes\ldots\otimes\Gamma_{\cdot_Q}(\delta_\psi^{p_\ell}\Phi^{k_\ell-N_\ell})]
		(f_1\otimes\ldots\otimes f_\ell\otimes 1_{2N};\varphi)\,,
	\end{multline*}
	which entails Equation \eqref{Eq: GammabulletQ functional derivative requirement}.
	This completes the inductive proof for the initial case $n_1=\ldots=n_\ell=0$.
	
	\paragraph{Step 2b -- Induction over $j_1,\ldots,j_\ell$: general case.}
	In view of the preceding step we can assume that $\Gamma_{\bullet_Q}(\tau_{k_1}\otimes\ldots\otimes \tau_{k_\ell})$ has been defined for all $\tau_{k_p}\in\Gamma_{\cdot_Q}(\mathcal{M}_{k_p}^{m_p})$ for $m_p<j_p$. Our goal is to show how to extend the definition of $\Gamma_{\bullet_Q}$ to the case $m_p=j_p$ for all $p\in\{1,\ldots,\ell\}$.
	
	Let $\tau_{k_1},\ldots,\tau_{k_\ell}$ be such that $\tau_{k_p}\in\Gamma_{\cdot_Q}(\mathcal{M}_{k_p}^{j_p})$ for all $p\in\{1,\ldots,\ell\}$ and $k_1+\ldots+k_\ell=k$. 	If $\tau_{k_p}=P\circledast\sigma_{k_p}$ for $\sigma_{k_p}\in\mathcal{M}_{k_p}^{j_p-1}$, then $\Gamma_{\bullet_Q}(\tau_{k_1}\otimes\ldots\otimes \tau_{k_\ell})$ is defined via Equation \eqref{Eq: GammabulletQ P-condition}.
	In the general case, for all $f_1,\dots,f_\ell\in\mathcal{D}(M)$ with disjoint supports and for all $\varphi\in\mathcal{E}(M)$, Equation \eqref{Eq: smooth factorization} entails
	\begin{align*}
		&\Gamma_{\bullet_Q}(\tau_{k_1}\otimes\ldots\otimes\tau_{k_\ell})(f_1\otimes\ldots\otimes f_\ell;\varphi)
		=(\tau_{k_1}\bullet_Q\ldots\bullet_Q \tau_{k_\ell})(f_1\otimes\ldots\otimes f_\ell;\varphi)
		\\&=\sum\limits_{N=0}^\infty
		\frac{1}{(2N)!!}
		\sum_{\substack{N_1,\dots N_\ell\\N_1+\ldots+ N_\ell=2N}}
		\frac{(2N)!}{N_1!\cdots N_\ell!}
		\big[
		(1_\ell\otimes Q^{\otimes J})\cdot
		(\tau_{k_1}^{(N_1)}\widetilde{\otimes}\ldots\widetilde{\otimes}\tau_{k_\ell}^{(N_\ell)})
		\big](f_1\otimes\ldots\otimes f_\ell\otimes 1_{2N};\varphi)\\
		&=\sum\limits_{N=0}^\infty
		\frac{1}{(2N)!!}\sum_{\substack{N_1,\dots N_\ell\\N_1+\ldots+ N_\ell=2N}}
		\frac{(2N)!}{N_1!\cdots N_\ell!}
		S_{\ell,N}(f_1\otimes\ldots\otimes f_\ell\otimes 1_{2N};\varphi)\,,
	\end{align*}
	where $\cdot$ still denotes the product of distributions. As before, this defines $\Gamma_{\bullet_Q}(\tau_{k_1}\otimes\ldots\otimes\tau_{k_\ell})$ as a functional-valued distribution on $M^\ell\setminus\mathrm{Diag}_\ell$.
	Notice that with the same argument used in the proof of Theorem \ref{Thm: Gamma cdotQ existence} we can infer that $S_{\ell,N}$ is a well-defined functional-valued distribution on $M^{\ell+2N}\setminus\mathrm{Diag}_{\ell+2N}$.
	
	To complete the definition of $\Gamma_{\bullet_Q}(\tau_{k_1}\otimes\ldots\otimes\tau_{k_\ell})$ we ought to discuss a suitable extension of $S_{\ell,N}\in\mathcal{D}'(M^{\ell+2N}\setminus\mathrm{Diag}_{\ell+2N})$.
	Since $\operatorname{sd}_{\mathrm{Diag}_{\ell+2N}}(S_{\ell,N})<+\infty$, we can apply Theorem \ref{Thm: extension with scaling degree} to ensure the existence of an extension $\widehat{S}_{\ell,N}$ of $S_{\ell,N}$ to the whole $M^{\ell+2J}$ such that $\operatorname{sd}_{\mathrm{Diag}_{\ell+2N}}(\widehat{S}_{\ell,N})=\operatorname{sd}_{\mathrm{Diag}_{\ell+2N}}(S_{\ell,N})$.
	Having chosen one such extension, we set for all $f_1,\dots,f_\ell\in\mathcal{D}(M)$ and for all $\varphi\in\mathcal{E}(M)$
	\begin{align*}
		\Gamma_{\bullet_Q}(\tau_{k_1}\otimes\ldots\otimes \tau_{k_\ell})(f_1\otimes\ldots\otimes f_\ell;\varphi)
		:=\sum_{\substack{N\geq 0\\N_1+\ldots+ N_\ell=2N}}
		\frac{1}{(2N)!!}\frac{(2N)!}{N_1!\cdots N_\ell!}
		\widehat{S}_{\ell,N}(f_1\otimes\ldots\otimes f_\ell\otimes 1_{2N};\varphi)\,.
	\end{align*}
	This formula entails that $\Gamma_{\bullet_Q}(\tau_{k_1}\otimes\ldots\otimes \tau_{k_\ell})\in\mathcal{D}'(M^\ell;\operatorname{Pol})$ which is per construction symmetric in $\tau_{k_1},\ldots,\tau_{k_\ell}$. Furthermore, for all $p\in\mathbb{N}\cup\{0\}$ it holds
	\begin{multline*}
		\Gamma_{\bullet_Q}(\tau_{k_1}\otimes\ldots\otimes \tau_{k_\ell})^{(p)}(f_1\otimes\ldots\otimes f_\ell\otimes\psi^{\otimes p};\varphi)
		\\=\sum\limits_{N=0}^\infty\frac{1}{(2N)!!}\sum_{\substack{N_1,\dots N_\ell\\N_1+\ldots+ N_\ell=2N}}
		\frac{(2N)!}{N_1!\cdots N_\ell!}
		\widehat{S}_{\ell,N}^{(p)}(f_1\otimes\ldots\otimes f_\ell\otimes 1_{2N}\otimes\psi^{\otimes p};\varphi)\,.
	\end{multline*}
	With reference to Equation \eqref{Eq: GammabulletQ functional derivative requirement} consider the formal expression
	\begin{multline*}
		(\tau_{k_1}\bullet_Q\ldots\bullet_Q \tau_{k_\ell})^{(p)}(f_1\otimes\ldots\otimes f_\ell\otimes\psi^{\otimes p};\varphi)=
		\\\sum_{\substack{N\geq 0\\N_1+\ldots+ N_\ell=2N\\p_1+\ldots+p_\ell=p}}
		\frac{1}{(2N)!!}\frac{(2N)!p!}{\prod_{i=1}^{\ell}N_i!p_i!}
		\big[
		(1_\ell\otimes Q^{\otimes N}\otimes 1_p)\cdot
		(\tau_{k_1}^{(N_1+p_1)}\widetilde{\otimes}\ldots\widetilde{\otimes}\tau_{k_\ell}^{(N_\ell+p_\ell)})
		\big](f_1\otimes\ldots\otimes f_\ell\otimes 1_{2N}\otimes\psi^{\otimes p};\varphi)
		\\=\sum_{\substack{N\geq 0\\N_1+\ldots+ N_\ell=2N\\p_1+\ldots+p_\ell=p}}
		\frac{1}{(2N)!!}\frac{(2N)!p!}{\prod_{i=1}^{\ell}N_i!p_i!}
		S_{\ell,N}^{[\widehat{p}_\ell]}(f_1\otimes\ldots\otimes f_\ell\otimes 1_{2N}\otimes\psi^{\otimes p};\varphi)\,,
	\end{multline*}
	where $S_{\ell,N}^{[\widehat{p}_\ell]}$ is a functional-valued distribution defined on $M^{\ell+2N+p}\setminus\mathrm{Diag}_{\ell+2N+p}$ with finite scaling degree at $\mathrm{Diag}_{\ell+2N+p}$. Here we set $\widehat{p}_\ell=(p_1,\ldots,p_\ell)$.
	On account of the explicit form of such $\widehat{S}_{\ell,M}$ -- \textit{cf.} theorem \ref{Thm: extension with scaling degree} -- we can choose $\widehat{S}_{\ell,M}$ so that
	\begin{align*}
		\widehat{S}_{\ell,M}^{(p)}
		=\sum_{p_1+\ldots+p_\ell=p}\frac{p!}{p_1!\cdots p_\ell!}\widehat{S_{\ell,N}^{[\widehat{p}_\ell]}}\,,
	\end{align*}
	where $\widehat{S_{\ell,N}^{[\widehat{p}_\ell]}}$ denotes a scaling degree preserving extension of $S_{\ell,N}^{[\widehat{p}_\ell]}$ on $M^{\ell+2N+p}$.
	With this choice it follows that $\Gamma_{\bullet_Q}(\tau_{k_1}\otimes\ldots\otimes\tau_{k_\ell})$ satisfies Equation \eqref{Eq: GammabulletQ functional derivative requirement}.
	In addition,  $\Gamma_{\bullet_Q}(\tau_{k_1}\otimes\ldots\otimes\tau_{k_\ell})\in\mathcal{D}_{\mathrm{C}}^\prime(M^\ell;\operatorname{Pol})$ as it can be shown by evaluating the wave front set of
	\begin{align*}
		\mathcal{D}(M)^{\otimes\ell}\otimes\mathcal{E}(M)^{\otimes p}\ni f_1\otimes\ldots\otimes f_\ell\otimes\psi^{\otimes p}
		\mapsto \widehat{S}_{\ell,N}^{(p)}(f_1\otimes\ldots\otimes f_\ell\otimes 1_{2N}\otimes\psi^{\otimes p};\varphi)\,.
	\end{align*}
	This concludes the proof by induction over $j_1,\ldots,j_\ell$ and consequently also that over $k$ and over $\ell$.
	
	\paragraph{Final Step -- Algebraic properties of $\bullet_{\Gamma_{\bullet_Q}}$.}
	Once $\Gamma_{\bullet_Q}$ has been defined, the algebraic properties of $\bullet_{\Gamma_{\bullet_Q}}$ can be proved by direct inspection. In particular commutativity of the product is inherited by the symmetry of $\Gamma_{\bullet_Q}$.
	Associativity follows instead by a direct computation. For all $\tau_1,\tau_2,\tau_3\in\mathcal{A}_{\bullet_Q}$ it holds
	\begin{align*}
		(\tau_1\bullet_{\Gamma_{\bullet_Q}}\tau_2)\bullet_{\Gamma_{\bullet_Q}}\tau_3
		=\Gamma_{\bullet_Q}(\Gamma_{\bullet_Q}^{-1}(\tau_1\bullet_{\Gamma_{\bullet_Q}}\tau_2)\otimes \Gamma_{\bullet_Q}^{-1}(\tau_3))
		&=\Gamma_{\bullet_Q}(\Gamma_{\bullet_Q}^{-1}(\tau_1)\otimes \Gamma_{\bullet_Q}^{-1}(\tau_2)\otimes \Gamma_{\bullet_Q}^{-1}(\tau_3))
		\\&=\tau_1\bullet_{\Gamma_{\bullet_Q}} (\tau_2\bullet_{\Gamma_{\bullet_Q}}\tau_3)\,.
	\end{align*}
\end{proof}

\begin{remark}[Parabolic Case]
As for Theorem \ref{Thm: Gamma cdotQ existence}, also Theorem \ref{Thm: GammabulletQ existence} holds true \emph{mutatis mutandis} also in the parabolic case $M=\mathbb{R}\times\Sigma$ and $E=\partial_t-\widetilde{E}$. As we discussed in Remark \ref{Remark: parabolic case of AcdotQ}, this is a consequence of the finiteness of the weighted scaling degree, {\it cf.} Remark \ref{Rmk: weighted scaling degree}.
\end{remark}

\section{Uniqueness theorems}\label{Sec: Uniqueness theorems}

In the previous two sections we have proven existence of the algebras $\mathcal{A}_{\cdot_Q}$ and $\mathcal{A}_{\bullet_Q}$, which, in turn depend on the linear maps $\Gamma_{\cdot_Q}\colon\mathcal{A}\to\mathcal{D}_{\mathrm{C}}'(M;\operatorname{Pol})$ and $\Gamma_{\bullet_Q}\colon\mathcal{A}_{\cdot_Q}\to\mathcal{T}_{\mathrm{C}}'(M;\operatorname{Pol})$, {\it cf.} Theorems \ref{Thm: Gamma cdotQ existence} and \ref{Thm: GammabulletQ existence}. A close inspection of the proofs, {\it e.g.} Equation \eqref{Eq: definition of Phi2}, unveils that, in general, one cannot hope for uniqueness of the above structures. This prompts the questions, whether it is possible to classify the freedom in the definition of $\Gamma_{\cdot_Q}$ and of $\Gamma_{\bullet_Q}$ and how the algebras $\mathcal{A}_{\cdot_Q}$ and $\mathcal{A}_{\bullet_Q}$ depend on them. In this section, we address these issues. We observe that our approach is very similar in spirit to the same one followed by Hollands and Wald in discussing renormalization freedoms in relativistic quantum field theory on globally hyperbolic spacetimes \cite{Hollands-Wald-01, Hollands-Wald-02}.

\begin{remark}\label{Rmk: uniqueness for Phi2}
	In order to better understand the rationale behind the characterization of the freedom in defining $\Gamma_{\cdot_Q}$ and $\Gamma_{\bullet_Q}$, we feel worth focusing first on a concrete example, which is traded from Equation \eqref{Eq: definition of Phi2}. The notation used is consistent with that of Theorem \ref{Thm: GammacdotQ uniqueness}.
	
	Suppose, therefore, that $\Gamma_{\cdot_Q}$ and $\widetilde{\Gamma}_{\cdot_Q}$ are two different maps, compatible with the constraints listed in Theorem \ref{Thm: Gamma cdotQ existence}. Consider in addition the functional $\tau=\Phi^2$, {\it cf.} Example \ref{Ex: examples of functional-valued polynomial distributions}, and let $\mathcal{C}_0\in\mathcal{D}'_{\mathrm{C}}(M;\operatorname{Pol})$ be 
	\begin{align*}
		\mathcal{C}_0(f;\varphi)
		:=\Gamma_{\cdot_Q}(\Phi^2)(f;\varphi)
		-\widetilde{\Gamma}_{\cdot_Q}(\Phi^2)(f;\varphi)\qquad\forall (f;\varphi)\in\mathcal{D}(M)\times\mathcal{E}(M)\,.
	\end{align*}
	By direct inspection of Equation \eqref{Eq: definition of Phi2}, it turns out that, for all $\psi\in\mathcal{E}(M)$,
	\begin{align*}
		\mathcal{C}_0^{(1)}(f\otimes\psi;\varphi)
		&=\Gamma_{\cdot_Q}(\Phi^2)^{(1)}(f\otimes\psi;\varphi)
		-\widetilde{\Gamma}_{\cdot_Q}(\Phi^2)^{(1)}(f\otimes\psi;\varphi)
		\\&=2\Gamma_{\cdot_Q}(\psi\Phi)(f;\varphi)
		-2\widetilde{\Gamma}_{\cdot_Q}(\psi\Phi)(f;\varphi)
		=0\,,
	\end{align*}
	where we used Equation \eqref{Eq: Pi-functional derivative} together with Equation \eqref{Eq: Gamma on M1}. It follows that $\mathcal{C}_0$ does not depend on  $\varphi\in\mathcal{E}(M)$.
	Combining this statement with Equation \eqref{Eq: C-set}, it turns out that $\mathcal{C}_0\in\mathcal{D}^\prime(M)$ and $\operatorname{WF}(\mathcal{C}_0)=\emptyset$.
	Hence there must exist $c_0\in \mathcal{E}(M)$ such that $\mathcal{C}_0=c_0\boldsymbol{1}$.
	It follows that
	\begin{align*}
		\Gamma_{\cdot_Q}(\Phi^2)
		=\widetilde{\Gamma}_{\cdot_Q}(\Phi^2)
		+c_0\boldsymbol{1}\,,
	\end{align*}
	that is, on $\Phi^2$, $\widetilde{\Gamma}_{\cdot_Q}$ and $\Gamma_{\cdot_Q}$, differs by a distribution generated by a smooth function.
\end{remark}

\noindent We start by focusing the attention on $\Gamma_{\cdot_Q}$, generalizing the result of Remark \ref{Rmk: uniqueness for Phi2}.

\begin{theorem}\label{Thm: GammacdotQ uniqueness}
	Let $\widetilde{\Gamma}_{\cdot_Q},\Gamma_{\cdot_Q}:\mathcal{A}\to\mathcal{D}^\prime(M;\operatorname{Pol})$ be two linear maps compatible with the constraints listed in Theorem \ref{Thm: Gamma cdotQ existence}.
	Then there exists a family $\{\mathcal{C}_k\}_{k\in\mathbb{N}_0}$ of linear maps $\mathcal{C}_k\colon\mathcal{A}\to\mathcal{M}_k$,  such that:	
%	Recalling that $\mathcal{A}=\bigoplus_{k\in\mathbb{N}_0}\mathcal{M}_k$, {\it cf.} Remark \ref{Rem: graded algebra}, there exists a family $\{\mathcal{C}_{k,\ell}\}_{k,\ell\in\mathbb{N}_0}$ of linear maps $\mathcal{C}_{k,\ell}\colon\mathcal{M}_k\to\mathcal{M}_\ell$,  such that:	
	\begin{subequations}
		\begin{align}
		\label{Eq: tauell consistency}
		\mathcal{C}_\ell[\mathcal{M}_k]
		&=0\,,
		\qquad\forall k\leq\ell+1\,,\\
		\label{Eq: tauell on Ptau}
		\mathcal{C}_k[P\circledast\tau]
		&=P\circledast\mathcal{C}_k[\tau]\,,
		\qquad\forall k\in\mathbb{N}_0\,,\forall\tau\in\mathcal{A}\\
		\label{Eq: tauell functional derivative}
		\delta_\psi\circ\mathcal{C}_k
		&=\mathcal{C}_{k-1}\circ\delta_\psi\,,
		\qquad\forall k\in\mathbb{N}\,,\forall\psi\in\mathcal{E}(M)\\
		\label{Eq: GammacdotsQ uniqueness}
		\widetilde{\Gamma}_{\cdot_Q}(\tau)
		&=\Gamma_{\cdot_Q}\left(\tau
		+\mathcal{C}_{k-2}[\tau]\right)\,,
		\qquad\forall\tau\in\mathcal{M}_k\,.
		\end{align}
	\end{subequations}
\end{theorem}
\begin{proof}
	The proof goes by induction on $k$, namely we show that, whenever the action of $\mathcal{C}_\ell$ has been consistently defined on all $\widetilde{\tau}\in\mathcal{M}_{k-1}$ for all $\ell\in\mathbb{N}_0$, then $\mathcal{C}_\ell[\tau]$ is known also for $\tau\in\mathcal{M}_k$.
	
	If $k=1$, then both Equation \eqref{Eq: Gamma on M1} and Equation \eqref{Eq: tauell consistency} entail that $\mathcal{C}_\ell[\mathcal{M}_1]=0$ for all $\ell\in\mathbb{N}_0$.
	By the induction hypothesis we assume that, for all $\ell\in\mathbb{N}_0$, $\mathcal{C}_\ell[\tau]$ has been defined for all $\tau\in\mathcal{M}_{k-1}$.
%	We shall then define $\tau_\ell[\tau]$ for all $\tau\in\mathcal{M}_k$ equations \eqref{Eq: tauell on Ptau}, \eqref{Eq: tauell functional derivative} and \eqref{Eq: GammacdotsQ uniqueness} are still satisfied.
	
	In order to define $\mathcal{C}_\ell[\tau]$ for all $\ell\in\mathbb{N}\cup\{0\}$ and $\tau\in\mathcal{M}_k$, we observe that Equation \eqref{Eq: tauell consistency} entails that $\mathcal{C}_\ell[\mathcal{M}_k]=0$ whenever $\ell\geq k-1$.
	To construct $\mathcal{C}_\ell$, for $0\leq\ell\leq k-2$, let $\tau\in\mathcal{M}_k$.
	If there exists $\bar{\tau}\in\mathcal{M}_k$ such that $\tau=P\circledast\bar{\tau}$, in agreement with Equation \eqref{Eq: tauell on Ptau}, we set  $\mathcal{C}_\ell[P\circledast\bar{\tau}]=P\circledast\mathcal{C}_\ell[\bar{\tau}]$.
	If we do not fall in this case, for all $\ell\in\{1,\ldots,k-2\}$, we define -- \textit{cf.} Remark \ref{Rmk: Pol isomorphic space} -- 
	\begin{align*}
		\mathcal{C}_\ell[\tau](f;\varphi)
		=\mathcal{C}_\ell[\tau](f;0)
		+\int_0^1\mathcal{C}_{\ell-1}[\delta_\varphi\tau](f;s\varphi)\mathrm{d}s\,.
		\qquad\forall\, f\in\mathcal{D}(M)\,,
		\forall\,\varphi\in \mathcal{E}(M)\,.
	\end{align*}
	Notice that $\delta_\varphi\tau\in\mathcal{M}_{k-1}$ so that $\mathcal{C}_{\ell-1}[\delta_\varphi\tau]\in\mathcal{M}_{\ell-1}$ is known by the inductive hypothesis, whereas $\mathcal{C}_\ell[\tau](f;0)$ can be seen as an arbitrary additive constant.
	Moreover $\mathcal{C}_\ell[\tau]$ satisfies Equation \eqref{Eq: tauell functional derivative}, since for all $\psi\in \mathcal{E}(M)$, 
	\begin{align*}
		(\delta_\psi\mathcal{C}_\ell[\tau])(\varphi)
		&=\frac{\mathrm{d}}{\mathrm{d}\lambda}\mathcal{C}_\ell[\tau](\varphi+\lambda\psi)\bigg|_{\lambda=0}
		\\&=\frac{\mathrm{d}}{\mathrm{d}\lambda}\int_0^1\mathcal{C}_{\ell-1}[\delta_{\varphi+\lambda\psi}\tau](s\varphi+s\lambda\psi)\mathrm{d}s\bigg|_{\lambda=0}
		\\&=\int_0^1\mathcal{C}_{\ell-1}[\delta_\psi\tau](s\varphi)\mathrm{d}s
		+\int_0^1s\delta_\psi\mathcal{C}_{\ell-1}[\delta_\varphi\tau](s\varphi)\mathrm{d}s\,.
	\end{align*}
	We observe that, on account of Equation \eqref{Eq: tauell functional derivative},
	\begin{align*}
		\delta_\psi\mathcal{C}_{\ell-1}[\delta_\varphi\tau](s\varphi)
		=\delta_\varphi\mathcal{C}_{\ell-1}[\delta_\psi\tau](s\varphi)
		=\frac{\mathrm{d}}{\mathrm{d}s}\mathcal{C}_{\ell-1}[\delta_\psi\tau](s\varphi)\,.
	\end{align*}
	Integration by parts yields
	\begin{align*}
		(\delta_\psi\mathcal{C}_\ell[\tau])(\varphi)
		&=\int_0^1\mathcal{C}_{\ell-1}[\delta_\psi\tau](s\varphi)\mathrm{d}s
		+\int_0^1s\frac{\mathrm{d}}{\mathrm{d}s}\mathcal{C}_{\ell-1}[\delta_\psi\tau](s\varphi)\mathrm{d}s
		\\&=s\mathcal{C}_{\ell-1}[\delta_\psi\tau](s\varphi)\bigg|^{s=1}_{s=0}
		=\mathcal{C}_{\ell-1}[\delta_\psi\tau](\varphi)\,.
	\end{align*}
	We have defined $\mathcal{C}_\ell[\tau]$ for all $\tau\in\mathcal{M}_k$ and $\ell\geq 1$ compatibly with the constraints in Equations \eqref{Eq: tauell consistency}-\eqref{Eq: tauell on Ptau}.
	We are left with working with $\mathcal{C}_0[\tau]$ showing in addition that Equation \eqref{Eq: GammacdotsQ uniqueness} holds true.
	For $\tau\in\mathcal{M}_k$ we set
	\begin{align*}
		\mathcal{C}_0[\tau]
		:=\widetilde{\Gamma}_{\cdot_Q}(\tau)
		-\Gamma_{\cdot_Q}(\tau)
		-\Gamma_{\cdot_Q}(\mathcal{C}_{k-2}[\tau])\,.
	\end{align*}
	The assignment $\tau\to\mathcal{C}_0[\tau]$ is linear and, moreover, $\mathcal{C}_0[\tau]\in\mathcal{M}_0$.
	This follows by direct inspection of $\delta_\psi\mathcal{C}_0[\tau]$ since
	\begin{align*}
		\delta_\psi\mathcal{C}_0[\tau]
		=\widetilde{\Gamma}_{\cdot_Q}(\delta_\psi\tau)
		-\Gamma_{\cdot_Q}(\delta_\psi\tau)
		-\Gamma_{\cdot_Q}(\mathcal{C}_{k-3}[\delta_\psi\tau])
		=0\,,
	\end{align*}
	where we used Equations \eqref{Eq: GammabulletQ functional derivative requirement} and \eqref{Eq: tauell functional derivative} together with the inductive hypothesis for $\delta_\psi\tau\in\mathcal{M}_{k-1}$.
	Equation \eqref{Eq: GammacdotsQ uniqueness} is trivially satisfied since $\Gamma_{\cdot_Q}(\mathcal{C}_0[\tau])=\mathcal{C}_0[\tau]$.
	Hence the inductive proof is complete.
\end{proof}

%\noindent Theorem \ref{Thm: GammacdotQ uniqueness} and Equation \eqref{Eq: GammacdotsQ uniqueness} in particular imply the following result.
%The proof is a direct application of the construction above and of Equation \eqref{Eq: cdot GammaQ product}, hence we omit it.

\begin{corollary}\label{Cor: AcdotQ does not depends on Gamma cdotQ}
	Let $\widetilde{\Gamma}_{\cdot_Q}, \Gamma_{\cdot_Q}:\mathcal{A}\to\mathcal{D}^\prime(M;\operatorname{Pol})$ be two linear maps compatible with the constraints listed in Theorem \ref{Thm: Gamma cdotQ existence}.
	Then the algebras $\mathcal{A}_{\cdot_Q}=\Gamma_{\cdot_Q}(\mathcal{A})$ and  $\widetilde{\mathcal{A}}_{\cdot_Q}=\widetilde{\Gamma}_{\cdot_Q}(\mathcal{A})$, defined as per Corollary \ref{Cor: construction of AcdotQ}, coincide.
\end{corollary}

Roughly speaking Theorem \ref{Thm: GammacdotQ uniqueness} ensures that, for $\tau\in\mathcal{M}_k$, $\widetilde{\Gamma}_{\cdot_Q}(\tau)$ and $\Gamma_{\cdot_Q}(\tau)$ differs by $\Gamma_{\cdot_Q}(\mathcal{C}_{k-2}[\tau])$, being $\mathcal{C}_{k-2}[\tau]\in\mathcal{M}_{k-2}$.
Though this space is quite large, the ambiguity in the actual form of $\mathcal{C}_{k-2}[\tau]$ can be further restricted via Equation \eqref{Eq: tauell functional derivative}.
This can be seen explicitly by considering specific algebra elements, namely $\Phi^k$ as per Example \ref{Ex: examples of functional-valued polynomial distributions}.
%Another notable consequence of Theorem \ref{Thm: GammacdotQ uniqueness} can be obtained consider specific algebra elements, namely $\Phi^k$ as per Example \ref{Ex: examples of functional-valued polynomial distributions}.

\begin{proposition}\label{Prop: Phik uniqueness}
	Let $\widetilde{\Gamma}_{\cdot_Q}, \Gamma_{\cdot_Q}:\mathcal{A}\to\mathcal{D}^\prime(M;\operatorname{Pol})$ be two linear maps compatible with the constraints listed in Theorem \ref{Thm: Gamma cdotQ existence}.
	Then there exists $\{c_\ell\}_{\ell\in\mathbb{N}_0}\subset \mathcal{E}(M)$ a family of smooth functions, such that for all $k\in\mathbb{N}$
	\begin{align}\label{Eq: Phik uniqueness}
		\widetilde{\Gamma}_{\cdot_Q}(\Phi^k)
		=\Gamma_{\cdot_Q}\bigg(\Phi^k
		+\sum_{\ell=0}^{k-2}{k \choose \ell}c_{k-\ell}\Phi^\ell\bigg)\,.
	\end{align}
\end{proposition}
\begin{proof}
 For $k=1$ Equation \eqref{Eq: Phik uniqueness} is a consequence of Equation \eqref{Eq: Gamma on M1}, while for $k=2$ it reduces to Remark \ref{Rmk: uniqueness for Phi2}. The proof can be concluded proceeding by induction over $k$ and using Equation \eqref{Eq: GammacdotsQ uniqueness}.
\end{proof}

\begin{remark}\label{Rem: Generalization of smooth arbitrariness}
	It is natural to wonder whether one can generalize Equation \eqref{Eq: Phik uniqueness} by considering for $\tau\in\mathcal{M}_k$
	\begin{align}\label{Eq: false GammacdotsQ uniqueness}
		\widetilde{\Gamma}_{\cdot_Q}(\tau)(f;\varphi)
		=\Gamma_{\cdot_Q}(\tau)(f;\varphi)
		+\sum_{\ell=0}^{k-2}\frac{1}{(k-\ell)!}\Gamma_{\cdot_Q}(\tau)^{(k-\ell)}(f\otimes \tilde{c}_{k-\ell};\varphi)\,,
	\end{align}
	where $\tilde{c}_{k-\ell}\in \mathcal{E}(M^{k-\ell})$.
	If $\tau=\Phi^k$ the latter formula reduces to equation \eqref{Eq: Phik uniqueness} by setting
	\begin{align*}
		\tilde{c}_{k-\ell}(\widehat{x}_{k-\ell})
		:=\frac{1}{k-\ell}\sum_{j=1}^{k-\ell} c_{k-\ell}(x_j)\,,
	\end{align*}
	and by observing that $(\Phi^k)^{(k-\ell)}=\frac{k!}{\ell!}\Phi^{\ell}\delta_{\mathrm{Diag}_{\ell+1}}$. Here $\widehat{x}_{k-\ell}=(x_1,\dots,x_{k-\ell})$.
	
	Yet, in general Equation \eqref{Eq: false GammacdotsQ uniqueness} does not hold true, as one can infer considering for example $\tau=\Phi P\circledast\Phi$. On account of Equations \eqref{Eq: GammabulletQ with one factor in M1} and \eqref{Eq: GammabulletQ functional derivative requirement}, it holds that there must exist $c\in\mathcal{E}(M)$ such that
	\begin{align*}
		\widetilde{\Gamma}_{\cdot_Q}(\tau)(f;\varphi)
		-\Gamma_{\cdot_Q}(\tau)(f;\varphi)
		=\int\limits_M c\,f\mu.
	\end{align*}
	At the same time Equation \eqref{Eq: false GammacdotsQ uniqueness} entails that there must exist $\tilde{c}_2\in \mathcal{E}(M^2)$ such that
	\begin{align*}
		\int_M c\,f\,\mu
		=\Gamma_{\cdot_Q}(\tau)^{(2)}(f\otimes \tilde{c}_2;\varphi)
		=\int\limits_{M\times M} f(x)P(x,y)[\tilde{c}_2(x,y)+\tilde{c}_2(y,x)]\mathrm{d}\mu(x)\mathrm{d}\mu(y)\,.
	\end{align*}
	In general there is no such $\tilde{c}_2$.
\end{remark}

To conclude this section we state a counterpart of Theorem \ref{Thm: GammacdotQ uniqueness} aimed at characterizing the non uniqueness in constructing $\Gamma_{\bullet_Q}$. 

\begin{theorem}\label{Thm: GammabulletQ uniqueness}
	Let $\widetilde{\Gamma}_{\bullet_Q}, \Gamma_{\bullet_Q}:\mathcal{A}_{\cdot_Q}\to\mathcal{T}^\prime_{\mathrm{C}}(M;\operatorname{Pol})$ be two linear maps compatible with the constraints listed in Theorem \ref{Thm: GammabulletQ existence}.
	Then there exists a family $\{\mathcal{C}_{\underline{k}}\}_{\underline{k}\in(\mathbb{N}_0)^{\mathbb{N}_0}}$ of linear maps $\mathcal{C}_{\underline{k}}\colon\mathcal{T}(\mathcal{A})\to\mathcal{T}(\mathcal{A})$ such that:
	\begin{enumerate}
		\item
		for all $\ell\in\mathbb{N}\cup\{0\}$, 
		$\mathcal{C}_{\underline{k}}[\mathcal{A}^{\otimes\ell}]\subseteq\mathcal{M}_{k_1}\otimes\ldots\otimes\mathcal{M}_{k_\ell}$ while
		$\mathcal{C}_{\underline{j}}[\mathcal{M}_{k_1}\otimes\ldots\otimes\mathcal{M}_{k_\ell}]=0$ whenever $k_i\leq j_i-1$ for some $i\in\{1,\ldots,\ell\}$
		\item
		for all $\ell\in\mathbb{N}\cup\{0\}$ and $\tau_1,\ldots,\tau_\ell\in\mathcal{A}$, it holds
		\begin{align}
			\label{Eq: tauj on Ptau}
			\mathcal{C}_{\underline{k}}(\tau_1\otimes\ldots\otimes P\circledast\tau_k\otimes\ldots\otimes\tau_\ell)
			&=(\delta_{\mathrm{Diag}_2}^{\otimes (k-1)}\otimes P\otimes\delta_{\mathrm{Diag}_2}^{\otimes \ell-k})
			\circledast\mathcal{C}_{\underline{k}}(\tau_1\otimes\ldots\otimes\tau_\ell)\\
			\label{Eq: tauj functional derivative}
			\delta_\psi\mathcal{C}_{\underline{k}}(\tau_1\otimes\ldots\otimes\tau_\ell)
			&=\sum_{a=1}^\ell\mathcal{C}_{\underline{k}(a)}(\tau_1\otimes\ldots\otimes\delta_\psi\tau_a\otimes\ldots\otimes\tau_\ell)\,,
		\end{align}
		where $\underline{k}(a)_i=k_i$ if $i\neq a$ and $\underline{k}(a)_a=k_a-1$.
		\item
		for all $\tau_{k_1},\ldots,\tau_{k_\ell}\in\mathcal{A}$, with $\tau_{k_j}\in\mathcal{M}_{k_j}$ for all $j\in\{1,\ldots,\ell\}$, and $f_1,\ldots,f_\ell\in\mathcal{D}(M)$ it holds
		\begin{multline}\label{Eq: GammabulletQ uniqueness}
			\widetilde{\Gamma}_{\bullet_Q}(\widetilde{\Gamma}_{\cdot_Q}^{\otimes\ell}(\tau_{k_1}\otimes\ldots\otimes\tau_{k_\ell}))(f_1\otimes\ldots\otimes f_\ell)
			=\Gamma_{\bullet_Q}\bigg[
			\Gamma_{\cdot_Q}^{\otimes\ell}(\tau_{k_1}\otimes\ldots\otimes\tau_{k_\ell})
			\bigg](f_1\otimes\ldots\otimes f_\ell)
			\\+\sum_{\wp\in\mathrm{P}(1,\ldots,\ell)}
			\Gamma_{\bullet_Q}\bigg[\Gamma_{\cdot_Q}^{\otimes|\wp|}\mathcal{C}_{\underline{k}_\wp}
			\bigg(\bigotimes_{I\in\wp}\prod_{i\in I}\tau_{k_i}\bigg)
			\bigg]\bigg(\bigotimes_{I\in\wp}\prod_{i\in I}f_i\bigg)\,.
		\end{multline}
		Here $\mathrm{P}(1,\ldots,\ell)$ denotes the set of partitions of $\{1,\ldots,\ell\}$ into non-empty disjoint sets while $\underline{k}_\wp=(k_I)_{I\in\wp}$ where $k_I:=\sum_{i\in I}k_i$ whereas $\underline{j}\leq\underline{k}_\wp$ means that $j_I\leq k_I$ for all $I\in\wp$.
%		Finally $\prod^{\cdot_Q}$ denotes the product in $\mathcal{A}_{\cdot_Q}$.
	\end{enumerate}
\end{theorem}
\begin{proof}
	The proof follows the same steps as those in the proof of Theorem \ref{Thm: GammabulletQ uniqueness}.
	Notice in particular that for $\ell=1$ Equation \eqref{Eq: GammabulletQ uniqueness} reduces to Equation \eqref{Eq: GammacdotsQ uniqueness}. At this stage we can proceed by induction exactly as in Theorem \ref{Thm: GammacdotQ uniqueness} just taking into account more indices. For this reason we omit giving all details.
	% implying that we can define $\tau_{\underline{j}}(\mathcal{A}):=\tau_{j_1}(\mathcal{A})$, being $\{\tau_j\}_{j\in\mathbb{N}\cup\{0\}}$ the family of maps introduced in Theorem \ref{Thm: GammacdotQ uniqueness}.
\end{proof}

\begin{example}
	For concreteness we now specialize equation \eqref{Eq: GammabulletQ uniqueness} for the case of two elements $\tau_{k_1}=\tau_{k_2}=\Phi^2$.
	In this latter case the admissible partitions are $\wp=\{1,2\}$ and $\wp=\{\{1\},\{2\}\}$ which lead to
	\begin{align*}
		\widetilde{\Gamma}_{\bullet_Q}(\widetilde{\Gamma}_{\cdot_Q}(\Phi^2)\otimes\widetilde{\Gamma}_{\cdot_Q}(\Phi^2))(f_1\otimes f_2)
		&=\Gamma_{\bullet_Q}(\Gamma_{\cdot_Q}(\Phi^2)\otimes\Gamma_{\cdot_Q}(\Phi^2))(f_1\otimes f_2)
		\\&+\Gamma_{\bullet_Q}(\Gamma_{\cdot_Q}^{\otimes 2}(\mathcal{C}_{2,2}(\Phi^2\otimes\Phi^2)))(f_1\otimes f_2)
		+\Gamma_{\bullet_Q}(\Gamma_{\cdot_Q}(\mathcal{C}_{4}(\Phi^4)))(f_1f_2)
		\\&=\Gamma_{\bullet_Q}(\Gamma_{\cdot_Q}(\Phi^2)\otimes\Gamma_{\cdot_Q}(\Phi^2))(f_1\otimes f_2)
		\\&+\Gamma_{\bullet_Q}(\Gamma_{\cdot_Q}^{\otimes 2}(\mathcal{C}_{2,2}(\Phi^2\otimes\Phi^2)))(f_1\otimes f_2)
		+\Gamma_{\cdot_Q}(\mathcal{C}_4(\Phi^4))(f_1f_2)\,,
	\end{align*}
	where we used Equation \eqref{Eq: GammabulletQ with one factor in M1}.
	In the previous identity $\mathcal{C}_4(\Phi^4)\in\mathcal{M}_4$ while $\mathcal{C}_{2,2}(\Phi^2\otimes\Phi^2)\in\mathcal{M}_2\otimes\mathcal{M}_2$.
	
	We shall now compute $\mathcal{C}_4(\Phi^4)$ and $\mathcal{C}_{2,2}(\Phi^2\otimes\Phi^2)$ explicitly.
	On account of Proposition \eqref{Prop: Phik uniqueness} we know that $\widetilde{\Gamma}_{\cdot_Q}(\Phi^2)=\Gamma_{\cdot_Q}(\Phi^2)+c_0\boldsymbol{1}$, where $c_0\in \mathcal{E}(M)$.
	This entails already that
	\begin{align*}
		\widetilde{\Gamma}_{\bullet_Q}(\widetilde{\Gamma}_{\cdot_Q}(\Phi^2)\otimes\widetilde{\Gamma}_{\cdot_Q}(\Phi^2))
		&=\widetilde{\Gamma}_{\bullet_Q}((\Gamma_{\cdot_Q}(\Phi^2)+c_0\boldsymbol{1})\otimes(\Gamma_{\cdot_Q}(\Phi^2)+c_0\boldsymbol{1}))
		\\&=\widetilde{\Gamma}_{\bullet_Q}(\Gamma_{\cdot_Q}(\Phi^2)\otimes\Gamma_{\cdot_Q}(\Phi^2))
		+2\Gamma_{\cdot_Q}(\Phi^2)\vee c_0\boldsymbol{1}
		+c_0\boldsymbol{1}\vee c_0\boldsymbol{1}\,,
	\end{align*}
	where $\vee$ denotes the symmetrized tensor product.
	It remains to be evaluated
	\begin{align*}
		R:=\widetilde{\Gamma}_{\bullet_Q}(\Gamma_{\cdot_Q}(\Phi^2)\otimes\Gamma_{\cdot_Q}(\Phi^2))
		-\Gamma_{\bullet_Q}(\Gamma_{\cdot_Q}(\Phi^2)\otimes\Gamma_{\cdot_Q}(\Phi^2))\,.
	\end{align*}
	A direct inspection using Equation \eqref{Eq: GammabulletQ functional derivative requirement} shows that $\delta_\psi R=0$, moreover, Equation \eqref{Eq: smooth factorization} entails that $R(f_1\otimes f_2)=0$ whenever $f_1,f_2\in\mathcal{D}(M)$ are such that $\operatorname{supp}(f_1)\cap\operatorname{supp}(f_2)=\emptyset$.
	It descends that there exists $c_1\in \mathcal{E}(M)$ such that
	\begin{align*}
		\widetilde{\Gamma}_{\bullet_Q}(\widetilde{\Gamma}_{\cdot_Q}(\Phi^2)\otimes\widetilde{\Gamma}_{\cdot_Q}(\Phi^2))(f_1\otimes f_2)
		&=\Gamma_{\bullet_Q}(\Gamma_{\cdot_Q}(\Phi^2)\otimes\Gamma_{\cdot_Q}(\Phi^2))(f_1\otimes f_2)
		\\&+2\Gamma_{\cdot_Q}(\Phi^2)\vee c_0\boldsymbol{1}(f_1\otimes f_2)
		+c_0\boldsymbol{1}\vee c_0\boldsymbol{1}(f_1\otimes f_2)
		+c_1\boldsymbol{1}(f_1f_2)\,.
	\end{align*}
	It follows that
	\begin{align*}
		\mathcal{C}_{2,2}(\Phi^2\otimes\Phi^2)
		=2\Gamma_{\cdot_Q}(\Phi^2)\vee c_0\boldsymbol{1}
		+c_0\boldsymbol{1}\vee c_0\boldsymbol{1}\,,\qquad
		\mathcal{C}_4(\Phi^4)
		=c_1\boldsymbol{1}\,.
	\end{align*}
\end{example}

\noindent Similarly to the case of Theorem \ref{Thm: GammacdotQ uniqueness}, Theorem \ref{Thm: GammabulletQ uniqueness} has the following important corollary.
\begin{corollary}\label{Cor: AbulletQ does not depends on Gamma cdotQ and Gamma bulletQ}
	Let $\widetilde{\Gamma}_{\cdot_Q}, \Gamma_{\cdot_Q}:\mathcal{A}\to\mathcal{D}^\prime(M;\operatorname{Pol})$ be two linear maps compatible with the constraints listed in Theorem \ref{Thm: Gamma cdotQ existence}.
	Moreover, let $\widetilde{\Gamma}_{\bullet_Q}\colon\mathcal{T}(\widetilde{\mathcal{A}}_{\cdot_Q})\to\mathcal{T}_{\mathrm{C}}(M;\operatorname{Pol})$, $\Gamma_{\bullet_Q}\colon\mathcal{T}(\mathcal{A}_{\cdot_Q})\to\mathcal{T}_{\mathrm{C}}(M;\operatorname{Pol})$ be linear maps as per Theorem \ref{Thm: GammabulletQ existence} -- here $\mathcal{A}_{\cdot_Q}=\Gamma_{\cdot_Q}(\mathcal{A})=\widetilde{\mathcal{A}}_{\cdot_Q}=\widetilde{\Gamma}_{\cdot_Q}(\mathcal{A})$ because of Corollary \ref{Cor: AcdotQ does not depends on Gamma cdotQ}. Then the algebras $\mathcal{A}_{\bullet_Q}=\Gamma_{\bullet_Q}(\mathcal{A}_{\cdot_Q})$ and  $\widetilde{\mathcal{A}}_{\bullet_Q}=\widetilde{\Gamma}_{\bullet_Q}(\widetilde{\mathcal{A}}_{\cdot_Q})$, defined as per Equation \eqref{Eq: bulletQ algebra}, coincide.
\end{corollary}

\section{Application to the $\Phi^3_d$ Model}\label{Sec: Application to the Phi3d Model}

Our goal is to prove the effectiveness of the framework developed in the previous sections by applying it to a concrete example, the so-called stochastic $\Phi^3_d$ model on $M=\mathbb{R}\times\mathbb{R}^d$, which has been already thoroughly discussed in the literature by several authors, especially in the context of stochastic quantization, see {\it e.g.} \cite{ADG20, Da Prato, GH20, Hairer14, Hairer15, JM85, PW81}. Observe that, to make contact with the existing literature, here $d$ does not refer to the dimension of $M$ which is instead $d+1$. 

To be more specific our goal is to study at a perturbative level  the following stochastic PDE on $\mathbb{R}\times\mathbb{R}^d$, 
\begin{align}\label{Eq: SPDE}
	\partial_t \widehat{\psi}=\Delta\widehat{\psi}-\lambda\widehat{\psi}^3+\widehat{\xi},
\end{align}
where $\lambda\in\mathbb{R}$ is a dimensionless coupling constant, $\widehat{\xi}$ denotes the space-time white noise whereas $\Delta$ is the Laplace operator on $\mathbb{R}^d$.

Following Section \ref{Sec: introduction} and using the tools of Section \ref{Sec: Basic definitions}, we start by translating Equation \eqref{Eq: SPDE} in the language of functionals.
For that we need to cope with the fact that $P\circledast\eta$, with $\eta\in\mathcal{E}(\mathbb{R}\times\mathbb{R}^d)$, is a priori ill-defined. For this reason we need to introduce a so-called infrared cut-off, namely we choose $\chi\in\mathcal{D}(\mathbb{R})$. Subsequently we would like to consider 
\begin{align}\label{Eq: algebraic SPDE}
	\Psi=\Phi-\lambda P_\chi\circledast(\Psi^3)\,,
\end{align}
where $P_\chi:=P\cdot(1\otimes\chi)$ while $\Phi\in\mathcal{D}^\prime(\mathbb{R}\times\mathbb{R}^d;\operatorname{Pol})$ is the functional defined in Equation \eqref{Eq: Phi,1 functional}.
Notice that Equation \eqref{Eq: algebraic SPDE} entails that $\Psi(f;\varphi)=\Phi(f;\varphi)$ for all $f\in\mathcal{D}(\mathbb{R}\times\mathbb{R}^d)$ such that $\operatorname{supp}(f)\cap(\operatorname{supp}(\chi)+\mathbb{R}_+)=\emptyset$ -- here we denoted by $\operatorname{supp}(\chi)+\mathbb{R}_+:=\{(t+s,x)\,|\,(s,x)\in\operatorname{supp}(\chi)\,,\,t>0\}$.

Yet, Equation \eqref{Eq: algebraic SPDE} is purely formal and we cannot claim that $\Psi\in\mathcal{D}^\prime(\mathbb{R}\times\mathbb{R}^d;\operatorname{Pol})$. To bypass this hurdle, our strategy is to give a perturbative scheme to find $\Psi$, the unknown in Equation \eqref{Eq: algebraic SPDE}.
Hence, recalling Definition \ref{Def: pointwise algebra}, we read $\Psi\in\mathcal{A}\subset\mathcal{D}^\prime(\mathbb{R}\times\mathbb{R}^d;\operatorname{Pol})$ as a formal power series with respect to the parameter $\lambda$.
More precisely, indicating with $\mathcal{A}\llbracket\lambda\rrbracket$ the space of formal power series in $\lambda$ with coefficients in $\mathcal{A}$, we consider
\begin{align}\label{Eq: formal power series}
	\Psi\llbracket\lambda\rrbracket=\sum_{j\geq0}\lambda^jF_j\,,\qquad F_j\in\mathcal{A}\,.
\end{align}

The rationale is to use Equation \eqref{Eq: algebraic SPDE} to compute order by order the coefficients $F_j\in\mathcal{D}^\prime(\mathbb{R}\times\mathbb{R}^d;\operatorname{Pol})$.
As a matter of fact, inserting Equation \eqref{Eq: formal power series} in Equation \eqref{Eq: algebraic SPDE} one obtains via a direct calculation 
\begin{align}\label{Eq: order-by-order perturbative expansion}
	F_0&=\Phi,\quad
	F_1=-P_\chi\circledast\Phi^3,\quad
	F_2=3P_\chi\circledast(\Phi^2P_\chi\circledast\Phi^3)\,,\\
	\nonumber
	F_j&=-\sum_{j_1+j_2+j_3=j-1}P_\chi\circledast(F_{j_1}F_{j_2}F_{j_3})\,,\qquad j\in\mathbb{N}\,.
\end{align}
For our expository purposes we will content ourselves with considering the coefficients only up to $j=2$.

So far, the whole construction seems oblivious to the stochastic nature of Equation \eqref{Eq: algebraic SPDE}.
This is indeed the case and, as discussed in Section \ref{Sec: construction of AcdotQ}, this shortcoming is cured by deforming the algebra $\mathcal{A}$, {\it cf.} Definition \ref{Def: pointwise algebra}, into the deformed counterpart $\mathcal{A}_{\cdot_Q}$. {\it cf.} Corollary \ref{Cor: construction of AcdotQ}.
This is implemented by the map $\Gamma_{\cdot_Q}$ constructed in Theorem \ref{Thm: Gamma cdotQ existence} and by Equation \eqref{Eq: cdot GammaQ product} in particular.
Notice that, in order to be consistent with the introduced cut-off $\chi$, we are now considering $Q=Q_\chi=P_\chi\circ P_\chi$.
In other words we consider
\begin{align*}
	\Psi\llbracket\lambda\rrbracket\mapsto\Psi_{\cdot_Q}[\![\lambda]\!]:=\Gamma_{\cdot_Q}(\Psi\llbracket\lambda\rrbracket)\,.
\end{align*}

\noindent As outlined in Section \ref{Sec: construction of AcdotQ}, $\Psi_{\cdot_Q}[\![\lambda]\!]\in\mathcal{A}_{\cdot_Q}\llbracket\lambda\rrbracket$ and, for all $\varphi\in\mathcal{E}(\mathbb{R}\times\mathbb{R}^d)$ and $f\in\mathcal{D}(\mathbb{R}\times\mathbb{R}^d)$, $\Psi_{\cdot_Q}[\![\lambda]\!](f;\varphi)$ is the expectation value of the $\varphi$-shifted, $f$-localized perturbative solution $\widehat{\psi}_{\varphi}\llbracket\lambda\rrbracket(f)$ of Equation \eqref{Eq: SPDE}, that is 
\begin{align}\label{Eq: Expectation value of the perturbative solution}
	\mathbb{E}(\widehat{\psi}_{\varphi}\llbracket\lambda\rrbracket(f))
	=\Gamma_{\cdot_Q}(\Psi\llbracket\lambda\rrbracket)(f;\varphi)
	=\sum_{j\geq0}\lambda^j\Gamma_{\cdot_Q}(F_j)(f;\varphi)\,,
\end{align}
where $\widehat{\psi}$ is the formal solution of Equation \eqref{Eq: SPDE}, while $\mathbb{E}$ stands for the expectation value. The subscript $\varphi$ reminds us that we are free to consider a shifted white noise, namely $\mathbb{E}(P_\chi\circledast\xi)=\varphi$. As remarked in Section \ref{Sec: introduction} the usual scenario of a Gaussian white noise centered at $0$ can be recovered by setting $\varphi=0$.

\begin{remark}\label{Rem: on the arbitrariness of GammaQ}
	It is important to bear in mind that, in the above procedure, there is the arbitrariness in picking $\Gamma_{\cdot_Q}$ as discussed in Theorem \ref{Thm: GammacdotQ uniqueness}. Implicitly we have taken one of the possible choices and different ones yield far reaching consequences in the construction of the perturbative solutions of Equation \eqref{Eq: algebraic SPDE}. In the following we shall discuss in detail these consequences. 
\end{remark}

Before dwelling into specific computations, we stress that a noteworthy advantage of our approach is the predictability at a perturbative level also of the correlations between the solutions of Equation \eqref{Eq: algebraic SPDE}.
This can be obtained using the product $\bullet_Q$ introduced in Section \ref{Sec: Correlations and bulletQ product} and working therefore with the algebra $\mathcal{A}_{\bullet_Q}$ or, rather, with $\mathcal{A}_{\bullet_Q}\llbracket\lambda\rrbracket$, namely the algebra of formal power series in $\lambda$ with coefficients in $\mathcal{A}_{\bullet_Q}$.
We recollect all interesting expressions involving $\Psi$ in two single algebraic objects.
\begin{definition}
	Let $\Gamma_{\cdot_Q}$, $\Gamma_{\bullet_Q}$ two maps as per Theorems \ref{Thm: Gamma cdotQ existence}-\ref{Thm: GammabulletQ existence}.
	We define $\mathcal{A}_{\cdot_Q}^\Psi\subset\mathcal{A}_{\cdot_Q}[\![\lambda]\!]$ as the subalgebra of $\mathcal{A}_{\cdot_Q}[\![\lambda]\!]$ generated by $\Psi_{\cdot_Q}:=\Gamma_{\cdot_Q}(\Psi)$.
	Similarly we denote by $\mathcal{A}_{\bullet_Q}^\Psi\subseteq\mathcal{A}_{\bullet_Q}[\![\lambda]\!]$ the smallest subalgebra of $\mathcal{A}_{\bullet_Q}[\![\lambda]\!]$ containing $\mathcal{A}_{\cdot_Q}^\Psi$.
\end{definition}

As exemplification, in the following we shall consider the two-point correlation function, but the reader is warned that, barring the sheer computational difficulties, we could analyze with the same methods all higher orders. Bearing in mind this comment, on account of Theorem \ref{Thm: GammabulletQ existence} and in particular of Equation \eqref{Eq: GammabulletQ on AcdotQ}, it holds that, for all $f_1,f_2\in\mathcal{D}(\mathbb{R}\times\mathbb{R}^d)$ and for all $\varphi\in\mathcal{E}(\mathbb{R}\times\mathbb{R}^d)$, the linearity of $\Gamma_{\bullet_Q}$ entails that
\begin{align}
	\nonumber
	\omega_2(f_1\otimes f_2;\varphi)\vcentcolon
	&=\big(\Gamma_{\cdot_Q}(\Psi\llbracket\lambda\rrbracket)\bullet_{\Gamma_{\bullet_Q}}\Gamma_{\cdot_Q}(\Psi\llbracket\lambda\rrbracket)\big)(f_1\otimes f_2;\varphi)
	\\\nonumber
	&=\Gamma_{\bullet_Q}\big(\Gamma_{\cdot_Q}(\Psi\llbracket\lambda\rrbracket)\otimes\Gamma_{\cdot_Q}(\Psi\llbracket\lambda\rrbracket)\big)(f_1\otimes f_2;\varphi)
	\\\label{Eq: two point correlation}
	&=\sum_{k\geq0}\lambda^k\sum_{j=0}^k\Gamma_{\bullet_Q}\big(\Gamma_{\cdot_Q}(F_j)\otimes\Gamma_{\cdot_Q}(F_{k-j})\big)(f_1\otimes f_2;\varphi)\,.
\end{align}

\subsection{First order in perturbation theory}\label{Sec: computations at first order}
In this section we compute the expectation value of the solution, as well as its two-point correlation at the first order in perturbation theory.
\paragraph{Expectation value of the solution.}
Using Equation \eqref{Eq: formal power series} together with Equation \eqref{Eq: order-by-order perturbative expansion}, the solution of Equation \eqref{Eq: algebraic SPDE} to order $O(\lambda^2)$ reads 
\begin{align}\label{Eq: Perturbative Expansion of Psi}
	\Psi\llbracket\lambda\rrbracket
	=\Phi
	-\lambda P_\chi\circledast\Phi^3
	+O(\lambda^2)\,.
\end{align}

On account of Equation \eqref{Eq: Gamma on M1}, $\Gamma_{\cdot_Q}(\Phi)=\Phi$, whereas, by Equations \eqref{Eq: Gamma on Ptau} and \eqref{Eq: definition of Phi2} together with the general construction of the map $\Gamma_{\cdot_Q}$ in the proof of Theorem \ref{Thm: Gamma cdotQ existence},
\begin{align*}
	\Gamma_{\cdot_Q}(P_\chi\circledast\Phi^3)
	=P_\chi\circledast\Gamma_{\cdot_Q}(\Phi^3)
	=P_\chi\circledast(\Phi^3+3C\Phi)\,.
\end{align*}
Here $C\in \mathcal{E}(\mathbb{R}\times\mathbb{R}^d)$ is a consequence of Theorem \ref{Thm: Gamma cdotQ existence}.
More precisely $C(z)=\chi(z)\widehat{P}_2(\delta_z\otimes\chi)$, with $\widehat{P}_2$ the chosen extension to the whole $(\mathbb{R}\times\mathbb{R}^d)^2$ of the bi-distribution $P^2\in\mathcal{D}^\prime((\mathbb{R}\times\mathbb{R}^d)^2\setminus\mathrm{Diag}_2)$.
At this stage we do not enter into the details of the explicit form of $\widehat{P}_2$ postponing it to the last part of this Section. 

\begin{remark}\label{Rem: C as a renromalization constant}
	It is worth emphasizing that, choosing a specific product $\Gamma_{\cdot_Q}$ entails in turn that we have also made a choice of $\widehat{P}_2$, an explicit extension of $P^2$. In this perspective the function $C$ is fixed. Yet, as discussed in the previous sections, there is a freedom in selecting $\Gamma_{\cdot_Q}$ codified by Theorem \ref{Thm: GammacdotQ uniqueness}. In the case in hand, see also Remark \ref{Rmk: uniqueness for Phi2}, this translates in different choices for $C$, which is thus referred to as a {\em renormalization ambiguity}.
%	Notice that in Remark \ref{Rmk: uniqueness for Phi2} $C$ was indicated with $c_{2,0}$, but we opted to using a different symbol for simplicity of the notation.
\end{remark}

\noindent Putting together all this information, we end up with
\begin{align}\label{Eq: perturbative solution SPDE}
	\Gamma_{\cdot_Q}(\Psi\llbracket\lambda\rrbracket)(\varphi)
	=\Phi(\varphi)-\lambda P_\chi\circledast(\Phi^3+3C\Phi)(\varphi)+O(\lambda^2)\,.
\end{align}
Up to order $O(\lambda^2)$, this is the expectation value of the $\varphi$-shifted solution $\widehat{\psi}$ of Equation \eqref{Eq: SPDE}.
Recall that, in order to reproduce the standard white noise behaviour, one has to evaluate this functional at $\varphi=0$.
Using the notation of Equation \eqref{Eq: Expectation value of the perturbative solution}, this entails that
\begin{align*}
	\mathbb{E}(\widehat{\psi}_0\llbracket\lambda\rrbracket)=O(\lambda^2)\,.
\end{align*}
The latter equation holds at all orders in perturbation theory, as shown by the following lemma.
\begin{lemma}\label{Lem: expectation value of solution}
	Let $\Psi\in\mathcal{A}[\![\lambda]\!]$ be defined as per Equation \eqref{Eq: order-by-order perturbative expansion}.
	Then $\Gamma_{\cdot_Q}(\Psi)(f;0)=0$ for all $f\in\mathcal{D}(M)$.
\end{lemma}
\begin{proof}
	Let $\mathcal{O}\subset\mathcal{A}$ be the vector space of $\mathcal{A}$ made by elements $\tau\in\mathcal{A}$ such that $\tau^{(2n)}(\cdot;0)=0$ for all $n\in\mathbb{N}_0$. Observe that the superscript indicates the $2n$-th functional derivative and, with a slight abuse of notation, we avoid indicating the directions of derivation as well as the test function, {\it cf.} Definition \ref{Def: functional-valued distributions}. 
	We observe that, for all $\tau\in\mathcal{O}$, it holds $\Gamma_{\cdot_Q}(\tau)\in\mathcal{O}$, moreover, $\tau_1\tau_2\tau_3\in\mathcal{O}$ for all $\tau_1,\tau_2,\tau_3\in\mathcal{O}$.
	
	We now prove the thesis by showing that $\Psi\in\mathcal{O}[\![\lambda]\!]$, namely that $F_j\in\mathcal{O}$ for all $j\in\mathbb{N}_0$, being $F_j\in\mathcal{A}$ defined as per Equation \eqref{Eq: order-by-order perturbative expansion}.
	For $j=0$ we have $F_j=\Phi\in\mathcal{O}$.
	Then Equation \eqref{Eq: order-by-order perturbative expansion} entails that, for all $j\in\mathbb{N}$,
	\begin{align*}
		F_j=-\sum_{j_1+j_2+j_3=j-1}P_\chi\circledast( F_{j_1}F_{j_2}F_{j_3})\,.
	\end{align*}
	By induction $F_{j_1}, F_{j_2}, F_{j_3}\in\mathcal{O}$ and, by direct inspection, $F_{j_1}F_{j_2}F_{j_3}\in\mathcal{O}$.
	This implies that $F_j\in\mathcal{O}$.
\end{proof}

\paragraph{Two-point correlation function.}
We compute, up to order $O(\lambda^2)$, the two-point correlation function of the solution of Equation \eqref{Eq: algebraic SPDE}. Using Equation \eqref{Eq: two point correlation}, at the zeroth-order in $\lambda$, we need to evaluate for $f_1,f_2\in\mathcal{D}(\mathbb{R}\times\mathbb{R}^d)$ and for $\varphi\in\mathcal{E}(\mathbb{R}\times\mathbb{R}^d)$
\begin{align*}
\Gamma_{\bullet_Q}(\Gamma_{\cdot_Q}(\Phi)\otimes\Gamma_{\cdot_Q}(\Phi))(f_1\otimes f_2;\varphi)=(\Phi\otimes\Phi)(f_1\otimes f_2;\varphi)+Q(f_1\otimes f_2)\,,
\end{align*}
whereas at first order
\begin{align*}
	\Gamma_{\bullet_Q}(\Gamma_{\cdot_Q}(\Phi)&\otimes\Gamma_{\cdot_Q}(P_\chi\circledast\Phi^3))(f_1\otimes f_2;\varphi)
	\\&=(\Phi\otimes(P_\chi\circledast(\Phi^3+3C\Phi))(f_1\otimes f_2;\varphi)
	+Q\cdot(1\otimes 3P_\chi\circledast(\Phi^2+C\mathbf{1})(f_1\otimes f_2;\varphi)\,,
\end{align*}
where $C$ is the smooth function introduced above, $\mathbf{1}$ denotes the identity functional and where $Q\cdot(1\otimes3P_\chi\circledast(\Phi^2+C\mathbf{1})(\varphi)$ is the bi-distribution whose integral kernel is $3Q(x,y)(P_\chi\circledast(\varphi^2+C))(y)$ -- we recall that $Q=P_\chi\circ P_\chi$.
As a consequence, we get
\begin{multline}\label{Eq: first order two point function}
	\omega_2(f_1\otimes f_2;\varphi)
	=[\Phi\otimes\Phi
	+Q-\lambda\big(\Phi\otimes(P_\chi\circledast(\Phi^3+3C\Phi))\\
	+3Q\cdot(1\otimes(P_\chi\circledast(\Phi^2+C\mathbf{1}))\big)](f_1\otimes f_2+f_2\otimes f_1;\varphi)
	+O(\lambda^2)\,.
\end{multline}
As in the previous case, the scenario with the standard white noise as a source can be recovered evaluating Equation \eqref{Eq: first order two point function} at $\varphi=0$. This yields 
\begin{gather*}
	\mathbb{E}(\widehat{\psi}\llbracket\lambda\rrbracket\otimes\widehat{\psi}\llbracket\lambda\rrbracket)(f_1\otimes f_2)
	=\omega_2(f_1\otimes f_2;0)=\\
	=Q(f_1\otimes f_2)
	-3\lambda Q\cdot(1\otimes(P_\chi\circledast C))(f_1\otimes f_2+f_2\otimes f_1)+O(\lambda^2).
\end{gather*}
Notice that in the computation the only freedom appearing is still codified by the lone function $C$ and thus the same comments as per Remark \ref{Rem: C as a renromalization constant} apply.

\paragraph{Explicit construction of $P^n$.}
In the previous formulae, particularly Equation \eqref{Eq: perturbative solution SPDE} and \eqref{Eq: first order two point function}, the fundamental solution of the heat operator enters the game together with the (arbitrarily chosen) extension $\widehat{P}_2\in\mathcal{D}'((\mathbb{R}\times\mathbb{R}^d)^2)$ of $P^2\in\mathcal{D}'((\mathbb{R}\times\mathbb{R}^d)^2\setminus\mathrm{Diag}_2)$, which contributes through the function $C$ -- \textit{cf.} Equation \eqref{Eq: first order two point function}.

For concreteness, in this paragraph we wish to discuss an explicit extension procedure for powers of the fundamental solution of the heat operator, to prove that our method can yield explicit expressions.

To this avail we observe that the fundamental solution $P\in\mathcal{D}'((\mathbb{R}\times\mathbb{R}^d)^2)$ of $\partial_t-\Delta$ is translation invariant -- that is, $P(t,x;s,y)=\mathsf{p}(t-s,x-y)$ where $\mathsf{p}\in\mathcal{D}'(\mathbb{R}\times\mathbb{R}^d)$ is the fundamental solution of $\partial_t-\Delta$ on the Euclidean space $\mathbb{R}^{d+1}$ -- \textit{cf.} Equation \eqref{Eq: heat fundamental solution}.

In the following we discuss the extensions of $\mathsf{p}^n\in\mathcal{D}'(\mathbb{R}^{d+1}\setminus\{0\})$ and, to this avail, it is convenient to start from the fundamental solution of $\partial_t-\kappa\Delta$ on $\mathbb{R}^{d+1}$
\begin{align}\label{Eq: heat fundamental solution}
	\mathsf{p}_\kappa(t,x)
	:=\frac{1}{(4\pi\kappa t)^{\frac{d}{2}}}\Theta(t)e^{-\frac{|x|^2}{4\kappa t}}\,,
\end{align}
where $\Theta$ is the Heaviside step function.
Per construction $\mathsf{p}_\kappa\in\mathcal{D}'(\mathbb{R}^{d+1})$ satisfies
\begin{align}
	(\partial_t-\kappa\Delta)\mathsf{p}_\kappa=\delta\,,
\end{align}
where $\delta\in\mathbb{D}'(\mathbb{R}^{d+1})$ is the Dirac delta distribution centred at the origin.
In what follows we shall denote $\mathsf{p}_1=\mathsf{p}$.

For all $n\in\mathbb{N}$ let us consider $\mathsf{p}^{n+1}$.
Since $\operatorname{WF}(\mathsf{p})=\operatorname{WF}(\delta)$,  $\mathsf{p}^{n+1}\in\mathcal{D}'(\mathbb{R}^{d+1}\setminus\{0\})$.
If $d=1$ and $n=1$, then $\mathsf{p}^2\in\mathcal{D}^\prime(\mathbb{R}^2)$ since
\begin{align*}
	\mathsf{p}^2(f)
%	=\int\mathrm{d}t\mathrm{d}x\mathsf{p}(t,x)^2f(t,x)
	=\int_0^{+\infty}\frac{\mathrm{d}t}{\sqrt{t}}\int\limits_{\mathbb{R}}\mathrm{d}x
	\frac{1}{4\pi}e^{-\frac{|x|^2}{2}}f(t,\sqrt{t}x)
	<+\infty\,,
\end{align*}
On the contrary, for $d\geq 2$ and $n\geq 2$, the singularity at the origin calls for an extension procedure.
To this end we use the identity
\begin{align*}
	\frac{1}{t^\alpha}
	=\frac{1}{\Gamma(\alpha)}\int_0^{+\infty}\mathrm{d}z z^{\alpha-1}e^{-tz}\,.
\end{align*}
Replacing it in the integral kernel of $\mathsf{p}^{n+1}\in\mathcal{D}'(\mathbb{R}^{d+1}\setminus\{0\})$, we obtain  the following K\"all\'en-Lehmann type formula
\begin{align}
	\nonumber
	\mathsf{p}(t,x)^{n+1}
	&=\frac{1}{(4\pi t)^{\frac{(n+1)d}{2}}}\Theta(t)e^{-\frac{(n+1)|x|^2}{4t}}\,.
	\\
	\label{Eq: Kallen-Lehmann heat formula}
	&=\frac{1}{(4\pi)^{\frac{nd}{2}}}\frac{1}{(n+1)^{\frac{d}{2}}}\frac{1}{\Gamma(\frac{nd}{2})}
	\int_0^{+\infty}\mathrm{d}z z^{\frac{nd}{2}-1}\mathsf{p}_{\frac{1}{n+1},z}(t,x)\,.
\end{align}
where $\mathsf{p}_{\frac{1}{n+1},z}$ is the fundamental solution of the parabolic equation
\begin{align}\label{Eq: mass-heat fundamenta solution}
	\left[\partial_t-\frac{1}{n+1}\Delta+z\right]\mathsf{p}_{\frac{1}{n+1},z}
	=\delta\,,\qquad
	\mathsf{p}_{\frac{1}{n+1},z}(t,x)
	:=\frac{(n+1)^{\frac{d}{2}}}{(4\pi t)^{\frac{d}{2}}}\Theta(t)e^{-\frac{(n+1)|x|^2}{4t}}e^{-zt}\,.
\end{align}
Equation \eqref{Eq: Kallen-Lehmann heat formula} represents the singularity of $\mathsf{p}^{n+1}$ at the origin in terms of the divergent integral in $z$.
This suggests a specific extension $\widehat{\mathsf{p}}^{n+1}$ of $\mathsf{p}^{n+1}$.
To wit, for $a>0$ we define $_a\mathsf{p}^{n+1}\in\mathcal{D}'(\mathbb{R}^{d+1})$ via the integral kernel
\begin{align}
	_a\mathsf{p}^{n+1}(t,x)
	:=\frac{1}{(4\pi)^{\frac{nd}{2}}}\frac{1}{(n+1)^{\frac{d}{2}}}\frac{1}{\Gamma(\frac{nd}{2})}
	\left[-\partial_t+\frac{1}{n+1}\Delta+a\right]^\ell\int_0^{+\infty}\mathrm{d}z \frac{z^{\frac{nd}{2}-1}}{(z+a)^\ell}\mathsf{p}_{\frac{1}{n+1},z}(t,x)\,,
\end{align}
where $\ell=\lfloor\frac{nd}{2}\rfloor\leq\frac{nd}{2}$.
Notice that this choice makes the $z$-integral convergent after smearing it against a compactly supported function $f\in\mathcal{D}(\mathbb{R}^{d+1})$.
By direct inspection, once restricted to $(t,x)\neq(0,0)$, $_a\mathsf{p}^{n+1}(t,x)=\mathsf{p}(t,x)^{n+1}$. Furthermore, the weighted scaling degree of $_a\mathsf{p}^{n+1}$ at the origin coincides with the one of $\mathsf{p}^{n+1}$, {\it cf.} Remark \ref{Rmk: weighted scaling degree}, though one should keep in mind that here the dimension of $M$ is $d+1$.

A different choice for $a>0$ leads to a result consistent with Theorem \ref{Thm: extension with scaling degree} and  Remark \ref{Rmk: weighted scaling degree}.
In particular, for all $a,b>0$,  $_a\mathsf{p}^{n+1}-\;_b\mathsf{p}^{n+1}$ is a linear combination of derivatives of Dirac delta distributions with weighted scaling degree at most $(n+1)d$.
To see this, we observe that, for all $m\in\mathbb{N}$ and $a\in\mathbb{R}$,
\begin{align*}
	(H_{\frac{1}{n+1}}+a)^m\mathsf{p}_{\frac{1}{n+1},z}
	=(z+a)^m\mathsf{p}_{\frac{1}{n+1},z}
	-\sum_{j=0}^{m-1}(z+a)^{m-1-j}(H_{\frac{1}{n+1}}+a)^j\delta\,,
\end{align*}
where $H_{\frac{1}{n+1}}:=-\partial_t+\frac{1}{n+1}\Delta$ is a short notation.
It then follows that -- setting $c_{n,d}:=\frac{1}{(4\pi)^{\frac{nd}{2}}}\frac{1}{(n+1)^{\frac{d}{2}}}\frac{1}{\Gamma(\frac{nd}{2})}$ --
\begin{align*}
	_{a+b}\mathsf{p}^{n+1}(t,x)
	-\;_{a}\mathsf{p}^{n+1}(t,x)
	&=c_{n,d}
	\left[(H_{\frac{1}{n+1}}+a+b)^\ell-(H_{\frac{1}{n+1}}+a)^\ell\right]
	\int_0^{+\infty}\mathrm{d}z \frac{z^{\frac{nd}{2}-1}}{(z+a+b)^\ell}\mathsf{p}_{\frac{1}{n+1},z}(t,x)
	\\&+
	c_{n,d}
	(H_{\frac{1}{n+1}}+a)^\ell
	\int_0^{+\infty}\mathrm{d}zz^{\frac{nd}{2}-1}
	\left[\frac{1}{(z+a+b)^\ell}-\frac{1}{(z+a)^\ell}\right]\mathsf{p}_{\frac{1}{n+1},z}(t,x)\,.
\end{align*}
With standard algebraic manipulations we find
\begin{align*}
	&_{a+b}\mathsf{p}^{n+1}
	-\;_{a}\mathsf{p}^{n+1}
	\\&=c_{n,d}
	\sum_{j=0}^{\ell-1}{\ell\choose j}b^{\ell-j}
	\int_0^{+\infty}\mathrm{d}z \frac{z^{\frac{nd}{2}-1}}{(z+a+b)^\ell(z+a)^\ell}
	\left[(H_{\frac{1}{n+1}}+a)^j(z+a)^\ell-(z+a)^j(H_{\frac{1}{n+1}}+a)^\ell\right]\mathsf{p}_{\frac{1}{n+1},z}
	\\&=c_{n,d}
	\sum_{j=0}^{\ell-1}{\ell\choose j}b^{\ell-j}
	\int_0^{+\infty}\mathrm{d}z \frac{z^{\frac{nd}{2}-1}}{(z+a+b)^\ell(z+a)^\ell}
	\sum_{q=j}^{\ell-1}(z+a)^{\ell+j-q-1}(H_{\frac{1}{n+1}}+a)^q\delta
	\\&=\sum_{q=0}^{\ell-1}\zeta_q(H_{\frac{1}{n+1}}+a)^q\delta\,,
\end{align*}
where $\zeta_j\in\mathbb{C}$.
Notice that the weighted scaling degree at the origin of $(H_{\frac{1}{n+1}}+a)^q\delta$ is at most $2\ell-2+d-2=(n+1)d$ as required.

\begin{remark}
	It is worth mentioning that our analysis can be repeated almost word by word also in the case where the underlying space is $\mathbb{R}\times\mathbb{T}^d$, $\mathbb{T}^d$ being the flat $d$-torus. In this case the counterpart of $\mathsf{p}$ is played by the distribution $p$ obtained via a Poisson formula as
	\begin{align}\label{Eq: Poisson formula}
		p(t,x):=\sum_{n\in\mathbb{Z}}\mathsf{p}(t,x+n)\,,\qquad t\in\mathbb{R}\,,\,x\in(0,1]\,.
	\end{align}
	Here we realized $\mathbb{T}^d=\mathbb{R}^d/\mathbb{Z}^d\simeq(0,1]^d$. Equation \eqref{Eq: Poisson formula} is particularly relevant because it shows that the singular structure of $p$ is the same as that of $\mathsf{p}$ and in particular $p^n-\mathsf{p}^n$ is smooth for all $n\in\mathbb{N}$.
\end{remark}

\subsection{Renormalized equation}\label{Sec: Renormalized equation}

In this section we construct the equation of which $\Psi_{\cdot_Q}$ is a formal solution and we shall refer to it as the {\em renormalized equation}. Observe that the renormalization procedure outlined in Section \ref{Sec: computations at first order} which makes use of the results of Section \ref{Sec: Correlations and bulletQ product} and of Section \ref{Sec: Uniqueness theorems} affects the structure of Equation \eqref{Eq: SPDE}.

To this end recall that $\Psi\in\mathcal{A}[\![\lambda]\!]$ is the perturbative solution of Equation \eqref{Eq: algebraic SPDE}, constructed in Equation \eqref{Eq: formal power series} and \eqref{Eq: order-by-order perturbative expansion}.
If one applies $\Gamma_{\cdot_Q}$ to Equation \eqref{Eq: algebraic SPDE},  Theorem \ref{Cor: construction of AcdotQ} -- \textit{cf.} Equation \eqref{Eq: cdot GammaQ product} -- entails that $\Psi_{\cdot_Q}$ obeys to
\begin{align}\label{eq: full renormalized equation}
	\Psi_{\cdot_Q}=\Phi-\lambda P_\chi\circledast(\Psi_{\cdot_Q}\cdot_Q\Psi_{\cdot_Q}\cdot_Q\Psi_{\cdot_Q})\,.
\end{align}
The term $\Psi_{\cdot_Q}\cdot_Q\Psi_{\cdot_Q}\cdot_Q\Psi_{\cdot_Q}$ is a signature that we are representing the full stochastic content of equation \eqref{Eq: algebraic SPDE}. However, for practical purposes, it may be desirable to get rid of the $\cdot_Q$-product, recasting Equation \ref{Eq: cdot GammaQ product} in terms of the pointwise counterpart. The price to pay for such change is a modification of Equation \eqref{Eq: SPDE}.

\begin{proposition}\label{Prop: renormalized equation}
	Let $\Psi\in\mathcal{A}[\![\lambda]\!]$ be a solution of Equation \eqref{Eq: algebraic SPDE} and let $\Psi_{\cdot_Q}:=\Gamma_{\cdot_Q}(\Psi)$.
	Then there exists a sequence of functional-valued linear operators $\{M_n\}_{n\in\mathbb{N}}$ such that
	\begin{enumerate}[(i)]
		\item
		for all $n\in\mathbb{N}$ and for all $\varphi\in\mathcal{E}(\mathbb{R}^{d+1})$, $M_n(\varphi)\colon\mathcal{E}(\mathbb{R}^{d+1})\to\mathcal{E}(\mathbb{R}^{d+1})$.
		\item
		for all $n\in\mathbb{N}$, $M_n(\varphi)$ has a polynomial dependence on $\varphi$ and, moreover, for all $j\in\mathbb{N}$, $M_n^{(2j+1)}(0)=0$, where the superscript indicates the $(2j+1)$-th functional derivative.
		\item
		$\Psi_{\cdot_Q}$ satisfies the following equation
		\begin{align}\label{Eq: renormalized equation}
			\Psi_{\cdot_Q}=\Phi-\lambda P_\chi\circledast \Psi_{\cdot_Q}^3-P_\chi\circledast M\Psi_{\cdot_Q}\,,
		\end{align}
	where $M:=\sum_{n\geq 1}\lambda^nM_n$, while $\chi$ is the cut-off function introduced in Equation \eqref{Eq: algebraic SPDE}.
	\end{enumerate}
\end{proposition}
\begin{proof}
	The proof is constructive and it goes by induction over the perturbative order $n$.
	In particular, for $n=0$, Equation \eqref{Eq: renormalized equation} is satisfied since $\Psi_{\cdot_Q}=\Phi+O(\lambda)$.
	For $n=1$, focusing first on the left hand side (LHS) and subsequently on the right hand side (RHS) of Equation \eqref{Eq: renormalized equation}, it holds
	\begin{align*}
		&\textrm{LHS}
		=\Phi-\lambda P_\chi\circledast\Gamma_{\cdot_Q}(\Phi^3)+O(\lambda^2)
		=\Phi-\lambda P_\chi\circledast\Phi^3-3\lambda P_\chi\circledast C_1\Phi+O(\lambda^2)\\
		&\textrm{RHS}
		=\Phi-\lambda P_\chi\circledast \Phi^3-\lambda P_\chi\circledast M_1\Phi+O(\lambda^2)\,,
	\end{align*}
	where $C_1(x):=\chi(x)(\widehat{P}_2\circledast\chi)(x)$, $\widehat{P}_2\in\mathcal{D}'((\mathbb{R}^{d+1})^2)$ being an extension of $P^2\in\mathcal{D}'((\mathbb{R}^{d+1})^2\setminus\mathrm{Diag}_2)$, {\it cf.} Section \ref{Sec: computations at first order}.
	We can set $M_1:=3C_1$ and it satisfies \textit{(i-ii)} as well as Equation \eqref{Eq: renormalized equation} up to order $O(\lambda^2)$.
	
	By induction, we now assume that $M_k$ has been defined for all $k\leq n$ compatibly with all requirements \textit{(i-ii)} and in such a way that Equation \eqref{Eq: renormalized equation} is satisfied up to order $O(\lambda^{n+1})$. We show that we can construct $M_{n+1}$ so that Equation \eqref{Eq: renormalized equation} holds true up to order $O(\lambda^{n+2})$.
	
	To this end, we expand Equation \eqref{Eq: renormalized equation} up to order $O(\lambda^{n+2})$ and we find
	\begin{align*}
		\textrm{LHS}
		&=R_n
		-\lambda^{n+1}\sum_{j_1+j_2+j_3= n}P_\chi\circledast\Gamma_{\cdot_Q}(\Psi_{j_1}\Psi_{j_2}\Psi_{j_3})
		+O(\lambda^{n+2})
		\\\textrm{RHS}
		&=R_n
		-\lambda^{n+1}\sum_{j_1+j_2+j_3= n}P_\chi\circledast\Gamma_{\cdot_Q}(\Psi_{j_1})\Gamma_{\cdot_Q}(\Psi_{j_2})\Gamma_{\cdot_Q}(\Psi_{j_3})
		\\&-\lambda^{n+1}\sum_{\substack{k_1+k_2= n+1\\k_2\geq 1}}P_\chi\circledast M_{k_1}\Gamma_{\cdot_Q}(\Psi_{k_2})
		-\lambda^{n+1}P_\chi\circledast M_{n+1}\Phi
		+O(\lambda^{n+2})\,,
	\end{align*}
	where we isolated the summand containing $M_{n+1}$. Here $R_n$ represents the lowest order contributions for which Equation \eqref{Eq: renormalized equation} holds by the inductive hypothesis.
	We scrutinize thoroughly the remaining terms. To this avail, recall that the vector space $\mathcal{O}\subset\mathcal{A}$ is built out of elements $\tau\in\mathcal{A}$ such that $\tau^{(2n)}(\cdot;0)=0$ for all $n\in\mathbb{N}_0$. The proof of Lemma \ref{Lem: expectation value of solution} entails that $\Psi\in\mathcal{O}$ and we observe that all contributions under analysis are of the form $P_\chi\circledast F$ where $F\in\mathcal{O}$. This is a by product of the left hand side of Equation \eqref{Eq: renormalized equation} since $\Psi\in\mathcal{O}$, \textit{cf.} Lemma \ref{Lem: expectation value of solution}. This holds true also for $M_{k_1}\Gamma_{\cdot_Q}(\Psi_{k_2})$ on account of the inductive hypothesis for $M_k$.

	Observe that any element $\tau\in\mathcal{O}$ can be written as $\tau=L\Phi$ where, for all $\varphi\in\mathcal{E}(\mathbb{R}^{d+1})$, $L(\varphi)\colon\mathcal{E}(\mathbb{R}^{d+1})\to\mathcal{E}(\mathbb{R}^{d+1})$, moreover, $L(\varphi)$ has a polynomial dependence on $\varphi$ with $L^{(2j+1)}(0)=0$ for all $j\in\mathbb{N}$. The superscript still indicates the functional derivative.
	Overall we obtained
	\begin{align*}
		\Psi_{\cdot_Q}-\Phi+\lambda P_\chi\circledast \Psi_{\cdot_Q}^3+P_\chi\circledast M\Psi_{\cdot_Q}
		=\lambda^{n+1}[P_\chi\circledast L\Phi-P_\chi\circledast M_{n+1}\Phi]
		+O(\lambda^{n+2})\,.
	\end{align*}
	We set $M_{n+1}:=L$, which satisfies \textit{(i-ii)}. Hence the induction step is complete and we have proven the sought after result.
\end{proof}

\begin{remark}
	It is worth computing $M$ at second order in perturbation theory, so to pinpoint the non-local behaviour encoded in the operator $M_2$. To this end observe that, up to $O(\lambda^3)$, Equation \eqref{Eq: renormalized equation} leads to
	\begin{align*}
		&\textrm{LHS}
		=R_1
		+3\lambda^2P_\chi\circledast\Gamma_{\cdot_Q}(\Phi^2P_\chi\circledast\Phi^3)
		+O(\lambda^3)\\
		&\textrm{RHS}
		=R_1
		+3\lambda^2P_\chi\circledast(\Phi^2\Gamma_{\cdot_Q}(P_\chi\circledast\Phi^3))
		+\lambda^2 P_\chi\circledast( M_1\Gamma_{\cdot_Q}(P_\chi\circledast\Phi^3))
		-\lambda^2P_\chi\circledast M_2\Phi
		+O(\lambda^3)\,,
	\end{align*}
	where, as above, $M_1=3C_1$, with $C_1(x)=\chi(x)(\widehat{P}_2\circledast\chi)(x)$ while $R_1$ are lower order contributions.
	Fulfilling Equation \eqref{Eq: renormalized equation} modulo $O(\lambda^3)$ entails that
	\begin{align*}
		P_\chi\circledast M_2\Phi=-18P_\chi\circledast\big[[(P_\chi\circ P_\chi)\cdot P_\chi\circledast(\Phi^2+C_1)+C_2]\Phi\big]\,,
	\end{align*}
where $C_2(x,y)=\widehat{P_\chi\cdot(P_\chi\circ P_\chi)^2}(x,y)\in\mathcal{D}'((\mathbb{R}^{d+1})^2)$ denotes the chosen extension of the bi-distribution $P_\chi \cdot(P_\chi\circ P_\chi)^2\in\mathcal{D}'((\mathbb{R}^{d+1})^2\setminus\mathrm{Diag}_2)$ on the whole space.
This leads to setting $M_2$ as
	\begin{align*}
		M_2=-18[(P_\chi\circ P_\chi)P_\chi\circledast(\Phi^2+C_1)+C_2]\,,
	\end{align*}
	which is a non local operator.
\end{remark}

\begin{remark}\label{Rem: comparison with Hairer}
	A reader might wonder whether the renormalized equation that we construct is connected with the one derived in \cite[Prop 4.9]{Hairer15}. This is not a straightforward comparison since, in \cite{Hairer15}, the r\^{o}le of $\widehat{\xi}$ in Equation \eqref{Eq: SPDE} is played by a smooth function, which could be chosen for example as an $\varepsilon$-regularized smooth version of a white noise. On the contrary our analysis is devised intrinsically to encode the singularities of the white noise in the product of the algebra of functionals and, in order to connect the two approaches it would be necessary to start from \cite{Hairer15} and to discuss the limit as $\varepsilon\to 0$ of the renormalized equation. We postpone such analysis to future work. 
\end{remark}

\subsection{Diverging diagrams in the sub-critical case}\label{Sec: Diverging diagrams in the sub-critical case}

In the perturbative analysis of the stochastic $\Phi^3_d$ model outlined above, the only occurring divergence was related to $P^2$. As a consequence we had to account for a so-called renormalization freedom encoded in the function $C$ -- \textit{cf.} Equation \eqref{Eq: perturbative solution SPDE}.

It is thus natural to wonder whether, carrying the perturbative analysis to all orders in $\lambda$, the number of divergences, which need to be tamed, becomes infinite. In this section we 
show that in the so called sub-critical regime, corresponding to $d\leq3$, this is not the case. Observe that this scenario is the same which can be discussed in the framework of \cite{Hairer14,Hairer15}.

To give an answer to this issue, it is of paramount relevance observing that we do not need to consider the construction of the map $\Gamma_{\cdot_Q}$, {\it cf.} Theorem \ref{Thm: Gamma cdotQ existence}, on the whole pointwise algebra $\mathcal{A}$. On the contrary Equation \eqref{Eq: algebraic SPDE} and its perturbative expansion involve only a suitable limited number of elements of $\mathcal{A}_{\cdot_Q}$.

\begin{remark}\label{Rem: A convincing example}
To convince the reader of the last assertion observe that at order $\lambda^2$, the perturbative expansion of Equation \eqref{Eq: algebraic SPDE} leads to
\begin{align*}
	\Psi[\![\lambda]\!]=\sum_{j\geq 0}\lambda^j F_j\,,\quad
	F_0=\Phi,\quad
	F_j=-\sum_{j_1+j_2+j_3=j-1}P_\chi\circledast(F_{j_1}F_{j_2}F_{j_3})\,,\qquad j\in\mathbb{N}\,.
\end{align*}	
This formula shows manifestly that, increasing the order in $\lambda$ always leads to an increase of the number of fields $\Phi$ involved.
For example, using the notation of Remark \ref{Rem: graded algebra}, no elements of $\mathcal{M}_{2k}$ for any $k\in\mathbb{N}$ occurs in the perturbative expansion (see Lemma \ref{Lem: expectation value of solution}).
\end{remark}

To account for this new feature, we introduce a collection $\mathcal{U}$ of formal expressions which is the smallest collection of elements of $\mathcal{A}$ containing $\Phi$ and $\mathbf{1}$ and such that the following implication holds true
\begin{align*}
\tau_1,\tau_2,\tau_3\in\mathcal{U}\quad\Rightarrow\quad P_\chi\circledast(\tau_1\tau_2\tau_3)\in\mathcal{U}.
\end{align*}
Accordingly we set
\begin{align*}
\mathcal{W}\vcentcolon=\lbrace\tau_1\tau_2\tau_3\,:\,\tau_i\in\mathcal{U}\rbrace,\quad\textrm{and}\quad T\vcentcolon=\mathrm{Span}_{\mathcal{E}(\mathbb{R}^{d+1})}\lbrace\mathcal{W}\rbrace\subset\mathcal{A}.
\end{align*}
The vector space $T$ contains all possible elements needed in the description of the right hand side of Equation \eqref{Eq: algebraic SPDE}.

\begin{remark}
In order to analyze all possible divergences generated by applying the map $\Gamma_{\cdot_Q}$ to elements of $T$ it is convenient, as customary in this kind of problems, to introduce a graph representation. In particular
\begin{itemize}
	\item we associate with $\Phi$ the symbol \fiammifero,
	\item we join at their roots any two graphs to denote their pointwise product as functionals,
	\item we denote by an edge \,\propagatore\, the composition by $P_\chi$.
%	Recall that, in the example here considered, beside $P_\chi$, it is necessary also to introduce a cut-off function $\chi$ to avoid potential infrared divergences.
\end{itemize}
As concrete examples, consider
\begin{center}
	$\Phi^2=$\begin{tikzpicture}[thick,scale=1.2]
	\filldraw (.15,.25)circle (1pt);
	\filldraw (-.15,.25)circle (1pt);
	\draw (0,0) -- (.15,.25);
	\draw (0,0) -- (-.15,.25);
	\end{tikzpicture}\,,\hspace{15mm}
	$P_\chi\circledast\Phi^2=$\begin{tikzpicture}[thick,scale=1.2]
	\filldraw (.15,.35)circle (1pt);
	\filldraw (-.15,.35)circle (1pt);
	\draw (0,.2) -- (.15,.35);
	\draw (0,.2) -- (-.15,.35);
	\draw (0,0) -- (0,.2);
	\end{tikzpicture}\,.
\end{center}
The whole notation is inspired by that of \cite{Hairer14}.
Notice in particular that $\Phi=$ \fiammifero\, is motivated by the fact that the functional $\Phi$ encodes the expectation value of $P_\chi\circledast\xi$ and, in this sense, the symbol $\bullet$ on top of \,\fiammifero\, can be seen as denoting the noise $\xi$.
\end{remark}

As next step we need to encode at the level of diagrams the action of the map $\Gamma_{\cdot_Q}$. In view of the analysis in Sections \ref{Sec: Basic definitions} and \ref{Sec: construction of AcdotQ}, this can be translated as collapsing two leaves in an integration vertex. As an example, \emph{cf}. Equation \eqref{Eq: definition of Phi2},
\begin{center}
$\Gamma_{\cdot_Q}(\Phi^2)=$\begin{tikzpicture}[thick,scale=1.2]
\filldraw (.15,.25)circle (1pt);
\filldraw (-.15,.25)circle (1pt);
\draw (0,0) -- (.15,.25);
\draw (0,0) -- (-.15,.25);
\end{tikzpicture}\,+
\begin{tikzpicture}[thick,scale=1.5]
\draw (0,0) edge [out=30,in=-30] node[above] {} (0,.25);
\draw (0,0) edge [out=150,in=210] node[above] {} (0,.25);
\end{tikzpicture}\,.
\end{center}

The general strategy for proving that only a finite number of diagrams needs to be renormalized if $d\leq 3$ goes as follows. When acting with $\Gamma_{\cdot_Q}$ on an element of $T$ we work with $t\in\mathcal{D}'(U)$ where $U\subseteq\mathbb{R}^{(d+1)N}$ for a given $N$.

As we have seen in Theorem \ref{Thm: Gamma cdotQ existence}, tipically $U=\mathbb{R}^{(d+1)N}\setminus\operatorname{Diag}_{(d+1)N}$, being $\operatorname{Diag}_{(d+1)N}$ the total diagonal of $\mathbb{R}^{(d+1)N}$. It follows that we may use the results of Appendix \ref{App: Scaling degree} to discuss the existence of an extension $\widehat{t}$ of $t$ to the whole $\mathbb{R}^{(d+1)N}$.

Using the pictorial representation introduced above, we associate to each distribution $t\equiv t_\mathcal{G}$ a graph $\mathcal{G}$ with the following features:
\begin{itemize}
	\item
		$\mathcal{G}$ has $N$ vertices of valency at most $4$;
	\item
		each edge $e$ of $\mathcal{G}$ corresponds to a propagator $P_\chi(x_{s(e)},x_{t(e)})$, where $s(e)$ (\textit{resp}. $t(e)$) denotes the source (\textit{resp}. the target) of the edge $e$;
	\item
		for a given $\mathcal{G}$, the integral kernel of the distribution $t_{\mathcal{G}}$ reads
		\begin{align*}
			t_{\mathcal{G}}(x_1,\ldots,x_N)
			:=\prod_{e\in \mathcal{E}}P_\chi(x_{s(e)},x_{t(e)})\,,
		\end{align*}
		where $\mathcal{E}$ denotes the set of edges of $\mathcal{G}$.
\end{itemize}

In what follows, for a given graph $\mathcal{G}$ we shall denote by $L$ (\textit{resp.} $V$) the number of edges (\textit{resp}. vertices) of $\mathcal{G}$.
Recalling that the scaling degree of $P_\chi$ on the total diagonal of $\mathbb{R}^{2(d+1)}$ is $d$, it follows that the degree of divergence of the distribution $t_{\mathcal{G}}$ -- \textit{cf.} Theorem \ref{Thm: extension with scaling degree} -- is

\begin{align}\label{Equation: degree of divergence of (L,V)}
	\rho(t_{\mathcal{G}})
	=Ld-(N-1)(d+2)
	=:\rho(\mathcal{G})\,,
\end{align}
since $(N-1)(d+2)$ is the effective codimension of the total diagonal of $\mathbb{R}^{(d+1)N}$ -- \textit{cf.} Remark \ref{Rmk: extension with sd for submanifolds}.

Observe that, if $\rho(\mathcal{G})<0$, Theorem \ref{Thm: extension with scaling degree} applies and the distribution $t_{\mathcal{G}}$ associated with $\mathcal{G}$ admits a unique extension $\widehat{t}_{\mathcal{G}}\in\mathcal{D}'(\mathbb{R}^{dN})$ which preserves the scaling degree of $t_{\mathcal{G}}$.

In what follows we shall show that, provided $\mathcal{G}$ has a sufficiently high number of vertices $N$, then the associated distribution $t_{\mathcal{G}}$ will be such that $\rho(\mathcal{G})<0$, that is, no ambiguities occur.

\begin{remark}
When considering graphs, divergences will only occur from closed subgraphs appearing as a consequence of the action $\Gamma_{\cdot_Q}$, which collapses two leaves into a vertex.
For this reason we will be only interested in leafless graphs.
As a clarification, consider the example
\begin{center}
$\Gamma_{\cdot_Q}(\Phi^2P_\chi\circledast\Phi^2)=$
\begin{tikzpicture}[thick,scale=1.2]
\filldraw (.1,0.45)circle (1pt);
\filldraw (-.1,0.45)circle (1pt);
\filldraw (.1,0.15)circle (1pt);
\filldraw (-.1,0.15)circle (1pt);
\draw (0,0) -- (0,0.3);
\draw (0,0.3) -- (.1,0.45);
\draw (0,0.3) -- (-.1,0.45);
\draw (0,0) -- (.1,0.15);
\draw (0,0) -- (-.1,0.15);
\end{tikzpicture}\,$+$
\begin{tikzpicture}[thick,scale=1.2]
\filldraw (.1,0.45)circle (1pt);
\filldraw (-.1,0.45)circle (1pt);
\draw (0,0) -- (0,0.3);
\draw (0,0.3) -- (.1,0.45);
\draw (0,0.3) -- (-.1,0.45);
\draw (0,0) edge [out=90,in=185] node[above] {} (0.15,0.25);
\draw (0,0) edge [out=0,in=-75] node[above] {} (0.15,0.25);
\end{tikzpicture}\,$+$
\begin{tikzpicture}[thick,scale=1.2]
\filldraw (.1,0.15)circle (1pt);
\filldraw (-.1,0.15)circle (1pt);
\draw (0,0) -- (0,0.3);
\draw (0,0) -- (.1,0.15);
\draw (0,0) -- (-.1,0.15);
\draw (0,0.3) edge [out=90,in=185] node[above] {} (0.15,0.55);
\draw (0,0.3) edge [out=0,in=-75] node[above] {} (0.15,0.55);
\end{tikzpicture}\,$+$
\begin{tikzpicture}[thick,scale=1.2]
\filldraw (-.1,0.45)circle (1pt);
\filldraw (-.1,0.15)circle (1pt);
\draw (0,0) -- (0,0.3);
\draw (0,0.3) -- (.15,0.15);
\draw (0,0.3) -- (-.1,0.45);
\draw (0,0) -- (.15,0.15);
\draw (0,0) -- (-.1,0.15);
\end{tikzpicture}\,$+2$
\begin{tikzpicture}[thick,scale=1.2]
\draw (0,0) -- (0,0.3);
\draw (0,0) edge [out=90,in=185] node[above] {} (0.15,0.25);
\draw (0,0) edge [out=0,in=-75] node[above] {} (0.15,0.25);
\draw (0,0.3) edge [out=90,in=185] node[above] {} (0.15,0.55);
\draw (0,0.3) edge [out=0,in=-75] node[above] {} (0.15,0.55);
\end{tikzpicture}\,$+$
\begin{tikzpicture}[thick,scale=1.2]
\draw (0,0) -- (0,0.3);
\draw (0,0.3) -- (.15,0.15);
\draw (0,0.3) -- (-.15,0.15);
\draw (0,0) -- (.15,0.15);
\draw (0,0) -- (-.15,0.15);
\end{tikzpicture}$\,.$
\end{center}
Once we have discussed the closed diagrams \fish and \rombo, all the other terms contributing to $\Gamma_{\cdot_Q}(\Phi^2P_\chi\circledast\Phi^2)$ are known. To this end, recall that any branch \,\fiammifero\, contributes to a diagram with a smooth factor $\varphi$ and thus no divergence occurs.  As a consequence, when considering a graph, if it contains branches with a leaf, we will {\em cut} them out.\\
Moreover, if we have an admissible diagram of the form $P_\chi\circledast\mathcal{G}$, for some other graph $\mathcal{G}$, then we can work directly with $\mathcal{G}$ since $\Gamma_{\cdot_Q}(P_\chi\circledast\tau)(f;\varphi)=\Gamma_{\cdot_Q}(\tau)(P_\chi\circledast f;\varphi)$ - \emph{cf}. Equation \eqref{Eq: Gamma on Ptau}.

\end{remark}

\noindent We introduce the notion of \emph{admissible graph}.

\begin{definition}
We say that $\mathcal{G}=(L,N)$ is an \emph{admissible graph} if it can be obtained from a tree in $T$ by collapsing some or all of its leaves into vertices.
\end{definition}

\begin{example}
Examples of admissible graphs are: \fish, \rombo, \begin{tikzpicture}[thick,scale=1.2]
\draw (0,0) -- (0,0.3);
\draw (0,0) edge [out=90,in=185] node[above] {} (0.15,0.25);
\draw (0,0) edge [out=0,in=-75] node[above] {} (0.15,0.25);
\draw (0,0.3) edge [out=90,in=185] node[above] {} (0.15,0.55);
\draw (0,0.3) edge [out=0,in=-75] node[above] {} (0.15,0.55);
\end{tikzpicture}, \begin{tikzpicture}[thick,scale=1.2]
\draw (0,0) -- (0,0.3);
\draw (0,0.3) -- (.15,0.15);
\draw (0,0) -- (.15,0.15);
\draw (0,0.3) edge [out=30,in=-30] node[above] {} (0,.55);
\draw (0,0.3) edge [out=150,in=210] node[above] {} (0,.55);
\end{tikzpicture}. The first one arises from \begin{tikzpicture}[thick,scale=1.2]
\filldraw (.1,0.15)circle (1pt);
\filldraw (-.1,0.15)circle (1pt);
\draw (0,0) -- (.1,0.15);
\draw (0,0) -- (-.1,0.15);
\end{tikzpicture}; the second and the third ones from \begin{tikzpicture}[thick,scale=1.2]
\filldraw (.1,0.35)circle (1pt);
\filldraw (-.1,0.35)circle (1pt);
\filldraw (.1,0.15)circle (1pt);
\filldraw (-.1,0.15)circle (1pt);
\draw (0,0) -- (0,0.2);
\draw (0,0.2) -- (.1,0.35);
\draw (0,0.2) -- (-.1,0.35);
\draw (0,0) -- (.1,0.15);
\draw (0,0) -- (-.1,0.15);
\end{tikzpicture} whereas the last one from \begin{tikzpicture}[thick,scale=1.2]
\filldraw (.1,0.35)circle (1pt);
\filldraw (-.1,0.35)circle (1pt);
\filldraw (.1,0.15)circle (1pt);
\filldraw (0,0.45)circle (1pt);
\draw (0,0) -- (0,0.2);
\draw (0,0.2) -- (.1,0.35);
\draw (0,0.2) -- (-.1,0.35);
\draw (0,0) -- (.1,0.15);
\draw (0,0.2) -- (0,0.45);
\end{tikzpicture}.\\
On the other hand, examples of non-admissible graphs are 
\begin{tikzpicture}[thick,scale=1.2]
\draw (0,0) edge [out=90,in=185] node[above] {} (0.15,0.25);
\draw (0,0) edge [out=0,in=-75] node[above] {} (0.15,0.25);
\draw (0,0) edge [out=180,in=255] node[above] {} (-0.15,0.25);
\draw (0,0) edge [out=90,in=5] node[above] {} (-0.15,0.25);
\end{tikzpicture} and \begin{tikzpicture}[thick,scale=1.2]
\draw (0,0) -- (0,0.3);
\draw (0,0.3) -- (.15,0.15);
\draw (0,0.3) -- (-.15,0.15);
\draw (0,0) -- (.15,0.15);
\draw (0,0) -- (-.15,0.15);
\draw (0,0.3) edge [out=30,in=-30] node[above] {} (0,.55);
\draw (0,0.3) edge [out=150,in=210] node[above] {} (0,.55);
\end{tikzpicture}.
\end{example}

\begin{remark}
As one can infer from the previous examples, to prove that only a finite number of \emph{admissible graphs} needs to be renormalized, we focus our attention only on admissible graphs since any tree in $T$ can be constructed as a finite linear combination of admissible graphs.
\end{remark}

To prove the finiteness of the number of graphs which require to be renormalized, we need an estimate of the form $L\leq pN$, for some $p>0$. This tells us how many lines one can expect, at most, for an admissible graph with a given number $N$ of vertices.

Before obtaining such estimate we stress that 
\begin{itemize}
\item all graphs considered are \emph{connected};
\item with $\mathrm{Val}(v)$ we indicate the {\em valency} of a vertex $v$, namely the number of lines incident to $v$. Recall that $\mathrm{Val}(v)\in\lbrace2,3,4\rbrace$;
\end{itemize}

\begin{remark}\label{Remark: first estimate on p}
Notice that the property of a graph being connected in combination with $\mathrm{Val}(v)\leq4$ implies a first non sharp estimate for the sought after coefficient: $p\leq2$. This can be realized considering the following example:
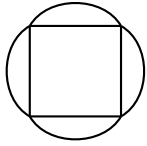
\begin{center}
\begin{tikzpicture}[thick,scale=1.2]
\draw (0,0) -- (0,1);
\draw (0,0) -- (1,0);
\draw (0,1) -- (1,1);
\draw (1,0) -- (1,1);
\draw (0,0) edge [out=150,in=210] node[above] {} (0,1);
\draw (0,0) edge [out=-60,in=240] node[above] {} (1,0);
\draw (0,1) edge [out=60,in=120] node[above] {} (1,1);
\draw (1,0) edge [out=30,in=-30] node[above] {} (1,1);
\end{tikzpicture}
\end{center}
Regardless of its actual realization in the perturbative analysis of Equation \eqref{Eq: algebraic SPDE}, this is the graph with $N=4$ and with the greatest possible number of edges, since the valency of each vertex can be at most $4$.
In this example we considered $N=4$, but the result holds true for any $N\in\mathbb{N}$. In other words the upper bound $\mathrm{Val}(v)\leq4$ implies $L\leq2N$.
\end{remark}

The estimate on $p$ can be improved by exploiting that graphs as those in Remark \ref{Remark: first estimate on p} are actually not admissible. As a matter of fact each graph must have at least a vertex of valency 2, coming from the contraction of two leaves into a vertex.

\begin{lemma}\label{Lemma: valency 2}
Let $\mathcal{G}=(L,N)$ be an admissible graph. Then it has at least $\lceil\frac{N}{3}\rceil$ vertices of valency 2, where $\lceil\frac{N}{3}\rceil$ denotes the smallest integer bigger than $\frac{N}{3}$.
\begin{proof}
We consider the following family of trees, which can be constructed by iteration. The first ones are
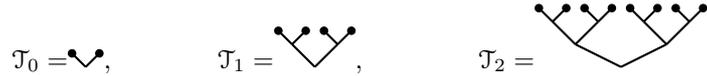
\begin{center}
$\mathcal{T}_0=$\begin{tikzpicture}[thick,scale=1.2]
\filldraw (.15,.15)circle (1pt);
\filldraw (-.15,.15)circle (1pt);
\draw (0,0) -- (.15,.15);
\draw (0,0) -- (-.15,.15);
\end{tikzpicture},\qquad\qquad
$\mathcal{T}_1=$\begin{tikzpicture}[thick,scale=1.2]
\filldraw (.40,.40)circle (1pt);
\filldraw (.10,.40)circle (1pt);
\filldraw (-.40,.40)circle (1pt);
\filldraw (-.10,.40)circle (1pt);
\draw (0,0) -- (.25,.25);
\draw (0,0) -- (-.25,.25);
\draw (.25,.25) -- (.40,.40);
\draw (.25,.25) -- (.10,.40);
\draw (-.25,.25) -- (-.40,.40);
\draw (-.25,.25) -- (-.10,.40);
\end{tikzpicture}, \qquad\qquad$\mathcal{T}_2=$\begin{tikzpicture}[thick,scale=1.2]
\filldraw (.90,.65)circle (1pt);
\filldraw (.6,.65)circle (1pt);
\filldraw (-.90,.65)circle (1pt);
\filldraw (-.6,.65)circle (1pt);
\filldraw (.40,.65)circle (1pt);
\filldraw  (.1,.65)circle (1pt);
\filldraw (-.40,.65)circle (1pt);
\filldraw (-.10,.65)circle (1pt);
\draw (0,0) -- (.5,.25);
\draw (0,0) -- (-.5,.25);
\draw (.5,.25) -- (.75,.50);
\draw (.5,.25) -- (.25,.50);
\draw (-.5,.25) -- (-.75,.50);
\draw (-.5,.25) -- (-.25,.50);
\draw (.75,.50) -- (.90,.65);
\draw (.75,.50) -- (.6,.65);
\draw (-.75,.5) -- (-.90,.65);
\draw (-.75,.5) -- (-.6,.65);
\draw (.25,.5) -- (.40,.65);
\draw (.25,.5) -- (.1,.65);
\draw (-.25,.5) -- (-.40,.65);
\draw (-.25,.5) -- (-.10,.65);
\end{tikzpicture},
\end{center}
and so on. Inductively, this can be realized as $\mathcal{T}_0=\Phi^2$ and $\mathcal{T}_{n+1}=[P_\chi\circledast\mathcal{T}_{n}]^2$ showing that $\mathcal{T}_n$ is admissible for all $n\in\mathbb{N}_0$.  

Once we collapse two by two all the leaves in vertices in any possible way, this is the family of trees contained in $T$ having the smallest possible number of vertices of valency two.  %Hence, we conclude the proof of this lemma if we prove that this family of graphs, letting $n$ vary, have at least a fraction $\frac{1}{3}$ of vertices of valency 2.
The ratio $r_n$ between the number of vertices of valency 2 and the number $N$ of vertices, for the $n$-th iteration, is 
\begin{align*}
r_n=\frac{2^n+1}{{2^n}+\sum_{j=0}^n2^j}=\frac{2^n+1}{2^n3-1}.
\end{align*}
This is a decreasing function of $n$, whose limit for $n\to\infty$ is $r=\frac{1}{3}$. As a consequence, $1/3$ is a lower bound for the fraction of vertices of valency 2 for any admissible graph.
\end{proof}
\end{lemma}

In order to further improve our estimate of $p$, we can also compute an upper bound for the fraction of vertices of valency 4 for any admissible graph.

\begin{lemma}\label{Lemma: valency 4}
Let $\mathcal{G}=(L,N)$ be an admissible graph. Then, at most half of its vertices are of valency 4.
\begin{proof}
We adopt a strategy similar to the one used in Lemma \ref{Lemma: valency 2}. Hence, we introduce a family of trees which maximizes the fraction of vertices of valency 4.
We consider again a sequence of admissible graphs of the following form: $\mathcal{T}_0=P_\chi\circledast\Phi^3$ and $\mathcal{T}_{n+1}=P_\chi\circledast(\mathcal{T}_n)^3$.
From a graphical point of view, this translates to
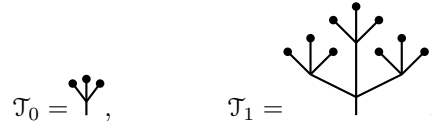
\begin{center}
$\mathcal{T}_0=$\begin{tikzpicture}[thick,scale=1.2]
\filldraw (.15,.35)circle (1pt);
\filldraw (-.15,.35)circle (1pt);
\filldraw  (0,.4)circle(1pt);
\draw (0,0) -- (0,.15);
\draw (0,.15) -- (.15,.35);
\draw (0,.15) -- (-.15,.35);
\draw (0,.15) -- (0,.4);
\end{tikzpicture}, \qquad\qquad
$\mathcal{T}_1=$\begin{tikzpicture}[thick,scale=1.2]
\filldraw (.75,.5)circle (1pt);
\filldraw (.25,.5)circle (1pt);
\filldraw (-.75,.50)circle (1pt);
\filldraw (-.25,.50)circle (1pt);
\filldraw (.25,.85)circle (1pt);
\filldraw (-.25,.85)circle (1pt);
\filldraw (.5,.65)circle (1pt);
\filldraw (-.5,.65)circle (1pt);
\filldraw (0,1)circle (1pt);
\draw (0,-.2) -- (0,0);
\draw (0,0) -- (.5,.25);
\draw (0,0) -- (-.5,.25);
\draw (0,0) -- (0,.6);
\draw (0,.6) -- (.25,.85);
\draw (0,.6) -- (0,1);
\draw (0,.6) -- (-.25,.85);
\draw (.5,.25) -- (.75,.50);
\draw (.5,.25) -- (.5,.65);
\draw (.5,.25) -- (.25,.50);
\draw (-.5,.25) -- (-.75,.50);
\draw (-.5,.25) -- (-.5,.65);
\draw (-.5,.25) -- (-.25,.50);
\end{tikzpicture},
\end{center}
and so on.

Notice that, for all $n$, these trees have an odd number of leaves. As a consequence, in order to maximize the vertices of valency 4, instead of cutting a branch we add a branch at the lowest vertex. In other words, we consider $\mathcal{Q}_n=\Phi\mathcal{T}_n$. At a graphical level this amounts to
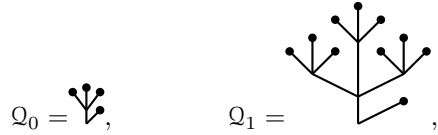
\begin{center}
$\mathcal{Q}_0=$\begin{tikzpicture}[thick,scale=1.2]
\filldraw (.15,.35)circle (1pt);
\filldraw (-.15,.35)circle (1pt);
\filldraw  (0,.4)circle(1pt);
\filldraw  (.15,.15)circle(1pt);
\draw (0,0) -- (.15,.15);
\draw (0,0) -- (0,.15);
\draw (0,.15) -- (.15,.35);
\draw (0,.15) -- (-.15,.35);
\draw (0,.15) -- (0,.4);
\end{tikzpicture}, \qquad\qquad$\mathcal{Q}_1=$\begin{tikzpicture}[thick,scale=1.2]
\filldraw (.75,.5)circle (1pt);
\filldraw (.25,.5)circle (1pt);
\filldraw (-.75,.50)circle (1pt);
\filldraw (-.25,.50)circle (1pt);
\filldraw (.25,.85)circle (1pt);
\filldraw (-.25,.85)circle (1pt);
\filldraw (.5,.65)circle (1pt);
\filldraw (-.5,.65)circle (1pt);
\filldraw (0,1)circle (1pt);
\filldraw (.5,-.05)circle (1pt);
\draw (0,-.3) -- (.5,-.05);
\draw (0,-.3) -- (0,0);
\draw (0,0) -- (.5,.25);
\draw (0,0) -- (-.5,.25);
\draw (0,0) -- (0,.6);
\draw (0,.6) -- (.25,.85);
\draw (0,.6) -- (0,1);
\draw (0,.6) -- (-.25,.85);
\draw (.5,.25) -- (.75,.50);
\draw (.5,.25) -- (.5,.65);
\draw (.5,.25) -- (.25,.50);
\draw (-.5,.25) -- (-.75,.50);
\draw (-.5,.25) -- (-.5,.65);
\draw (-.5,.25) -- (-.25,.50);
\end{tikzpicture},
\end{center}
and so on and so forth. This family of trees, once we collapse the leaves into vertices, maximizes the fraction of vertices of valency 4. We compute such a fraction as a function of $n$. The number of vertices of valency 4 in $\mathcal{Q}_n$ is $\sum_{j=0}^n3^j$, whereas the total number of vertices in $\mathcal{Q}_n$, after the collapse of the leaves into vertices, is $3+\sum_{j=0}^n3^{j+1}-\frac{1}{2}(3^{n+1}+1)$. The ratio $r_n$ is
\begin{align*}
r_n=\frac{\sum_{j=0}^n3^j}{3+\sum_{j=0}^n3^{j+1}-\frac{1}{2}(3^{n+1}+1)}=\frac{3^{n+1}-1}{3^{n+1}2+2}.
\end{align*}
This is an increasing function of $n$ and its limit as $n\to\infty$ is $r=\frac{1}{2}$. As a consequence, $1/2$ is an upper bound for the fraction of vertices of valency 4 for any admissible graph.
\end{proof}
\end{lemma}

\begin{proposition}\label{Proposition: bound on p}
With the notation above, $p\leq\frac{19}{12}$, \emph{i.e.}, $L\leqslant\frac{19}{12}N$.
\end{proposition}
\begin{proof}
The proof is based on the previous results, in particular Lemma \ref{Lemma: valency 2} and Lemma \ref{Lemma: valency 4}. Consider a graph with $N$ vertices. In order to have all of them of valency at least 2, we need $N$ lines. Then, on account of Lemma \ref{Lemma: valency 4}, we can add at most $N/2$ lines in order to get the upper bound of half of the vertices having valency 4. In addition, on account of Lemma \ref{Lemma: valency 2}, the vertices of valency 2 must be at least $N/3$. As a consequence, in order to maximize the number of lines, we can still add $\frac{1}{2}(\frac{N}{2}-\frac{N}{3})=\frac{N}{12}$ lines, where the factor $\frac{1}{2}$ comes from the fact that each line accounts for two vertices. Summarizing, we started with $N$ lines, adding additional $N/2$ lines in order to saturate the vertices of valency four and $N/12$ lines to account that, in the worst case scenario, we can only have $N/3$ vertices of valency two. In other words
\begin{align*}
L\leq N+\frac{N}{2}+\frac{N}{12}=\frac{19}{12}N.
\end{align*}
%See Remark \ref{Remark: graphical proof} to see a graphical example of this proof.
\end{proof}

\begin{remark}\label{Remark: graphical proof}
In order to better grasp the idea underneath Proposition \ref{Proposition: bound on p}, we discuss an explicit example. Consider the case $N=9$. The first step consists of adding $9$ lines so that all the vertices are of valency at least 2. This procedure brings no arbitrariness.
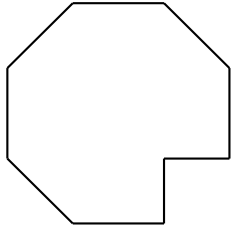
\begin{center}
\begin{tikzpicture}[thick,scale=1.2]
\draw (0,0) -- (1,0);
\draw (1,0) -- (1,.717);
\draw (1,.717) -- (1.717,.717);
\draw (1.717,.717) -- (1.717,1.717);
\draw (1.717,1.717)-- (1,2.434);
\draw (1,2.434) -- (0,2.434);
\draw (0,2.434) -- (-.717,1.717);
\draw (-.717,1.717) -- (-.717,.717);
\draw (-.717,.717) -- (0,0);
\end{tikzpicture}
\end{center}
Next, we add $4$ lines in order to account for worst case scenario, namely $\lfloor N/2\rfloor=4$ vertices of valency 4:
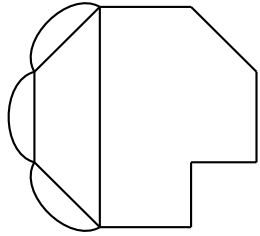
\begin{center}
\begin{tikzpicture}[thick,scale=1.2]
\draw (0,0) -- (1,0);
\draw (1,0) -- (1,.717);
\draw (1,.717) -- (1.717,.717);
\draw (1.717,.717) -- (1.717,1.717);
\draw (1.717,1.717)-- (1,2.434);
\draw (1,2.434) -- (0,2.434);
\draw (0,2.434) -- (-.717,1.717);
\draw (-.717,1.717) -- (-.717,.717);
\draw (-.717,.717) -- (0,0);
\draw (0,0) edge [out=210,in=240] node[above] {} (-.717,.717);
\draw (-.717,.717) edge [out=165,in=195] node[above] {} (-.717,1.717);
\draw (-.717,1.717) edge [out=120,in=150] node[above] {} (0,2.434);
\draw (0,0) -- (0,2.434);
\end{tikzpicture}
\end{center}
The next step consists of implementing Lemma \ref{Lemma: valency 2}. This entails that we must have at least $\lceil N/3\rceil=3$ vertices of valency 2. As a consequence, since in the last graph we have $5$ vertices of valency 2, we can add one line, obtaining at the end of the day
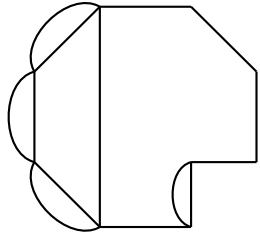
\begin{center}
\begin{tikzpicture}[thick,scale=1.2]
\draw (0,0) -- (1,0);
\draw (1,0) -- (1,.717);
\draw (1,.717) -- (1.717,.717);
\draw (1.717,.717) -- (1.717,1.717);
\draw (1.717,1.717)-- (1,2.434);
\draw (1,2.434) -- (0,2.434);
\draw (0,2.434) -- (-.717,1.717);
\draw (-.717,1.717) -- (-.717,.717);
\draw (-.717,.717) -- (0,0);
\draw (0,0) edge [out=210,in=240] node[above] {} (-.717,.717);
\draw (-.717,.717) edge [out=165,in=195] node[above] {} (-.717,1.717);
\draw (-.717,1.717) edge [out=120,in=150] node[above] {} (0,2.434);
\draw (0,0) -- (0,2.434);
\draw (1,0) edge [out=165,in=195] node[above] {} (1,.717);
\end{tikzpicture}
\end{center}
Observe that the outcome abides to the conclusions of Proposition \ref{Proposition: bound on p}, since here $L=14<\frac{19}{12}N=\frac{57}{4}$. 
\end{remark}

\begin{theorem}\label{Theorem: finite number of graphs}
If $d\leq 3$, only a finite number of admissible graph needs to be renormalized.
\begin{proof}
On account of the previous results, given an admissible graph $\mathcal{G}=(L,N)$, and recalling both Equation \eqref{Equation: degree of divergence of (L,V)} and Proposition \ref{Proposition: bound on p}, the associated degree of divergence is
\begin{align*}
\rho(\mathcal{G})=Ld-(N-1)(d+2)\leq\frac{19}{12}Nd-(N-1)(d+2)=N\bigg(\frac{7}{12}d-2\bigg)+d+2.
\end{align*}
It turns out that, only if $d\leq 3$, hence only in the subcritical regime, $\rho(\mathcal{G})$ becomes negative increasing sufficiently $N$.
Thus only a finite number of admissible graphs have a non negative degree of divergence and needing possibly to be renormalized.
\end{proof}
\end{theorem}

\begin{remark}\label{Rem: not sharp but enough}
This above proof does not give a sharp estimate on the number of diagrams which needs to be renormalized. It is actually an overestimate since only if a graph lies in $T$, it needs to be accounted for. In the above analysis we discarded this fact and we considered a larger class of graphs which is not necessarily subordinated to Equation \eqref{Eq: algebraic SPDE}, though all those lying in $T$ are included.
\end{remark}

\appendix
\section{Wavefront Set: basic definitions and results}\label{App: Wave Front Set kurzgesagt}

In this appendix we recollect some basic concepts and results proper of microlocal analysis, which play a key r\^ole in the paper. We base our summary on \cite[Ch. 8]{Hormander-I-03} and we give no proof of the various statements since they are standard and they can be found in the mentioned reference. Our goal is only to facilitate a smooth reading of the main body of this work in which we often make use of several microlocal techniques, without the need of constantly consulting the literature. See also the review \cite{BDH14}.

%\begin{definition}\label{Def: Fourier transform}
%	Let $U\subseteq\mathbb{R}^n$ $t\in\mathcal{E}'(U)$ be a distribution of compact support.
%	The Fourier transform of $t$ is the smooth function $\widehat{t}\in\mathcal{E}(\mathbb{R}^n)$ defined by
%	\begin{align}\label{Eq: Fourier transform of distribution}
%		\widehat{t}(\xi):=t(e_\xi)\,,
%	\end{align}
%	being $e_\xi(x):=e^{ix\xi}$.
%\end{definition}
\begin{definition}\label{Def: Wave Front Set}
	Let $U\subseteq\mathbb{R}^n$ and $t\in\mathcal{D}'(U)$.
	The wave front set of $t$ is the subset $\operatorname{WF}(t)\subseteq T^*U\setminus\{0\}$ such that $(x_0,\xi_0)\in T^*U\setminus\{0\}$ is \textit{not} in $\operatorname{WF}(t)$ if and only if there exists:
	\begin{enumerate}[(i)]
		\item
		an open neighbourhood $U_{x_0}\subseteq U$ containing $x_0$,
		\item
		an open conic neighbourhood $V_{\xi_0}\subseteq\mathbb{R}^n$ of $\xi_0\in\mathbb{R}^n$,
		\item
		a test function $f\in\mathcal{D}(U_{x_0})$ with $f(x_0)\neq 0$,
	\end{enumerate}	
	such that
	\begin{align*}
		\sup_{\xi\in V_{\xi_0}}|\xi|^k\,|\mathcal{F}(ft)(\xi)|<+\infty\,,\qquad \forall \, k\in\mathbb{Z}_+\,,
	\end{align*}
	where $\mathcal{F}$ denotes the Fourier transform.
\end{definition}

\begin{remark}
	Observe that, in view of Definition \ref{Def: Wave Front Set}, $\operatorname{WF}(t)$, $t\in\mathcal{D}^\prime(U)$, is a closed subset of $T^*U\setminus\{0\}$ which is diffeomorphism invariant. In view of this last property Definition \ref{Def: Wave Front Set} can be naturally generalized to elements of $\mathcal{D}^\prime(M)$, where $M$ is a smooth manifold.
\end{remark}

\begin{example}\label{Ex: WF of Dirac delta on total diagonal}
	Let $\mathrm{Diag}_n\subseteq\mathbb{R}^n$ the total diagonal of $\mathbb{R}^n$ -- \textit{i.e.} $(x_1,\ldots,x_n)\in\mathrm{Diag}_n$ if and only if $x_1=\ldots=x_n$.
	We denote by $\delta_{\mathrm{Diag}_n}\in\mathcal{D}'(\mathbb{R}^n)$ the Dirac delta distribution centred at $\operatorname{Diag}_n$, that is,
	\begin{align*}
		\delta_{\mathrm{Diag}_n}(f)
		:=\int_{\mathbb{R}}f(x,\ldots,x)\mathrm{d}x
		\qquad\forall\, f\in\mathcal{D}(\mathbb{R}^n)\,.
	\end{align*}
	A direct inspection entails that
	\begin{align}\label{Eq: WF of Dirac delta on total diagonal}
		\operatorname{WF}(\delta_{\mathrm{Diag}_n})
		=\{(x_1,\ldots,x_n,\xi_1,\ldots,\xi_n)\in T^*\mathbb{R}^n\setminus\{0\}\,|\,x_1=\ldots=x_n\,,\,\sum_{\ell=1}^n\xi_\ell=0
		\}\,.
	\end{align}
\end{example}

\noindent In the following theorem we list some of the main properties of the wavefront set of a distribution:

\begin{theorem}\label{Thm: WF results}
	Let $U\subseteq\mathbb{R}^n$ be an open subset. The following statements hold true:
	\begin{enumerate}
		\item
		If $D$ is a differential operator and $t\in\mathcal{D}'(U)$,
		\begin{align}\label{Eq: WF bound for derivatives}
			\operatorname{WF}(Dt)
			\subseteq\operatorname{WF}(t)
			\subseteq\operatorname{WF}(Dt)
			\cup\operatorname{Char}(D)\,,
		\end{align}
		where $\operatorname{Char}(D)$ denotes the characteristic set of $D$, namely
		\begin{align*}
			\operatorname{Char}(D)
			:=\{(x,\xi)\in T^*U\setminus\{0\}\,|\,\sigma_D(x,\xi)=0\}\,,
		\end{align*}
		$\sigma_D$ being the principal symbol of $D$.
		\item - \cite[Thm.8.2.10]{Hormander-I-03}:
		Let $t_1,t_2\in\mathcal{D}'(U)$.
		If
		\begin{align}\label{Eq: WF product condition}
			(x,0)\notin\{(x,\xi_1+\xi_2)\in T^*U\;|\;(x,\xi_1)\in\operatorname{WF}(t_1)\,,\,(x,\xi_2)\in\operatorname{WF}(t_2)\}\,,
		\end{align}
		then the product $t_1\cdot t_2:=(t_1\otimes t_2)\cdot\delta_{\mathrm{Diag}_2}\in\mathcal{D}'(U)$ is well-defined and
		\begin{align}\label{Eq: WF product bound}
			\operatorname{WF}(t_1\cdot t_2)
			\subseteq\{(x,\xi_1+\xi_2)\in T^*U\;|\;(x,\xi_1)\in\operatorname{WF}(t_1)\,,\,(x,\xi_2)\in\operatorname{WF}(t_2)\}\,.
		\end{align}
%		(Notice that the last set is contained in $T^*U\setminus\{0\}$ on account of equation \eqref{Eq: WF product condition}.)
		\item - \cite[Thm. 8.2.12-13]{Hormander-I-03}.
		Let $V\subseteq\mathbb{R}^k$ be an open subset, let $K\in\mathcal{D}^\prime(U\times V)$, $t\in\mathcal{E}^\prime(V)$ and set
		\begin{align}\label{Eq: WF 2-component projection}
			\operatorname{WF}^\prime_2(K)
			:=\{(x_2,\xi_2)\in T^*V\,|\,
			\exists x_1\in U\,,\,(x_1,x_2,0,-\xi_2)\in\operatorname{WF}(K)\}\,.
		\end{align}
		It holds that, if
		\begin{align}\label{Eq: WF convolution condition}
			\operatorname{WF}^\prime_2(K)\cap\operatorname{WF}(t)=\emptyset\,,
		\end{align}
		then $K\circledast t\in{D}^\prime(U)$ where
		\begin{align}\label{Eq: convolution definition}
			[K\circledast t](f)
			:=K(f\otimes t)
			:=[K\cdot(1_n\otimes t)](f\otimes 1_k)\,,\qquad\forall\, f\in\mathcal{D}(U)\,.
		\end{align}
%		\footnote{
%			Here $K\cdot(1_n\otimes t)$ denotes the product between the distribution $K$ of integral kernel $K(\widehat{y}_n,\widehat{z}_k)$ and the distribution $1_n\otimes t$ with integral kernel $1_n(\widehat{y}_n)t(\widehat{z}_k)$.
% 		}
		Furthermore 
		\begin{align}\label{Eq: WF convolution bound}
			\operatorname{WF}(K\circledast t)
			\subseteq\operatorname{WF}_1(K)
			\cup\operatorname{WF}'(K)\circ\operatorname{WF}(t)\,,
		\end{align}
		where
		\begin{align}\label{Eq: WF 1-component projection}
			\operatorname{WF}_1(K)
			&:=\{(x_1,\xi_1)\in T^*U\setminus\{0\}\;|\;
			\exists x_2\in V\,,\,(x_1,x_2,\xi_1,0)\in\operatorname{WF}(K)
			\}\,,\\
			\label{Eq: Primed WF}
			\operatorname{WF}^\prime(K)&:=\{(x_1,x_2,\xi_1,\xi_2)\in T^*(U\times U)\setminus\{0\}\;|\;(x_1,x_2,\xi_1,-\xi_2)\in\operatorname{WF}(K)\}\\
			\label{Eq: WF composition}
			\operatorname{WF}^\prime(K)\circ\operatorname{WF}(t)
			&:=\{
			(x_1,\xi_1)\in T^*U\setminus\{0\}\;|\;
			\exists (x_2,\xi_2)\in\operatorname{WF}(t)\,,\,(x_1,x_2,\xi_1,-\xi_2)\in\operatorname{WF}(K)
			\}\,.
		\end{align}
		\item
		- \cite[Thm. 8.2.14]{Hormander-I-03}:
		Let $V\subseteq\mathbb{R}^k$ be an open subset and let $K_1\in\mathcal{D}'(U\times V)$ and $K_2\in\mathcal{D}'(V\times U)$ be such that
		\begin{align}\label{Eq: WF composition condition}
			\operatorname{WF}^\prime_2(K_1)\cap\operatorname{WF}_1(K_2)=\emptyset\,.
		\end{align}
		In addition assume that $\operatorname{pr}_1(\operatorname{supp}(K_2))\cap\operatorname{pr}_2(\operatorname{supp}(K_1))$ is compact, where $\operatorname{pr}_1\colon V\times U\to V$, $\operatorname{pr}_2\colon U\times V\to V$ are the canonical projections.		
		Then $K_1\circ K_2\in\mathcal{D}^\prime(U\times U)$ is completely defined by
		\begin{align}\label{Eq: distribution composition}
			(K_1\circ K_2)(f_1\otimes f_2)
			:=[(K_1\otimes K_2)\cdot(1_n\otimes\delta_{\mathrm{Diag}_2}\otimes 1_n)](f_1\otimes 1_{2k}\otimes f_2)\,\quad\forall\, f_1,f_2\in\mathcal{D}(U)\,,
		\end{align}
		where $\delta_{\mathrm{Diag}_2}\in\mathcal{D}'(\mathbb{R}^n\times\mathbb{R}^n)$ is the Dirac delta distribution centred on the diagonal $\mathrm{Diag}_2\subset\mathbb{R}^n\times\mathbb{R}^n$ of $\mathbb{R}^n\times\mathbb{R}^n$ -- \textit{i.e.} $\mathrm{Diag}_2:=\{(x,x)\in\mathbb{R}^n\times\mathbb{R}^n\,|\,x\in\mathbb{R}^n\}$.
		In addition it holds that
		\begin{align}\label{Eq: WF composition bound}
			\operatorname{WF}'(K_1\circ K_2)
			\subseteq\operatorname{WF}'(K_1)\circ\operatorname{WF}'(K_2)
			\cup(\operatorname{WF}_1(K_1)\times U\times\{0\})
			\cup(U\times\{0\}\times\operatorname{WF}_2(K_2))\,,
		\end{align}
		where
		\begin{align*}
			\operatorname{WF}^\prime(K)
			:=\{&(x_1,x_2,\xi_1,\xi_2)\in T^*(U\times U)\setminus\{0\}\;|\;(x_1,x_2,\xi_1,-\xi_2)\in\operatorname{WF}(K)\}\,,\\
			\operatorname{WF}'(K_1)\circ\operatorname{WF}^\prime(K_2)
			:=\{&(x_1,x_2,\xi_1,\xi_2)\in T^*(U\times U)\setminus\{0\}\;|\;\exists(z_3,\zeta_3)\in T^*V,\,\\
			&(x_1,z_3,\xi_1,\zeta_3)\in\operatorname{WF}^\prime(K_1)\,,
			(z_3,x_2,\zeta_3,\xi_2)\in\operatorname{WF}^\prime(K_2)
			\}\,.
		\end{align*}
		\item
		- \cite[Thm. 8.2.9]{Hormander-I-03}:
		Let $V\subseteq\mathbb{R}^k$ be an open subset and let $t_1\in\mathcal{D}^\prime(U)$ and $t_2\in\mathcal{D}^\prime(V)$, then
		\begin{align}\label{Eq: WF of the tensor product}
		\operatorname{WF}(t_1\otimes t_2)\subset\big(\operatorname{WF}(t_1)\times\operatorname{WF}(t_2)\big)\cup\big((\operatorname{supp}(t_1)\times\{0\})\times\operatorname{WF}(t_2)\big)\cup\big(\operatorname{WF}(t_1)\times(\operatorname{supp}(t_2)\times\{0\})\big)
		\end{align}
	\end{enumerate}
\end{theorem}

\begin{example}\label{Ex: WF of parametrix}
	As a useful application of Equation \eqref{Eq: WF bound for derivatives} we compute the wave front set of a parametrix $P$ of an elliptic operator $D$. In other words we consider $P\in\mathcal{D}'(U\times U)$ which satisfies
	\begin{align*}
		P\circledast D=DP\circledast=\delta_{\mathrm{Diag}_2}\quad\mod \mathcal{E}(U)\,,
	\end{align*}
	that is, $P\circledast Df-f,DP\circledast f-f\in \mathcal{E}(U)$.
	Since $D$ is assumed to be elliptic, the principal symbol $\sigma_D$ is nowhere vanishing. As a consequence $\operatorname{Char}(D)=\emptyset$. Equation \eqref{Eq: WF bound for derivatives} entails
	\begin{align}\label{Eq: WF of parametrix}
		\operatorname{WF}(P)
		=\operatorname{WF}(DP)
		=\operatorname{WF}(\delta_{\mathrm{Diag}_2})\,.
	\end{align}
\end{example}

\begin{remark}\label{Rmk: convolution with smooth function}
	As a particular case of Equation \eqref{Eq: WF convolution bound} consider $K\in\mathcal{D}'(U\times V)$ and $t\in\mathcal{D}(V)$.
	In this case $\operatorname{WF}(t)=\emptyset$ and therefore the condition stated in Equation \eqref{Eq: WF convolution condition} is fulfilled and $K\circledast t\in\mathcal{D}'(U)$ is well-defined.
	Moreover equation \eqref{Eq: WF convolution bound} entails that
	\begin{align}
		\operatorname{WF}(K\circledast t)
		\subseteq\operatorname{WF}_1(K)\,.
	\end{align}
\end{remark}

\section{Scaling degree}\label{App: Scaling degree}

In this appendix we recall the key tools and results related to the theory of the scaling degree of distributions.
These play a pivotal r\^ole in the main body of the section and, similarly to Appendix \ref{App: Wave Front Set kurzgesagt}, our goal is to allow a smooth reading of this work without the need of consulting continuously the literature.
In the remainder of this section we refer mainly to \cite{Brunetti-Fredenhagen-00}, in which all proofs of the following statements are present. For an extension of the range of applicability of these results, we refer also to \cite{Dang-16}.

\begin{definition}\label{Def: scaling degree}
	Let $U\subseteq \mathbb{R}^d$ be a conic open subset of $\mathbb{R}^d$ and, for any $x_0\in \mathbb{R}^d$, consider $U_{x_0}:=U+x_0$.
	For $f\in\mathcal{D}(U)$ and $\lambda>0$ let $f^\lambda_{x_0}:=\lambda^{-d}f(\lambda^{-1}(x-x_0))\in\mathcal{D}(U_{x_0})$.
	Similarly, for $t\in\mathcal{D}'(U_{x_0})$ denote as $t^\lambda_{x_0}\in\mathcal{D}'(U)$ the distribution  $t^\lambda_{x_0}(f):=t(f^\lambda_{x_0})$ for all $f\in\mathcal{D}(U)$. The scaling degree of $t$ at $x_0$ is 
	\begin{align}\label{Eq: scaling degree at x0}
	\operatorname{sd}_{x_0}(t)
	:=\inf\bigg\lbrace\omega\in\mathbb{R}\,|\,\lim_{\lambda\to 0^+}\lambda^\omega t^\lambda_{x_0}=0\bigg\rbrace\,.
	\end{align}
\end{definition}

\begin{example}\label{Ex: sd of delta on a point}
	We shall consider $\delta_{x}\in\mathcal{D}'(\mathbb{R}^d)$.
	A direct computation shows that $\delta_x^\lambda=\lambda^{-d}\delta_x$ so that $\operatorname{sd}_x(\delta_x)=d$.
\end{example}

\begin{remark}\label{Rmk: sd for submanifolds}
	Definition \ref{Def: scaling degree} can be generalized so to encompass the notion of scaling degree with respect to an embedded submanifold $N\subseteq M$ -- \textit{cf.} \cite{Brunetti-Fredenhagen-00}.
	For our purposes it suffices to consider only the case $N=\mathrm{Diag}_n\subseteq M^n$.
	Hence, let $g$ be an arbitrary Riemanniann metric on $M$ and let $U$ be a star-shaped neighbourhood of the zero section of $T_NM:=TM|_N$ and let $\alpha\colon U\to N\times M$ be the smooth map
	\begin{align*}
	\alpha(x,\xi):=(x,\exp_x(\xi))\qquad\forall (x,\xi)\in U\,,
	\end{align*}
	where $\exp_x$ is the exponential map of $M$ centered at $x$.
	Notice that $\alpha(x,0)=(x,x)$ so that $\operatorname{pr}_2(\alpha(U))$ is an open neighbourhood of $N$ in $M$, where $\operatorname{pr}_2\colon N\times M\to M$.
	Let $t\in\mathcal{D}'(\operatorname{pr}_2(\alpha(U)))$ and let $t^\alpha:=(1\otimes t)\circ\alpha_*\in\mathcal{D}'(U)$, where $\alpha_\ast\colon\mathcal{D}(U)\to\mathcal{D}(\alpha(U))$ is defined by $(\alpha_\ast f)(z):=f(\alpha^{-1}(z))$. Considering a reference top density $\mu_U$ on $TU$, if $t$ is generated by a smooth function we have
	\begin{align*}
	t^\alpha(f)=
	\int_{U} t(\exp_x\xi)f_{\mu_U}(x,\xi)\,,
	\end{align*}	
	where $f_{\mu_U}:=f\mu_U$. Similarly, for all $0<\lambda\leq 1$, we define $t^\alpha_\lambda\in\mathcal{D}'(U)$ via $t^\alpha_\lambda(f):=t^\alpha(f^\lambda)$ where $f^\lambda(x,\xi):=\lambda^{-d}f(x,\lambda^{-1}\xi)$ for all $f\in\mathcal{D}(U)$. Once more, if $t$ is generated by a smooth function, it holds that
	\begin{align*}
	t^\alpha_\lambda(f)=\int_{U} t(\exp_x\lambda\xi)f_{\mu_U}\,.
	\end{align*}	
	The scaling degree of $t$ with respect to $N$ is 
	\begin{align}
	\operatorname{sd}_N(t):=\inf\bigg\lbrace
	\omega\in\mathbb{R}\,|\,\lim_{\lambda\to 0^+}\lambda^\omega t^\alpha_\lambda=0
	\bigg\rbrace\,.
	\end{align}
%	Roughly speaking $\operatorname{sd}_N(t)$ can be understood as $\inf_{x\in N}\operatorname{sd}_x(t)$.
\end{remark}

\begin{remark}[Parabolic Case]\label{Rmk: weighted scaling degree}
	Whenever we consider a product manifold $M=\mathbb{R}\times\Sigma$ together with a hypoelliptic operator of the form $\partial_t-E$, being $E$ a $t$-independent elliptic operator on $\Sigma$, the notion of scaling degree introduced in Definition \ref{Def: scaling degree} has to be modified.
	In particular one has to rely on the notion of weighted scaling degree, the $t$-coordinate has a different scaling with respect to the other coordinates.
To consider a concrete example, if $M=\mathbb{R}\times\mathbb{R}^{d-1}$, Definition \ref{Def: scaling degree} remains the same expect for the scaling of the test function $f\in\mathcal{D}(\mathbb{R}^d)$ which is replaced with
	\begin{align*}
		f_{t,x}^\lambda(s,y):=\frac{1}{\lambda^{d+1}}f\left(\frac{s-t}{\lambda^2},\frac{z-x}{\lambda}\right)\,.
	\end{align*}
To distinguish this scenario from the standard one, we shall introduce the symbol $wsd$ referring to it as ``weighted scaling degree''. From a practical viewpoint, hopping from the elliptic to the parabolic case amounts to replacing the dimension of the underlying manifold $M$, say $d$, with $d+1$. We shall refer to the latter as ``effective dimension" of $M$. 
\end{remark}

\begin{example}\label{Ex: sd of delta on total diagonal}
	Following the same pattern on Example \ref{Ex: sd of delta on a point}, it holds that $\operatorname{sd}_{\mathrm{Diag}_n}(\delta_{\mathrm{Diag}_{n}})=(n-1)d$. On the other hand, if we consider the parabolic case, $\operatorname{wsd}_{\mathrm{Diag}_n}(\delta_{\mathrm{Diag}_{n}})=(n-1)(d+1)$.
\end{example}

\begin{example}\label{Ex: sd of parametrix}
	For any elliptic operator $E$ on a smooth manifold $M$ of dimension $\dim M=d$ it holds $\operatorname{sd}_{\mathrm{Diag}_2}P=d-2$, $P$ being any parametrix of $E$.
	Similarly, for the case of a hypoelliptic operator of the form $\partial_t-E$ on $M=\mathbb{R}\times\Sigma$ it holds $\operatorname{wsd}_{\mathrm{Diag}_2}P=d-1$, where $P$ is a parametrix of $\partial_t-E$ while $\operatorname{wsd}$ denotes the weighted scaling degree introduced in Remark \ref{Rmk: weighted scaling degree}.
\end{example}

\begin{remark}\label{Rmk: sd product estimate}
	The scaling degree is multiplicative with respect to tensor product, namely \cite[Lem. 5.1]{Brunetti-Fredenhagen-00}
	\begin{align}\label{Eq: sd of tensor product}
		\operatorname{sd}_{(x_1,x_2)}(T_1\otimes T_2)
		=\operatorname{sd}_{x_1}(T_1)
		+\operatorname{sd}_{x_2}(T_2)\,.
	\end{align}
	Moreover \cite[Lem. 6.6]{Brunetti-Fredenhagen-00} entails (in particular) that if $T_1,T_2\in\mathcal{D}'(\mathbb{R}^d)$ satisfy condition \eqref{Eq: WF product condition}, then $T_1T_2\in\mathcal{D}'(\mathbb{R}^d)$ satisfies
	\begin{align}\label{Eq: sd of product}
		\operatorname{sd}_x(T_1T_2)
		\leq\operatorname{sd}_x(T_1)
		+\operatorname{sd}_x(T_2)\,.
	\end{align}
\end{remark}

The following theorem has been proved in \cite[Thm. 5.2-5.3]{Brunetti-Fredenhagen-00} -- see also \cite[Thm 6.9]{Brunetti-Fredenhagen-00} for the corresponding result for submanifolds.

\begin{theorem}\label{Thm: extension with scaling degree}
	Let $x\in\mathbb{R}^d$ and $t\in\mathcal{D}'(\mathbb{R}^d_x)$ where $\mathbb{R}^d_x:=\mathbb{R}^d\setminus\{x\}$ and set $\rho:=\operatorname{sd}_x(t)-d$.
	Then
	\begin{enumerate}
		\item
		if $\rho<0$ then there exists a unique $\widehat{t}\in\mathcal{D}'(\mathbb{R}^d)$ such that $t\subseteq\widehat{t}$ as well as $\operatorname{sd}_x(\widehat{t})=\operatorname{sd}_x(t)$.
		\item
		if $\rho\geq 0$ then all distributions $\widehat{t}\in\mathcal{D}'(\mathbb{R}^d)$ such that $t\subseteq\widehat{t}$ as well as $\operatorname{sd}_x(\widehat{t})=\operatorname{sd}_x(t)$ are of the form
		\begin{align*}
			\widehat{t}
			=t\circ W_\rho
			+\sum_{|\alpha|\leq\rho}a_\alpha\partial^\alpha\delta_x\,,
		\end{align*}
		where $\{a_\alpha\}_\alpha\subset\mathbb{C}$ while $W_\rho\colon\mathcal{D}(\mathbb{R}^d)\to\mathcal{D}(\mathbb{R}^d)$ is defined by
		\begin{align*}
			W_\rho f
			:=f
			-\sum_{|\alpha|\leq \rho}\frac{1}{\alpha!}(\partial^\alpha f)(x)\psi_\alpha\,,
		\end{align*}
		being $\{\psi_\alpha\}_\alpha\subseteq\mathcal{D}(\mathbb{R}^d)$ any family of smooth compactly supported functions such that $\partial^\beta\psi_\alpha(x)=\delta^\beta_\alpha$.
		\item
		if $\rho=+\infty$ then there are no extensions $\widehat{t}\in\mathcal{D}'(\mathbb{R}^d)$ of $t$.
	\end{enumerate}
\end{theorem}

\begin{remark}[Parabolic Case]
In the parabolic case, given $t\in\mathcal{D}^\prime(M)$, its \emph{weighted degree of divergenge} at $x\in M$, $\operatorname{w\rho}(t)$, is defined as $\operatorname{w\rho}(t)=\operatorname{wsd}_x(t)-(d+1)$. 
\end{remark}
\begin{remark}\label{Rmk: extension with sd for submanifolds}
	The results of Theorem \ref{Thm: extension with scaling degree} can be generalized to the problem of extending a given distribution $t\in\mathcal{D}'(M_N)$ -- where $N$ is a smooth submanifold of $M$ and $M_N:=M\setminus N$ -- to a distribution $\widehat{t}\in\mathcal{D}'(M)$ such that $\operatorname{sd}_N(\widehat{t})=\operatorname{sd}_N(t)$ -- \textit{cf.} remark \ref{Rmk: sd for submanifolds}.
	In this latter case the parameter $\rho$ appearing in Theorem \ref{Thm: extension with scaling degree} has to be replaced with
	\begin{align*}
		\rho:=\operatorname{sd}_N(t)-\operatorname{codim}(N)\,,
	\end{align*}
	which, in the particular case of $M=Z^n$, $N=\operatorname{Diag}_n$ reduces to $\operatorname{\operatorname{sd}_{\mathrm{Diag}_n}}(t)-(n-1)\dim Z$.
	For all details we refer to \cite[Thm. 6.9]{Brunetti-Fredenhagen-00}.
\end{remark}

We highlight a few results on the scaling degree of the convolution $K\circledast t$ and of the composition $K_1\circ K_2$ for suitable $K,K_1,K_2\in\mathcal{D}'(\mathbb{R}^d\times \mathbb{R}^d)$ and $t\in\mathcal{D}'(\mathbb{R}^d)$.

\begin{lemma}\label{Lemma: finite sd of convolution}
Let $K\in\mathcal{D}^\prime(\mathbb{R}^d\times\mathbb{R}^d)$ and $t_\ell\in\mathcal{D}^\prime(\underbrace{\mathbb{R}^d\times\ldots\times\mathbb{R}^d}_{\ell-\mathrm{times}})=\mathcal{D}^\prime((\mathbb{R}^d)^\ell)$ be such that
\begin{align*}
\mathrm{WF}_2^\prime(K)\cap\mathrm{WF}_1(t_\ell)=\emptyset\,,\qquad\mathcal{K}\vcentcolon=\operatorname{pr}_2(\mathrm{supp}(K))\cap\operatorname{pr}_1(\mathrm{supp}(t_\ell))\;\mathrm{is\: compact\:in\:\mathbb{R}^d}\,,
\end{align*}
where we employ the same notation of item \textit{4}. of Theorem \ref{Thm: WF results}. Moreover, let $\mathrm{sd}_{\mathrm{Diag}_2}(K)<\infty$, $\mathrm{sd}_{\mathrm{Diag}_{\ell}}(t_\ell)<\infty$ and, defining
\begin{align}\label{Eq: definition of T}
T\vcentcolon=(K\otimes t_\ell)\cdot(1\otimes\delta_{\mathrm{Diag}_2}\otimes1_{\ell-1})\in\mathcal{D}^\prime((\mathbb{R}^d)^{\ell+2})\,,
\end{align}
let
\begin{align}\label{Eq: good WF condition}
	\mathrm{WF}(T)\cap\{(x,y,z,\hat{x};\xi,0,0,\hat{\eta}_\ell)\in T^*(\mathbb{R}^{d})^{\ell+2}\setminus\{0\}
	\,|\,x\neq y\,,\,x\neq z\}=\emptyset\,,
\end{align}
where $\hat{x}=\underbrace{(x,\ldots,x)}_{\ell-\mathrm{times}}$ and similarly $\hat{\eta}_\ell=(\eta_1,\ldots,\eta_\ell)$. Then 
\begin{align}\label{Eq: sd is finite}
\mathrm{sd}_{\mathrm{Diag}_{k}}(K\circledast t_\ell)<+\infty\,.
\end{align}
\end{lemma}
\begin{proof}
First of all, on account of Theorem \ref{Thm: WF results}, we notice that $T$ as in Equation \eqref{Eq: definition of T} identifies an element of $\mathcal{D}^\prime((\mathbb{R}^d)^{\ell+2})$. Moreover, for any $f,h\in\mathcal{D}(\mathbb{R}^d)$,
\begin{align*}
(K\circledast t_\ell)(f\otimes h^{\otimes\ell})=T(f\otimes1_2\otimes h^{\otimes\ell})\,.
\end{align*}
Let $\chi\in\mathcal{D}(\mathbb{R}^d)$ and $\lambda\in(0,1)$, then
\begin{align*}
\lambda^\omega(K\circledast t_\ell)^\lambda_\alpha(\chi\otimes f\otimes h^{\otimes\ell})&=\lambda^\omega\int_{\mathbb{R}^d}T[f_x^\lambda\otimes1_2\otimes (h^{\otimes\ell})^\lambda_{\hat{x}}]\chi(x)\,dx\\
&=\lambda^{\omega+2d}\int_{\mathbb{R}^d}T[f_x^\lambda\otimes (g^{\otimes2})_{(x,x)}^\lambda\otimes (h^{\otimes\ell})^\lambda_{\hat{x}}]\chi(x)\,dx\\
&+\lambda^\omega\int_{\mathbb{R}^d}T[f_x^\lambda\otimes(1_2-\lambda^{2d}(g^{\otimes2})_{(x,x)}^\lambda)\otimes (h^{\otimes\ell})^\lambda_{\hat{x}}]\chi(x)\,dx\,,
\end{align*}
where $g\in\mathcal{D}(\mathbb{R}^d)$ is such that $\mathrm{supp}(g)\subset B(0,1)$ and $g|_{B(0,1/2)}=1$. After a few direct, algebraic manipulations, we can rewrite the above expression as
\begin{align}\label{Eq: appearence of lambda_0}
\nonumber
\lambda^\omega(K\circledast t_\ell)^\lambda_\alpha(\chi\otimes f\otimes h^{\otimes\ell})&=\underbrace{\lambda^{\omega+2d}T^\lambda_\alpha(\chi\otimes f\otimes g^{\otimes2}\otimes h^{\otimes\ell})}_A\\
\nonumber
&+\underbrace{\lambda^\omega\int_{\mathbb{R}^d}T[f_x^\lambda\otimes(1_2-\lambda_0^{2d}(g^{\otimes2})_{(x,x)}^{\lambda_0})\otimes (h^{\otimes\ell})^\lambda_{\hat{x}}]\chi(x)\,dx}_B\\
&+\underbrace{\lambda^\omega\int_{\mathbb{R}^d}\chi(x)\,dx\int_\lambda^{\lambda_0}T[f_x^\lambda\otimes(Dg^{\otimes2})_{(x,x)}^\mu)\otimes (h^{\otimes\ell})^\lambda_{\hat{x}}]\mu^{2d-1}\,d\mu}_C\,,
\end{align}
where $\lambda_0\in(\lambda,1)$ will be fixed later in the proof, while $(Dg)(z) := (z\cdot\nabla g)(z)$. \\
Focusing on $A$, it holds $A\to0$ for $\lambda\to0^+$ for any $\omega>\mathrm{sd}_{\mathrm{Diag}_{\ell+2}}(T)-2d<\infty$, where in the last inequality we exploited the bound $\mathrm{sd}_{\mathrm{Diag}_{\ell+2}}(T)\leq\mathrm{sd}_{\mathrm{Diag}_{2}}(K)+\mathrm{sd}_{\mathrm{Diag}_{\ell}}(t_\ell)+d$ (\textit{cf}. Remark \ref{Rmk: sd product estimate} and Equation \eqref{Eq: definition of T}).\\
Considering now $B$, it descends that
\begin{align*}
\lim_{\lambda\to0^+}\int_{\mathbb{R}^d}T[f_x^\lambda\otimes(1_2-\lambda_0^{2d}(g^{\otimes2})_{(x,x)}^\lambda)\otimes (h^{\otimes\ell})^\lambda_{\hat{x}}]\chi(x)\,dx=\int_{\mathbb{R}^d}T_{(x,\hat{x})}[1_2-\lambda_0^{2d}(g^{\otimes2})_{(x,x)}^{\lambda_0}]\chi(x)\,dx\,,
\end{align*}
where $T_{(x,\hat{x})}\in\mathcal{E}^\prime(\mathbb{R}^d\times\mathbb{R}^d\setminus\{(x,\hat{x})\})$ is such that, for any $w\in\mathcal{E}(\mathbb{R}^d\times\mathbb{R}^d\setminus\{(x,\hat{x})\})$, $T_{(x,\hat{x})}(w)\vcentcolon=T(\delta_x\otimes w\otimes\delta_x^{\otimes\ell})$. Well-definiteness of $T_{(x,\hat{x})}$ as a distribution is guaranteed by the hypothesis of Equation \eqref{Eq: good WF condition} together with Theorem \ref{Thm: WF results}. As a consequence, we see that, for any value of $\lambda_0$,  $B\to0$ for $\lambda\to0^+$ whenever $\omega>0$.\\
Eventually, we focus on $C$. As a starting point, we recall that, if $\omega>\mathrm{sd}_{\mathrm{Diag}_{\ell+2}}(T)$, then for any compact set $\mathfrak{K}\subset(\mathbb{R}^d)^{\ell+3}$ there exists a polynomial function $P$ of degree $p\in\mathbb{N}$ such that
\begin{align}\label{Eq: bound for sd with polynomial}
|\lambda^\omega T^\lambda_\alpha(\chi\otimes f)|\leq\sup_{\mathfrak{K}}|P(\partial)(\chi\otimes f)|\,,
\end{align}
for any $\chi\otimes f\in\mathcal{D}((\mathbb{R}^d)^{\ell+3})$ such that $\mathrm{supp}(\chi\otimes f)\subset\mathfrak{K}$. Focusing on $C$, it holds that
\begin{align*}
C=\lambda^\omega\int_\lambda^{\lambda_0}T^\mu_\alpha(\chi\otimes f^{\frac{\lambda}{\mu}}\otimes(Dg)^{\otimes2}\otimes(h^{\otimes\ell})^{\frac{\lambda}{\mu}})\mu^{2d-1}d\mu\,.
\end{align*}
Since $0<\lambda/\mu<1$, we have
\begin{align}\label{Eq: uniformly compact support}
\mathrm{supp}(\chi\otimes f^{\frac{\lambda}{\mu}}\otimes(Dg)^{\otimes2}\otimes(h^{\otimes\ell})^{\frac{\lambda}{\mu}})\subset\mathbf{K}\subset(\mathbb{R}^d)^{\ell+3}\,,
\end{align}
with $\mathbf{K}$ compact. To conclude the proof, it suffices to prove that the compact set $\mathbf{K}$ is independent of $\chi$. As a matter of fact, if this is the case, then we can apply uniformly the bound of Equation \eqref{Eq: bound for sd with polynomial}. This yields that, for any $\varepsilon>0$ and uniformly in $\lambda$,
\begin{align*}
|C|\leq\lambda^\omega\int^{\lambda_0}_\lambda\mu^{-\mathrm{sd}_{\mathrm{Diag}_{\ell+2}}(T)-\varepsilon}\mu^{2d+p}\lambda^{-2d-p}\mu^{2d-1}\,d\mu=\mathcal{O}(\lambda^{\omega-p-2d})+\mathcal{O}(\lambda^{\omega-\mathrm{sd}_{\mathrm{Diag}_{\ell+2}}(T)-\varepsilon+2d})\,.
\end{align*}
It follows that $C\to0$ for $\lambda\to0^+$ and for $\omega>\max\{2d+p,\mathrm{sd}_{\mathrm{Diag}_{\ell+2}}(T)-2d\}$. As a consequence, recollecting the estimates for $A$, $B$ and $C$, we conclude that 
\begin{align*}
\lim_{\lambda\to0^+}\lambda^\omega(K\circledast t_\ell)^\lambda_\alpha(\chi\otimes f\otimes h^{\otimes\ell})=0\,,\qquad\forall\,\omega>\max\{2d+p,\mathrm{sd}_{\mathrm{Diag}_{\ell+2}}(T)-2d\}<\infty\,,
\end{align*}
\emph{i.e.}, $\mathrm{sd}(K\circledast t_\ell)<\infty$.

To conclude, we focus on the compact set $\mathbf{K}$ introduced in Equation \eqref{Eq: uniformly compact support} and we show independence from choice of $\chi$. This is a consequence of the hypothesis of $\mathcal{K}\vcentcolon=\operatorname{pr}_2(\mathrm{supp}(K))\cap\operatorname{pr}_1(\mathrm{supp}(t_\ell))$ being compact, since the only non vanishing contribution to $C$ comes from the portion of $\mathrm{supp}(\chi)$ contained in $\mathcal{K}$. More precisely, consider $C$ at the level of integral kernels
\begin{align*}
C=\lambda^\omega\int_{\lambda}^{\lambda_0}\int_{(\mathbb{R}^d)^{\ell+3}}K(x+\lambda z,y)t_{\ell}(y, \hat{x}+\lambda\hat{z}_\ell)f(z)h^{\otimes\ell}(\hat{z}_\ell)(Dg)^{\otimes2}\big(\frac{y-x}{\mu}\big)\chi(x)\mu^{2d-1}\,dx\,dy\,dz\,d\hat{z}_\ell\,d\mu\,.
\end{align*}
We recall that $Dg$ is supported in an annulus whose external radius is $\mu$. If we consider the case $\mathrm{supp}(\chi)\cap\mathcal{K}=\emptyset$, then $d(\mathrm{supp}(\chi),\mathcal{K})\vcentcolon=R>0$, where $d(\mathrm{supp}(\chi),\mathcal{K})$ denotes the Euclidean distance between the compact sets. As a consequence, it suffices to choose $\lambda_0$ small enough in Equation \eqref{Eq: appearence of lambda_0} so that the above contribution vanishes for any $\mu\in(\lambda,\lambda_0)$. This concludes the proof. 
\end{proof}

The above result has been formulated for the $\circledast$ operation between distributions. Nonetheless, it can be derived also for the composition $\circ$ of distributions introduced in item \textit{4}. of Theorem \ref{Thm: WF results}.

\begin{corollary}\label{Cor: sd of composition}
	Let $K_1, K_2\in\mathcal{D}'(\mathbb{R}^d\times\mathbb{R}^d)$ be such that:
	\begin{enumerate}[(i)]
		\item
		$\mathrm{sd}_{\mathrm{Diag}_2}(K_1)<\infty$, $\mathrm{sd}_{\mathrm{Diag}_2}(K_2)<\infty$;
		\item
		$\operatorname{WF}^\prime_2(K_1)\cap\operatorname{WF}_1(K_2)=\emptyset$;
		\item
		$\operatorname{pr}_2(\mathrm{supp}(K_1))\cap\operatorname{pr}_1(\mathrm{supp}(K_2))$ is compact;
		\item $\mathrm{WF}(T)\cap\{(x,y,z,x;\xi,0,0,\eta)\in T^*\mathbb(\mathbb{R}^{d}\times\mathbb{R}^d)^2\setminus\{0\}\,|\,x\neq y\,,\,x\neq z\,\}=\emptyset\,,$ with $T\vcentcolon=(K_1\otimes K_2)\cdot(1\otimes\delta_{\mathrm{Diag}_2}\otimes1)$
	\end{enumerate}
	Then $K_1\circ K_2\in\mathcal{D}'(\mathbb{R}^d\times\mathbb{R}^d)$ and we have
	\begin{align}\label{Eq: sd of composition}
		\operatorname{sd}_{\mathrm{Diag}_2}(K_1\circ K_2)<\infty\,.
	\end{align}
\end{corollary}
\begin{proof}
	Well-definiteness of $K_1\circ K_2$ is a consequence of item \textit{4}. of Theorem \ref{Thm: WF results}.
	To prove the estimate \eqref{Eq: sd of composition} it is enough to observe that, for any $f_1,f_2\in\mathcal{D}(\mathbb{R}^d)$,
	\begin{align*}
		(K_1\circ K_2)(f_1\otimes f_2):=[(K_1\otimes K_2)\cdot(1_d\otimes\delta_{\mathrm{Diag}_2}\otimes 1_d)](f_1\otimes 1_{2d}\otimes f_2)\,.
	\end{align*}
	The statement descends applying Lemma \ref{Lemma: finite sd of convolution} together with the identifications
	\begin{align*}
		\ell=2\,,\qquad
		T=(K_1\otimes K_2)\cdot(1_d\otimes\delta_{\mathrm{Diag}_2}\otimes 1_d)\,.
	\end{align*}
\end{proof}

\begin{remark}[Parabolic Case]
		The results of Theorem \ref{Thm: extension with scaling degree} and Lemma \ref{Lemma: finite sd of convolution} remain valid with minor modifications also in this scenario. Notice that, on account of the \emph{weighted scaling}, the effective dimension of $M=\mathbb{R}\times\Sigma$, with $\dim(\Sigma)=d-1$, is $d+1$.
	A systematic analysis of the relation between the notion of weighted scaling degree and the standard one will be investigated elsewhere.
\end{remark}

%\begin{remark}
%	With minor adjustments lemma \ref{Cor: sd of convolution} can be extended to the case of scaling degree with respect to submanifolds -- \textit{cf.} remark \ref{Rmk: sd for submanifolds}.
%	We shall state the result without proof.
%	Let $K\in\mathcal{D}'(M^{k+\ell})$ and $t\in\mathcal{E}'(M^\ell)$ be such that $\operatorname{WF}(K)\subseteq\operatorname{WF}(\delta_{\mathrm{Diag}_{k+\ell}})$.
%	Then $K\circledast t\in\mathcal{D}'(M^k)$ is well-defined and
%	\begin{align}\label{Eq: sd convolution for submanifolds}
%		\operatorname{sd}_{\mathrm{Diag}_k}(K\circledast t)
%		\leq\max\{0\,,\,
%		\operatorname{sd}_{\mathrm{Diag}_{k+\ell}}(K)
%		+\operatorname{sd}_{\mathrm{Diag}_\ell}(t)
%		-kd\}
%	\end{align}
%\end{remark}

\end{document}